\documentclass[11pt,letterpaper]{article}

\usepackage{geometry}
\renewenvironment{abstract}
 {\small
  \begin{center}
  \bfseries \abstractname\vspace{-.5em}\vspace{0pt}
  \end{center}
  \list{}{
    \setlength{\leftmargin}{.63cm}%
    \setlength{\rightmargin}{\leftmargin}%
  }%
  \item\relax}
 {\endlist}
 
\usepackage[skip=.35\baselineskip plus 1pt]{parskip}

\usepackage[authoryear]{natbib}
\usepackage{titlesec}
\usepackage{amsmath}
\usepackage{amsfonts}
 \usepackage{amssymb}
\usepackage{amsthm}
\usepackage{amsbsy}
\usepackage{graphicx}
\usepackage{verbatim}
\usepackage{url}

\usepackage[T1]{fontenc}
\usepackage[tt=false, type1=true]{libertine}
\usepackage[varqu]{zi4}
\usepackage[libertine]{newtxmath}

\usepackage{bbm}
\usepackage{multirow}
\usepackage{xspace}
\usepackage{footnote}
\usepackage{hyperref}
\usepackage[svgnames]{xcolor}

\usepackage[capitalise,nameinlink,nosort]{cleveref}
\hypersetup{colorlinks={true},linkcolor={DarkBlue},citecolor=[named]{DarkGreen},urlcolor=black}

\AddToHook{env/claim/begin}{\crefalias{theorem}{claim}}
\AddToHook{env/corollary/begin}{\crefalias{theorem}{corollary}}
\AddToHook{env/proposition/begin}{\crefalias{theorem}{proposition}}
\AddToHook{env/observation/begin}{\crefalias{theorem}{observation}}
\AddToHook{env/definition/begin}{\crefalias{theorem}{definition}}
\AddToHook{env/remark/begin}{\crefalias{theorem}{remark}}
\AddToHook{env/lemma/begin}{\crefalias{theorem}{lemma}}

\crefname{claim}{Claim}{Claims}
\crefname{lemma}{Lemma}{Lemmata}

\usepackage{bm}
\usepackage{tikz}
\usepackage{makecell}

\usepackage{titling}
\usepackage[noblocks]{authblk}

\newtheorem{inflem}{Informal Lemma}
\newtheorem{theorem}{Theorem}[section]
\newtheorem{lemma}[theorem]{Lemma}
\newtheorem{claim}[theorem]{Claim}
\newtheorem{corollary}[theorem]{Corollary}

\theoremstyle{definition}
\newtheorem{example}{Example}
\newtheorem{definition}[theorem]{Definition}

\theoremstyle{remark}

\usepackage[ruled,vlined]{algorithm2e} 

\SetAlFnt{\small}
\SetAlCapFnt{\small}
\SetAlCapNameFnt{\small}
\SetAlCapHSkip{0pt}
\IncMargin{-\parindent}
\newcommand\mycommfont[1]{\footnotesize\ttfamily\textcolor{gray}{#1}}
\SetCommentSty{mycommfont}

\SetKwInput{KwData}{Input}
\SetKwInput{KwResult}{Output}

\usepackage{setspace}

\usepackage{float}

\usepackage{enumitem}

\newlist{threecases}{enumerate}{1}
\setlist[threecases]{label={(\alph*)}}
\crefname{threecasesi}{Case}{Cases}

\newlist{properties}{enumerate}{1}
\setlist[properties]{label={(\alph*)}}
\crefname{propertiesi}{Property}{Properties}

\newlist{pproperties}{enumerate}{1}
\setlist[pproperties]{label={(F.\arabic*)}}
\crefname{ppropertiesi}{Property}{Properties}

\newlist{levels}{enumerate}{1}
\setlist[levels]{label={(L.\arabic*)}}
\crefname{levelsi}{Level}{Levels}

\usepackage{tcolorbox}

\xspaceaddexceptions{]\}}

\def \X {\mathbf{X}\xspace}

\newcommand{\Gr}{G_\mathrm{r}\xspace}
\newcommand{\Ge}{G_\mathrm{e}\xspace}
\newcommand{\G}{G\xspace}
\newcommand{\numagents}{7\xspace}
\newcommand{\Gde}{G_\mathrm{de}\xspace}
\newcommand{\potential}{\mathcal{F}\xspace}

\newcommand{\PP}{\mathcal{P}\xspace}

\newcommand{\CR}{\textnormal{\textsc{CycleResolution}}\xspace}
\newcommand{\ACR}{\textnormal{\textsc{AllCyclesResolution}}\xspace}
\newcommand{\PR}{\textnormal{\textsc{PathResolution}}\xspace}
\newcommand{\PRPA}{\textnormal{\textsc{PathResolution$^*$}}\xspace}
\newcommand{\PPA}{\textnormal{\textsc{3PA}}\xspace}
\newcommand{\SP}{\textnormal{\textsc{SingletonPool}}\xspace}
\newcommand{\UCG}{\textnormal{\textsc{UncontestedCri\-ti\-cal}}\xspace}
\newcommand{\SRR}{\textnormal{\textsc{SingleRoundRobin}}\xspace}
\newcommand{\MGA}{\textnormal{\textsc{MultigraphAllocate}}\xspace}
\newcommand{\FAA}{\textnormal{\textsc{FewAgentsAllocate}}\xspace}
\newcommand{\TVA}{\textnormal{\textsc{ThreeValuesAllocate}}\xspace}

\newcommand{\iso}{secluded\xspace}
\newcommand{\Iso}{Secluded\xspace}

\newcommand{\regplus}{reduced$+$ envy graph\xspace}
\newcommand{\REGplus}{Reduced$+$ Envy Graph\xspace}

\newcommand{\enhplus}{enhanced$+$ envy graph\xspace}
\newcommand{\ENHplus}{Enhanced$+$ Envy Graph\xspace}

\newcommand{\Grplus}{G_\mathrm{r}^+\xspace}
\newcommand{\Geplus}{G_\mathrm{e}^+\xspace}

\newcommand{\PPAplusplus}{\textnormal{\textsc{3PA$^{++}$}}\xspace}

\newcommand{\ECE}{\textnormal{\textsc{EnvyCycleElimination}}\xspace}

\DeclareMathOperator*{\argmax}{arg\,max}

\newcounter{note}[section]

\renewcommand{\paragraph}[1]{\medskip\noindent\textbf{#1.\;}}

\begin{document}

\title{Pushing the Frontier on Approximate EFX Allocations}

\author[1,3]{Georgios Amanatidis}
\author[2]{Aris Filos-Ratsikas}
\author[1,3]{Alkmini Sgouritsa}

\affil[1]{Department of Informatics, Athens University of Economics and Business; Athens, Greece}

\affil[2]{School of Informatics, University of Edinburgh; Edinburgh, UK}

\affil[3]{Archimedes / Athena RC; Athens, Greece}

\predate{}
\postdate{}
\date{}

\maketitle

\begin{abstract}
\noindent We study the problem of allocating a set of indivisible goods to a set of agents with additive valuation functions, aiming to achieve approximate envy-freeness up to any good ($\alpha$-EFX). The state-of-the-art results on the problem include that (exact) EFX allocations exist when (a) there are at most three agents, or (b) the agents' valuation functions can take at most two values, or (c) the agents' valuation functions can be represented via a graph. For $\alpha$-EFX, it is known that a $0.618$-EFX allocation exists for any number of agents with additive valuation functions. 
In this paper, we show that $2/3$-EFX allocations exist when (a) there are at most \emph{seven agents}, (b) the agents' valuation functions can take at most \emph{three values}, or (c) the agents' valuation functions can be represented via a \emph{multigraph}. Our results can be interpreted in two ways. First, by relaxing the notion of EFX to $2/3$-EFX, we obtain existence results for strict generalizations of the settings for which exact EFX allocations are known to exist. Secondly, by imposing restrictions on the setting, we manage to beat the barrier of $0.618$ and achieve an approximation guarantee of $2/3$. Therefore, our results  push the \emph{frontier} of existence and computation of approximate EFX allocations, and provide insights into the challenges of settling the existence of exact EFX allocations. 
\end{abstract}

\section{Introduction}
\label{sec:intro}

The existence of EFX allocations is arguably the most important open problem in the area of computational fair division with indivisible goods. In this setting, a set of agents have values over a set of indivisible goods, and the goal is to allocate the goods to the agents in a way that is perceived as \emph{fair} by everyone. One of the most well-established notions of fairness is \emph{envy-freeness}, introduced in the late 1950s by \citet{gamow1958puzzle} (see also \citep{foley1966resource,varian1973equity}) in the context of \emph{divisible} resources;
this notion stipulates that no agent would prefer another agent's allocation to her own. For indivisible goods, it is not hard to see that envy-free allocations may not be possible. Indeed, consider the case where there is a single good which is positively valued by multiple agents; some agent will receive the good and the others will be envious.

Motivated by this obvious impossibility, the literature has defined \emph{relaxed} fairness notions, which are appropriate for indivisible goods allocation. \citet{budish2011combinatorial} first defined the notion of \emph{envy-freeness up to one good (EF1)}, which deems an allocation fair if the envy of an agent is eliminated after the removal of \emph{some} good from the bundle of another envied agent. This notion was implicitly introduced earlier by \citet{lipton2004approximately}, who showed that EF1 allocations for monotone valuation functions always exist and can be computed in polynomial time via simple algorithms. On the minus side, in certain applications EF1 might be unsatisfactory, as it might require the removal of a very valuable good to restore envy-freeness. To address this shortcoming, \citet{GourvesMT14} and \citet{caragiannis2019unreasonable} introduced the notion of \emph{envy-freeness up to any good (EFX)}, which stipulates that the envy is eliminated even if the least valuable good, from the envious agent's perspective, is removed from the envied agent's bundle. 

Contrary to the case of EF1 allocations, the existence of EFX allocations is much more intricate. In fact, as we mentioned above, this is one of the major open problems in the area of computational fair division, currently carrying great momentum and being met with intensive efforts from the research community. The associated research has adopted a systematic approach to tackling this evidently very challenging question, by first obtaining existence results for special cases of the problem, developing a deeper understanding of its intricacies, and ultimately aiming to synthesize these ideas into an answer to the main problem. Among those results, three of the most notable ones are that EFX allocations exist when 
\begin{itemize}[leftmargin=20pt]
\item[-] there are at most three agents \citep{chaudhury2020efx},
\item[-] each agent's value for each good can be one of \emph{two} numbers $a$ or $b$ \citep{amanatidis2021maximum},
\item[-] the agents' valuation functions can be represented by a \emph{graph}, with edges corresponding to goods and nodes corresponding to agents \citep{christodoulou2023fair}. 
\end{itemize}
Other interesting restrictions that have been studied in the literature include agents with values that induce the same ordering over goods \citep{PR18} and leaving some goods unallocated \citep{caragiannis2019charity,CKMS20,aaai/BergerCFF22}.

A related line of work has studied \emph{approximations} to the EFX notion, referred to as $\alpha$-EFX. An allocation is $\alpha$-EFX if after the removal of any good from the envied agent's bundle, the envy is bounded by a multiplicative factor of $\alpha$. A $1/2$-EFX allocation is achievable via slight modifications to techniques that achieve EF1, e.g., see \citep{ChanCLW19}. The state of the art for approximate EFX allocations is a $(\!\sqrt{5}-1)/2 \approx 0.618$ approximation due to \citet{ANM2019}, 
which in fact seems to be the limit of the techniques that the literature has developed for the problem; see also \citet{farhadi2021almost}. \citet{markakis2023improved} were able to beat the $0.618$ barrier and achieve $2/3$-EFX allocations for a  restricted case where all of the $n$ agents agree on which $n$ goods are the most valuable. Whether $2/3$-EFX allocations can be achieved in general is a very intriguing open problem. If sufficiently many goods are left unallocated, it has been shown that near optimal approximations to EFX can be achieved \citep{ChaudhuryGMMM21}.

The state-of-the-art results in the literature can thus be seen as lying on a certain kind of \emph{frontier}: 
indeed, a certain set of parameters (e.g., the number of agents, the type of values of the agents, the approximation factor), can be seen as a point on a search space, with those points for which we have obtained existence results constituting the frontier of our current understanding of the problem. The ultimate goal is to move towards the point corresponding to (exact) EFX for any number of agents and without restrictions on the values. Given the challenging nature of this endeavor, however, it seems sensible to try to ``push'' this frontier towards that goal.

\subsection{Our Results and Techniques}
In this work, we prove that $2/3$-EFX allocations  exist and can be efficiently computed for agents with additive valuation functions in three important cases, namely when:
\begin{itemize}[leftmargin=20pt]
    \item[-] There are at most \textit{seven} agents (\cref{thm:7_agents}).\smallskip
    \item[-] Each agent's value for each good can be one of \emph{three} non-negative numbers $a$, $b$, or $c$ (\cref{thm:3-value-instances}). \smallskip
    \item[-] The agents' values can be represented by a \emph{multigraph}, with edges corresponding to goods and nodes corresponding to agents (\cref{thm:multigraph}). Here an agent has nonzero value for a good only if this good is incident to her. This setting generalizes the setting studied recently by \citet{christodoulou2023fair}.
\end{itemize}
We make progress in all three settings through the same algorithmic framework, although each one requires its own modifications. While all of these are nontrivial, the most intricate case is, somewhat surprisingly, the case of three values.
We next present an overview of our techniques.

\paragraph{Property-Preserving Partial Allocations}  Our approach is based on the following general principle: obtain a \emph{partial} allocation $\X$ of goods to agents that satisfies a certain set of properties. Then this allocation can be transformed into a complete allocation that is $2/3$-EFX. 

To be more precise, all of the known algorithms for $\alpha$-EFX in the literature \citep{ANM2019,markakis2023improved,farhadi2021almost} start by producing a partial allocation $\X$ using only a subset of the goods. How this partial allocation is obtained may differ between different algorithms, but they all serve the same purpose: once $\X$ is obtained, then one can run the Envy Cycle Elimination algorithm of \citet{lipton2004approximately}, with initial input $\X$, to produce an allocation that is approximately EFX. For this to be possible, $\X$ has to satisfy certain properties, mainly that (a) it is $\alpha$-EFX, for the approximation factor $\alpha$ that we are aiming to prove, and (b) none of the agents consider any of the goods that are left unallocated \emph{too valuable}. Formally, an unallocated good is ``too valuable'' for an agent $i$ if her value for that good is {more than $\beta$ times} her value for her allocated bundle in the partial allocation $\X$. 
These goods lie at the heart of our approach and we refer to them as \emph{critical goods}; see \Cref{def:critical}. This is captured by the following lemma, stated informally here (its formal statement, restricted to $2/3$-EFX, is \cref{lem:partialefx-to-efx-extension} in \Cref{sec:key-lemmas}). 

\begin{inflem}[\citet{markakis2023improved}]\label{inflem:partial-to-full}
If a partial allocation $\X$ is $\alpha$-EFX and does not induce any critical goods (defined via a parameter $\beta$), then it can be transformed into a complete $\min\{\alpha,\frac{1}{\beta+1}\}$-EFX allocation. 
\end{inflem}

Therefore, the value of $\beta$ that makes a good critical depends on the approximation factor $\alpha$ that we are aiming to prove existence for. For $\alpha=0.618$ (achieved by \citet{ANM2019}), it is also the case that $\beta=0.618$. 
From a technical perspective, the ``balance'' of these terms makes the construction of a partial allocation $\X$ that does not induce any critical goods achievable via relatively simple algorithms. 
In fact, said algorithms also guarantee that the cardinality of each bundle in $\X$ is at most $2$. In the context of (approximate) EFX, bundles of cardinality at most $2$ are much more manageable. For larger values of $\alpha$, however, $\alpha \neq \beta$ and we will have a natural imbalance. For $\alpha=2/3$ in particular, we have $\beta=1/2$. In this case, it can be shown that even for the cases that we consider, it is not possible for $\X$ to be both $2/3$-EFX and induce no critical goods, unless agents receive bundles of cardinality $3$ or more. We demonstrate that concretely via the following example with only three agents and three possible values. We generalize this example to provide similar limitations for higher cardinality bundles, where there are still only three possible values.

\begin{example}\label{ex:non-existence-of-23efx-with-bundles-of-size-two}
Consider an instance with three agents $1,2,3$ and five goods $h_1,h_2,g_1,g_2$ and $g_3$. The values of the agents for the goods are given by the following table.
\begin{table}[ht]
\begin{center}
\begin{tabular}{|c|c|c|c|c|c|}
\hline
          & $h_1$ & $h_2$ & $g_1$   & $g_2$   & $g_3$    \\ \hline
agent $1$ & $1$   & $1$   & $3/5$   & $1/100$ & $1/100$  \\ \hline
agent $2$ & $1$   & $1$   & $1/100$ & $3/5$   & $1/100$  \\ \hline
agent $3$ & $1$   & $1$   & $1/100$ & $1/100$ & $3/5$    \\ \hline
\end{tabular}
\label{tab:example}
\end{center}
\end{table} 
Assume that we have a partial allocation $\X$ where the bundle $X_i$ of each agent $i \in \{1,2,3\}$ \textit{contains at most two goods}. If either $h_1$ or $h_2$ is unallocated, they are critical for each of the agents. If both are allocated, assume without loss of generality that $h_1$ is allocated to agent $1$. For $\X$ to be $2/3$-EFX, $h_2$ cannot be allocated to agent $1$ as well, so it must be allocated to some other agent, say agent $2$. This means that neither of agents $1$ and $2$ can have another good in her bundle, as otherwise agent $3$ (who does not get $g_1$ or $g_2$) would envy that agent by a factor greater than $2/3$. Since agent $3$'s bundle contains at most two goods, one of $g_1$, $g_2$, and $g_3$ must be unallocated. These would be critical for agents $1$, $2$ and $3$ respectively, so in the end there must be some critical good with respect to $\X$. 

It is not hard to extend this example to show that, for any constant $k\in \mathbb{N}$, it is not possible for $\X$ to be $2/3$-EFX, induce no critical goods, and only contain bundles of cardinality $k$ or less. More specifically, for any $k\in \mathbb{N}$, consider $k+1$ agents and a set $M$ of $2k+1$ goods: $M=\{h_1, \ldots, h_k, g_1, \ldots , g_{k+1}\}$. Each good $h_i$ has value $1$ for all $k+1$ agents, and each good $g_i$ has value $3/5$ for agent $i$ and $1/100k$ for the rest of the agents. The arguments for $k=2$ are trivially generalized to any $k\in \mathbb{N}$.

For instance, adding an agent $4$ to the current instance, along with two new goods,  $h_3$ (high-valued for everyone) 
\begin{table}[ht]
\begin{center}
\begin{tabular}{|c|c|c|c|c|c|c|c|}
\hline
          & $h_1$ & $h_2$ & $h_3$   & $g_1$   & $g_2$   & $g_3$   & $g_4$   \\ \hline
agent $1$ & $1$   & $1$   & $1$     & $3/5$   & $1/100$ & $1/100$ & $1/100$ \\ \hline
agent $2$ & $1$   & $1$   & $1$     & $1/100$ & $3/5$   & $1/100$ & $1/100$ \\ \hline
agent $3$ & $1$   & $1$   & $1$     & $1/100$ & $1/100$ & $3/5$   & $1/100$ \\ \hline
agent $3$ & $1$   & $1$   & $1$     & $1/100$ & $1/100$ & $1/100$ & $3/5$  \\ \hline
\end{tabular}
\label{tab:example-k}
\end{center}
\end{table} 
and $g_4$ (medium value for agent $4$ and  low for everyone else), according to the updated table, immediately extends the above argument for $k=3$.
\renewcommand\qedsymbol{{\footnotesize $\blacksquare$}}\qed
\end{example}

Our goal will be to obtain a $2/3$-EFX partial allocation without critical goods \emph{in two stages}.  First, we devise a general algorithm called \textsc{Property-Preserving Partial Allocation}  (\PPA), which obtains a partial allocation $\X^1$ that satisfies a certain set of properties (see Algorithm \ref{alg:3PA}). One such property is that it is $2/3$-EFX, as required by \cref{inflem:partial-to-full} above. This partial allocation $\X^1$ will still have critical goods, but it will limit the number of such goods to at most \emph{one per agent}, and only for agents that have singleton bundles in $\X^1$. This, together with the other properties of $\X^1$, will prepare the ground for allocating the critical goods to the agents in a subsequent stage, resulting in a new partial allocation $\X^2$. This $\X^2$ will now satisfy the properties of \cref{inflem:partial-to-full}, and it will eventually be transformed into a complete $2/3$-EFX allocation. 

\paragraph{Swap steps and different types of envy graphs} The aforementioned \PPA algorithm is based on a series of steps that are executed in sequence according to a certain priority structure. Most of these are \emph{swap steps}, i.e., steps that enable certain agents to exchange (parts of) their bundles in the partial allocation $\X$ with certain unallocated goods. The priority is determined by the cardinality of the bundles in $\X$, as well as the value of the agents for those unallocated goods. Besides the swap steps, the algorithm also includes steps that are performed on several different types of envy graphs, associated with allocation $\X$. A (standard) envy graph for a partial allocation $\X$ is a graph in which the nodes correspond to agents, and an edge $(i,j)$ signifies that agent $i$ envies agent $j$. Envy graphs are very common in the design of algorithms for fair division of indivisible goods, starting with the Envy Cycle Elimination algorithm of \citet{lipton2004approximately}. We also consider two different types of graphs, namely \emph{reduced} graphs and \emph{enhanced} graphs. In the former, envy towards agents with singleton bundles is disregarded, unless it is high enough. In the latter, we add edges of \emph{near} envy, as long as the value of the target bundle is above a certain threshold. By exchanging bundles along cycles and paths in these graphs, we obtain the flexibility to deal with the inherently most challenging case of ``moving value'' to and from agents with singleton bundles. 
In the case of three values, we further refine our graphs, including edges that indicate ties in the values. This allows for more options on how the partial allocations evolve, but adds an extra layer of complexity in the analysis.

Throughout the execution of the \PPA algorithm, the value of an agent for her bundle may decrease several times, but this is done in a controlled way that allows us to allocate the critical goods in subsequent steps. We remark that since the steps of the \PPA algorithm (or its variant \PPAplusplus from \Cref{sec:refined3PA}) repeatedly allocate and de-allocate goods, their polynomial running time, or even their termination is far from obvious; to this end, we provide proofs based on carefully crafted potentials and appropriate combinatorial arguments.

\subsection{Further Discussion on Related Work}
The literature on the fair division of indivisible goods is extensive. Here we discuss the works most relevant to us, and refer the reader to the recent survey of \citet{amanatidis2023fair} for more details. 

The existence of EFX allocations for three agents with additive valuations by \citet{chaudhury2020efx} was a breakthrough which followed the existence result for two agents due to \citet{PR18}. This result was subsequently simplified (and generalized to somewhat larger valuation classes) by \citet{akrami2023efx}. Our first aforementioned result above shows that if we relax the requirement from EFX to $2/3$-EFX, we can guarantee its existence for up to \numagents agents. Our algorithm for this case also runs in polynomial time, contrary to the aforementioned algorithms for exact EFX and three agents. \citet{amanatidis2021maximum} were the first to study the existence of EFX on instances with two values ($2$-value instances). First, they showed that the Maximum Nash Welfare (the geometric mean of the agents' values for their allocated bundles) guarantees EFX in this setting, and then developed a polynomial-time algorithm to achieve EFX allocations. \citet{babaioff2021fair} consider a setting where the valuations need not be additive, but have binary marginals and prove the existence of EFX allocations in that case. The approaches for two values clearly fall short if one considers instances with three values. Our approach is fundamentally different from that of \citet{amanatidis2021maximum} and shows that existence is possible, at the expense of a loss in the approximation factor. Finally, the question of the existence of EFX allocations on graphs was introduced by \citet{christodoulou2023fair}; the authors point out that their techniques are insufficient when applied to instances on multigraphs. 
Using our framework, we show how to obtain $2/3$-EFX allocations on multigraphs in polynomial time, and in a relatively simple manner.

\subsubsection{Work Following the Conference Version} Since the publication of the conference version of this work in 2024 \cite{AFS24}, a number of related works have appeared in the literature, including works which explicitly build on our framework and the \PPA algorithm. \citet{PrakashMN25} proposed a modified version of the \PPA algorithm which achieves $2/3$-EFX allocations on instances with $4$ types of agents; note that for three types of agents, (exact) EFX allocations were shown to exist earlier by \citet{hv2025efx}. For values represented via a multigraph, \citet{kaviani2025improved} showed an improved approximation of $1/\sqrt{2}$, using different techniques to the ones we introduced in this paper, and \citet{KSS24} proved the existence of (exact) EFX for restricted additive valuations. Very recently, \citet{Afshinmehr2026EFXMultigraphs} and  \citet{Christodoulou2026MultigraphEFX} independently showed the existence of (exact) EFX allocations on multigraphs for additive and cancelable valuations, respectively. 
Other works that study approximate EFX allocations on specific multigraph structures and related settings include \cite{bhaskar2024efx,bhaskar2025extending,afshinmehr2025efx,lianeas2026efx,kakatelis2026almost}. Specifically, for hypergraphs with girth at least 3, \citet{kakatelis2026almost} recovered the $1/\sqrt{2}$-EFX result of \cite{KSS24} via a variation of the \PPA algorithm, which they further extended to show $2/3$-EFX when the edge multiplicity is at most two. 

Quite recently, \citet{filos2026approximate} studied approximate envy-free allocations up to any $k$ goods ($\alpha$-EFkX), a notion introduced by \citet{AkramiRS2022EF2X}. Their main result is obtained via a generalization of the \PPA algorithm (coined ``Generalized'' 3PA (G3PA)), which in polynomial time achieves $\frac{k+1}{k+2}$-EFX allocations for any number of agents. In particular, their result implies a $3/4$-EF2X existence result for any number of agents; (exact) EF2X allocations were known to exist only for the case of up to $4$ agents, due to \citet{ashuri2025ef2x}. From a technical perspective, the case of $k>1$ allows for some crucial flexibility when allocating the critical goods: agents may have more than one critical good in the pool, but once an agent receives such a good, the remaining goods may no longer be critical, in relation to the agent's newly enhanced bundle.

\section{Preliminaries}
\label{sec:preliminaries}

Let $N$ be the set of agents and let $M$ be the set of indivisible goods. Let $n=|N|$ and $m=|M|$. Our goal is to assign the goods to the agents in a fair way. This assignment is captured by the notion of an \emph{allocation}, which can be \emph{partial} or complete. 

\begin{definition}[(Partial) Allocation]
A \emph{partial allocation} $\X=(X_1, \ldots, X_n)$ is an ordered tuple of disjoint subsets of $M$,  i.e.,
(i) for every agent $i$, $X_i \subseteq M$, and
(ii) for any two distinct agents $i$ and $j$, it holds that $X_i \cap X_j = \emptyset$.  
\end{definition}

A partial allocation $\X$ corresponds to an assignment of (some of the) goods to the agents; we call $X_i$ the \emph{bundle} of agent $i$. An \emph{allocation} is a partial allocation for which $\bigcup_{i \in N} X_i = M$. We also say that a partial allocation \emph{is of size at most $k$} if $\max_{i \in N} |X_i| \leq k$.

It will also be useful to refer to the set of unallocated goods as ``the pool''. 
\begin{definition}[The Pool]
Given a partial allocation $\X$, we refer to $M \setminus \cup_{i \in N} X_i$ as \emph{the pool}, and we denote it by $\PP(\X)$. 
\end{definition}

\subsection{Valuation Functions and Fairness Notions}

Next we define the value of an agent for a bundle of goods.
\begin{definition}[Agent Valuation Functions]
Given a function $v_i: M \to \mathbb{R}_{\ge 0}$, where $v_i(g)$ represents the value of agent $i$ for good $g$, we can define the \emph{value} of an agent $i \in N$ for
a bundle of goods $X \subseteq M$ by extending $v_i$ to an \emph{additive} function, i.e., $v_i(X) = \sum_{g \in X}v_i(g)$. 
\end{definition}

In this work, we only study additive instances of the problem, i.e., instances where all agents have additive valuation functions.
We will also be interested in two special cases of additive instances, namely \emph{multigraph value} and \emph{$3$-value} instances. As we discussed in the Introduction, these strictly generalize the classes of instances for which the state-of-the-art results on the existence of EFX have been established.

\begin{definition}[Multigraph Value Instances]
We say that an instance of the problem is a \emph{multigraph value instance}, if for any agent $i \in N$, the function $v_i$ is additive, and there is a labeled multigraph $H = (V, E)$ with its vertices labeled via a bijection to $N$ and its edges labeled via a bijection to $M$, such that for any $i \in N$ and for any good $g \in M$, it holds that $v_i(g) \neq 0$ only if $g$ is the label of an edge incident to the vertex labeled $i$. 
\end{definition}

Intuitively, we can see a multigraph value instance as a multigraph, where the vertices are agents, the edges are the goods, and any agent has value $0$ for any edge not incident to her.

\begin{definition}[$3$-Value Instances]
We say that an instance of the problem is a \emph{$3$-value instance}, if for any agent $i \in N$ and for any good $g \in M$, the function $v_i$ is additive and it holds that $v_i(g) \in \{a,b,c\}$, {for some $a, b, c \in \mathbb{R}_{\ge 0}$ with $a \ge b \ge c$.} 
\end{definition}

Clearly, in instances where $a = b = c = 0$ any allocation is EFX, so we only consider instances where $a >0$. In fact, without loss of generality, we assume that $a=1$, as the {other two values $b$ and $c$} can be scaled appropriately.

We are now ready to define our main fairness notion, that of \emph{approximate EFX}.

\begin{definition}[Envy, Envy-Freeness, $\alpha$-EFX]
\label{def:ef-efx}
Given a partial allocation $\X$, a constant $\alpha>0$, and two agents $i, j \in N$, we will say that 
\begin{itemize}[leftmargin=20pt]
    \item[-] agent $i$ \emph{does not $\alpha$-envy} or, equivalently, \emph{is $\alpha$-envy-free towards} agent $j$, if $v_i(X_i) \geq \alpha \cdot v_i(X_j)$;
    \item [-] agent $i$ \emph{does not $\alpha$-envy} agent $j$ \emph{up to any good} or, equivalently, \emph{is $\alpha$-EFX towards} agent $j$, if {$X_j = \emptyset$ or if} for any good $g \in X_j$, it holds that $v_i(X_i) \geq \alpha \cdot v_i(X_j \setminus \{g\})$.
\end{itemize}
A partial allocation $\X$ is $\alpha$-envy-free (respectively $\alpha$-EFX) if, given $\X$, every agent $i \in N$ is $\alpha$-envy-free (respectively $\alpha$-EFX) towards any other agent $j \in N$. When $\alpha = 1$, we drop the prefixes; e.g., we write EFX instead of $1$-EFX. 
\end{definition}

Often, we need to compare the value  $v_i(X_i)$ an agent has for her bundle with  the value $v_i(S)$ she has for a set $S$ which is not the bundle of any other agent. Then, if $v_i(S) \ge \alpha \cdot  v_i(X_i)$ {(respectively $v_i(S) > \alpha \cdot  v_i(X_i)$)}, instead of using envy terminology, we say that 
agent $i$ \textit{prefers} $S$ \textit{by a factor of at least} $\alpha$ {(respectively \textit{greater than} $\alpha$)}; moreover, {we drop the braces of $S$ in such phrasing when $S$ is a singleton}.

The goal is to find a (complete) allocation which is $\alpha$-EFX. Note that for $m\le n$, it is trivial to produce an exact EFX allocation by giving at most one good to each agent arbitrarily. So, in what follows we assume that $m > n$. This will allow us to consider partial $\alpha$-EFX allocations where no bundle is empty throughout the course of our algorithms, thus avoiding trivial cases in several of our proofs.
We next define the notion of \emph{critical goods}, which lies at the heart of our approach.  

\begin{definition}[Critical Good]\label{def:critical}
Given a partial allocation $\X$, a good $g \in \PP(\X)$ is \emph{critical} for agent $i \in N$ if 
$i$ prefers $g$ by a factor {greater than} $1/2$, i.e., $v_i(g) > \frac{1}{2} v_i(X_i)$.
\end{definition}
As we have mentioned above (recall the related discussion in the Introduction), one could give a parameterized definition of criticality \citep[see, e.g.,][]{markakis2023improved}, but as we only care about a specific value here that is appropriate for our goal of achieving $2/3$-EFX allocations, we opted not to do so.

\subsection{Envy Graphs}

The notion of an envy graph is ubiquitous in the design of algorithms for variants and relaxations of envy-freeness. Besides the ``vanilla'' version, we define and use two additional envy-based graphs.

\begin{definition}[Envy Graph $G$]
Given a partial allocation $\X$, the \emph{envy graph on $\X$}, $G(\X) = (N, E(\X))$, is a directed graph whose vertex set is the set of agents. Its edge set contains a directed edge $(i,j)$ if and only if agent $i$ envies agent $j$, i.e., $E(\X) = \{(i, j) : v_i(X_i) < v_i(X_j)\}$. 
\end{definition}

\begin{definition}[Reduced Envy Graph $\Gr$]\label{def:Gr}
Given a partial allocation $\X$, the \emph{reduced envy graph on $\X$}, $\Gr(\X) = (N, E_\mathrm{r}(\X))$, is a subgraph of the envy graph $G(\X)$, where we remove all edges $(i, j)$, such that
$|X_i|>1$, $|X_j|=1$, and $i$ is $\frac{2}{3}$-envy-free towards $j$. 
That is $E_\mathrm{r}(\X) = E(\X) \setminus \{(i, j) : |X_i|>1, |X_j|=1, \text{ and } v_i(X_i) \ge \frac{2}{3} v_i(X_j)\}$.
\end{definition}

\begin{definition}[Enhanced Envy Graph $\Ge$]\label{def:Ge}
Given a partial allocation $\X$, the \emph{enhanced envy graph on $\X$}, $\Ge(\X) = (N, E_\mathrm{e}(\X))$, is a supergraph of the reduced envy graph $\Gr(\X)$, where we add all edges $(i, s)$, such that $s$ is a source (i.e., a vertex of in-degree 0) in $\Gr(\X)$, 
$|X_i|=1$, $|X_s|>1$, and $v_i(X_s) \geq \frac{2}{3} v_i(X_i)$. We color these additional edges \emph{red}. 
\end{definition}

Occasionally, when it is clear what the partial allocation $\X$ is from context, we just write $G$, $\Gr$, and $\Ge$ instead of $G(\X)$, $\Gr(\X)$, and $\Ge(\X)$, respectively.

\subsection{Basic Subroutines}

Before we proceed with the presentation of our algorithms, we first define some useful subroutines that we will be using throughout their design. The first one is a procedure that eliminates a given cycle in a (standard, reduced, or enhanced) envy graph by swapping the bundles of the agents backwards along the cycle. We refer to this process as a \emph{cycle resolution} on the graph; see Subroutine~\ref{alg:cycle-resolution}. Similarly, we can define a subroutine that eliminates all cycles in a (standard, reduced, or enhanced) envy graph; see Subroutine~\ref{alg:all-cycles-resolution}. 

Here we should discuss two small conventions, made for the sake of readability. First, the full information about the instance is available to all of our subroutines and algorithms. Nevertheless, we only include the ``necessary'' arguments as input, e.g., we write $\ACR(\X,\tilde{G})$ rather than \ACR$(N, (v_i)_{i\in N}, \allowbreak \X,\tilde{G}(\X), M)$. Secondly, whenever one of the envy graph variants is used in any subroutine or algorithm, it is assumed that it is up to date. That is, every time the allocations change, the  relevant envy graphs are ``quietly'' updated, although we have opted not to add these extra steps to our already lengthy pseudocode descriptions.
In general, tie-breaking is not important for the properties of our subroutines and algorithms, so we assume that any tie-breaking is always done lexicographically.

\SetAlgorithmName{SUBROUTINE}{Subroutine}

\begin{algorithm}[!ht]
\DontPrintSemicolon
\caption{$\CR(\X,\tilde{G},C)$} \label{alg:cycle-resolution}
\SetKwComment{Comment}{/* }{ */}
\KwData{A partial allocation $\X$, its graph $\tilde{G}(\X)$, and a cycle $C$ in $\tilde{G}(\X)$}
\KwResult{An updated partial allocation $\X$ such that the (implied) graph $\tilde{G}(\X)$ no longer contains the cycle $C$.\vspace{3pt}}
   $\tilde{\X}\gets \X$ \tcp*[r]{$\tilde{\X}$ is an auxiliary allocation}
   \For{every edge $(i,j) \in C$}
   {
   $X_i \gets \tilde{X}_j$ \tcp*[r]{swap the bundles backwards along the cycle}
   }
   \Return $\X$
\end{algorithm}

In our case, the graph $\tilde{G}$ in Subroutines~\ref{alg:cycle-resolution} and \ref{alg:all-cycles-resolution} can be any of the envy graph $G$, the reduced envy graph $\Gr$, or the enhanced envy graph $\Ge$, {although we never call Subroutine \ref{alg:all-cycles-resolution} on the enhanced envy graph $\Ge$}. It is well-known that the \ACR subroutine terminates in polynomial time {in the case of a standard envy graph $G$}; we add a short proof here {for $G$ and $\Gr$} for completeness.

\begin{algorithm}[!ht]
\DontPrintSemicolon
\caption{$\ACR(\X,\tilde{G})$} \label{alg:all-cycles-resolution}
\SetKwComment{Comment}{/* }{ */}
\KwData{A partial allocation $\X$ and its graph $\tilde{G}(\X)$.}
\KwResult{An updated partial allocation $\X$ such that its graph $\tilde{G}(\X)$ is acyclic.\vspace{3pt}}
   \While{there exists a cycle $C$ in $\tilde{G}(\X)$}
   {
   $\X \gets \CR({\X,\ } \tilde{G}(\X),C)$ 
   }
    \Return $\X$
\end{algorithm}

\begin{lemma}\label{obs:all-cycle-resolution-running-time}
The \ACR subroutine terminates in polynomial time {when $\tilde{G}$ is the envy graph $G$ or the reduced envy graph $\Gr$}.\footnote{{In the conference version of this work \citep{AFS24}, the corresponding lemma was stated for all three types of envy graph we consider (standard, reduced, or enhanced). When the allocation is of size at most $2$, as is the case here, the statement is indeed true for enhanced graphs as well, but not with the same simple proof we use here for standard and reduced envy graphs; see also \citep{PrakashMN25}. Although we suspect that the statement holds for enhanced graphs in general, we do not attempt to prove it here, as \Cref{obs:all-cycle-resolution-running-time} is never needed for enhanced graphs in the current version of our algorithms.}}
\end{lemma}

\begin{proof}
{Checking the condition of the while loop of the subroutine can be done in $O(n^2)$ time by running a depth-first search on the graph $\tilde{G}(\X)$ and looking for back edges.}
Every time a cycle is resolved within the body of the while loop, the number of edges in $\tilde{G}(\X)$ strictly decreases.
This means that the while loop will terminate in at most {$O(n^2)$} iterations. Each iteration only needs $O(n)$ {operations for  resolving the cycle and $O(n^2)$ operations for updating $\tilde{G}(\X)$.}
\end{proof}

The third subroutine deals with swapping bundles along paths of the (standard, reduced, enhanced) envy graph. In particular, for any edge $(i,j)$ in the path, agent $i$ will receive the bundle of agent $j$. We do not specify what the bundle of the last agent on the path (for whom, since we will be considering paths of maximum cardinality, there will be no outgoing edge) will be. Similarly, we do not specify what will happen to the bundle that the first agent in the path (the source) had before the swap. Both of these will be handled by our algorithms later. We refer to this process as a \emph{path resolution} on the graph; see Subroutine~\ref{alg:path-resolution}. The \PR subroutine obviously needs $O(|\Pi|) = O(n)$ time {for  resolving the path and $O(n^2)$ time for updating $\tilde{G}(\X)$.} Again, in our case, the graph $\tilde{G}$ in Subroutine~\ref{alg:path-resolution} can be any of the envy graph variants we use.

\begin{algorithm}[!ht]
\DontPrintSemicolon
\caption{$\PR(\X,\tilde{G},\Pi)$} \label{alg:path-resolution}
\SetKwComment{Comment}{/* }{ */}
\KwData{A partial allocation $\X$, its graph $\tilde{G}(\X)$, and a path $\Pi = (i_1, i_2, \ldots, i_\ell)$ in $\tilde{G}(\X)$}
\KwResult{An updated set of bundles $X_i$, one for every agent $i \in \{i_1, i_2, \ldots, i_{\ell-1}\}$. \vspace{3pt}} 
\For(\tcp*[f]{go through every $i$ such that $(i,j) \in \Pi$ following the path}){$k \gets 1$\  to\  $\ell-1$} 
{
$X_{i_k} \gets X_{i_{k+1}}$ \tcp*[r]{assign to $i_k$ the bundle of agent $i_{k+1}$ who she envies}
}
\Return $(X_i)_{i \in N: \exists (i, i') \in \Pi}$
\end{algorithm}

As we noted above, \PR does not return a partial allocation, but only the updated bundles of the agents along the path, except for the bundle of the last agent. This is to allow some flexibility in building the latter when \PR is used as a subroutine. Using the \PR subroutine, we also define the following useful subroutine, \PRPA (Subroutine~\ref{alg:path-resolution-and-critical}). This subroutine assumes that the sources of the graph $\G$ have two goods in their bundles; our algorithms will ensure this is always the case. Then, it performs \PR on the path with respect to a target agent $i \in N$, and allocates to $i$ one of the goods of the respective source and one designated good from the pool. 

\begin{algorithm}[!ht]
\DontPrintSemicolon
\caption{$\PRPA(\X,\tilde{G},\Pi)$} \label{alg:path-resolution-and-critical}
\SetKwComment{Comment}{/* }{ */}
\KwData{A partial allocation $\X$, its graph $\tilde{G}(\X)$, and a path $\Pi = (s, \ldots, i)$ in\, $\tilde{G}(\X)$ starting at a source $s$ of $\tilde{G}(\X)$ with $|X_s|=2$}
\KwResult{An updated partial allocation $\X$.\vspace{3pt}}
Let $g_s \in \argmax_{g\in X_s} v_i(g)$ \tcp*{recall that $s$ is the first and $i$ is the last vertex of $\Pi$}
Let $g_* \in \argmax_{g\in \PP(\X)} v_i(g)$\;
$(X_k)_{k \in N: \exists (k, k') \in \Pi} \gets \PR(\X,\tilde{G},\Pi)$\;\vspace{2pt} 
$X_i \gets \{g_*,g_s\}$ \tcp*{agent $i$ gets her favorite goods from $X_s$ and $\PP(\X)$}
\Return $\X$
\end{algorithm}

\section{The Property-Preserving Partial Allocation (3PA) Algorithm}\label{sec:3pa}

As we mentioned in the Introduction, our approach to achieving $2/3$-EFX allocations will be the following. We will first construct a partial allocation $\X^1$ that satisfies certain properties {with respect to the values each agent has for her bundle, for other agents' bundles and for unallocated goods}, as well as the number of critical goods. Then, we will allocate these critical goods, ensuring that the resulting partial allocation $\X^2$ is still $2/3$-EFX. Then, it is known that we can run the Envy Cycle Elimination algorithm of \citet{lipton2004approximately} to complete $\X^2$, resulting in a $2/3$-EFX allocation; see also \cref{lem:partialefx-to-efx-extension} in \cref{sec:key-lemmas}. Given this, it suffices to focus on obtaining the allocation $\X^2$.

In this section, we first define the aforementioned properties, and then we develop the algorithm that produces the partial allocation $\X^1$ that satisfies these properties; see Algorithm \ref{alg:3PA}. We refer to $\X^1$ as a \emph{Property-Preserving Partial Allocation (3PA)}, and to the corresponding algorithm as the \emph{\PPA algorithm}. 
This will in fact be the ``basic'' variant of the algorithm, which is sufficient to achieve $2/3$-EFX for multigraph value instances and for a small number of agents, which we present in \cref{sec:multigraph-and-few-agents}. 

\begin{tcolorbox}
\underline{Desired properties of a partial allocation $\X$:}
\begin{properties}[topsep=1pt,itemsep=1pt]
\item Every agent $i \in N$ with $|X_i|=1$ is EFX towards any other agent. 
\label{prop:a}
\item Every agent $i \in N$ is $2/3$-EFX towards any other agent. 
\label{prop:b}
\item Every agent $i$ {weakly} prefers her bundle to any single unallocated good, i.e., for every agent $i \in N$ and good $g \in \PP(\X)$, $v_i(X_i)\geq v_i(g)$.\label{prop:c}
\item Every agent $i \in N$ with $|X_i|>1$ does not have any critical goods, i.e., for every {such agent $i$} and any good $g \in \PP(\X)$, $v_i(g)\leq \frac12 v_i(X_i)$.\label{prop:d}
\item Any agent $i$ with $|X_i|=1$ has at most one critical good $g_i$, and she values that good at most $\frac23$ of the value of her current bundle, i.e., for every
{such agent $i$}, there is at most one good $g_i \in \PP(\X)$ such that $v_i(g_i)> \frac12 v_i(X_i)$ {and, even if such $g_i$ exists,} it holds that $v_i(g_i) \le  \frac23 v_i(X_i)$.\label{prop:e}
\end{properties}
\end{tcolorbox}

It is easy to see that {\cref{prop:d,prop:e} together imply} \cref{prop:c}. Nevertheless, we keep this redundancy for the sake of presentation, {as the weaker \cref{prop:c} is easier to interpret at a glance}. For the case of $3$-value instances, we will need a refined version of this list {of properties and the corresponding variant of the \PPA algorithm}, which we present in \cref{sec:three-values}. 

We are almost ready to state the \PPA algorithm. To allow flexibility in the use of the algorithm as a potential subroutine, \PPA takes as input a partial allocation $\X$ of size at most $2$, which 
does not contain any empty bundles and
satisfies \cref{prop:a,prop:b}. We call any such allocation a \emph{seed allocation}.
For our results in the following section, it suffices to start with a seed allocation in which each agent arbitrarily receives a single good (recall that this is possible as we have assumed that $m>n$). Then the algorithm executes a series of steps, each of which is executed only if the previous one is not. At a high level, for several steps of the algorithm, certain agents are enabled to swap (parts) of their bundles with goods from the pool. These agents are considered according to a certain priority, which depends on the cardinality of their bundles in the partial allocation $\X$. Several steps of the algorithm also consider the reduced envy-graph $\Gr$ and the enhanced envy-graph $\Ge$ and eliminate cycles in those graphs by invoking the \ACR {and \CR subroutines, respectively}. Finally, the algorithm also takes care of the corner case where the pool no longer contains enough goods  valuable enough for trades, which is addressed via the \SP subroutine (Subroutine~\ref{alg:singleton-pool}).

\SetAlgorithmName{SUBROUTINE}{Subroutine}

\setcounter{algocf}{4}
\begin{algorithm}[!h]
\DontPrintSemicolon
\caption{$\SP(\X)$} \label{alg:singleton-pool}
\SetKwComment{Comment}{/* }{ */}
\KwData{A partial allocation $\X$ (of size at most $2$ with a \textit{single} unallocated good $g$, which some agent $i$---who has $|X_i|=1$ and is not a source in $\Gr(\X)$---values more than $2/3\cdot v_i(X_i)$). }
\KwResult{A partial allocation $\X'$, such that $g \notin \PP(\X')$ and some other good $g'$ is returned to $\PP(\X')$.\vspace{3pt}}
Let $g$ be the only good in $\PP(\X)$ and let $i \in N$ be such that $|X_i|=1$ and $v_i(g) > \frac{2}{3}v_i(X_i)$\;
\tcp{when \SP is called, such $g$ and $i$ do exist and $i$ is not a source in $\Gr(\X)$ }
Let $\Pi$ be a path in $\Gr(\X)$ starting from some source $s$ of $\Gr(\X)$ and terminating at $i$\;
\tcp{when \SP is called, no sources in $\Gr(\X)$ own only one good {(due to Step \ref{step6} of \PPA)}}
$\X \gets \PRPA(\X, \Gr, \Pi)$\;
   \Return $\X$
\end{algorithm}

{Recall that the pool $\PP(\cdot)$ is a function of the allocation and, thus, it is always implicitly updated every time one or more bundles change.}
The following lemma establishes that the partial allocation produced by  \PPA  satisfies properties \cref{prop:a,prop:b,prop:c,prop:d,prop:e}, and that its enhanced envy  graph has a source without critical goods. 

\begin{lemma}\label{lem:3PA-properties}
Let  $\X$ be a seed allocation. Then the output of \,\PPA\!$(\left(v_i\right)_{i 
\in N}, \X)$ is a partial allocation $\X^1$ of size at most $2$ that satisfies \cref{prop:a,prop:b,prop:c,prop:d,prop:e}. 
Furthermore, if\, $\X^1$ is not a complete allocation, then $\Ge(\X^1)$ has at least one source, and every source $s$ in $\Ge(\X^1)$ has $|X^1_s| = 2$.
\end{lemma}

\newgeometry{margin = 0.75in} 

\SetAlgorithmName{ALGORITHM}{Algorithm}

\setcounter{algocf}{0}


\begin{algorithm}[p]
\setstretch{1.2}
\DontPrintSemicolon
\SetNoFillComment
\LinesNotNumbered 
\caption{\textsc{\sc Property-Preserving Partial Allocation (\PPA)} 
$\left(\left(v_i\right)_{i 
\in N}, \X \right)$} \label{alg:3PA}
\SetKwComment{Comment}{/* }{ */}
\SetKw{Continue}{continue}
\SetKw{Break}{break}
\SetKw{Step}{Step}
\SetKwData{Kw}{}
\KwData{The values $v_i(g)$ for every agent $i \in N$ and every good $g \in M$, and a partial allocation $\X$ of size at most $2$ which satisfies \Cref{prop:a,prop:b}. 
}
\KwResult{A Property-Preserving Partial Allocation $\X^1$ of size at most $2$, i.e., a partial allocation $\X^1$ of size at most $2$ which satisfies \Cref{prop:a,prop:b,prop:c,prop:d,prop:e}.\vspace{3pt}}
\While{$\PP(\X) \neq \emptyset$}
{
\nl \Kw{\texttt{Step 1}} \label{step1}
\uIf{there is $i \in N$ with $|X_i|=1$ and a good $g \in \PP(\X)$ such that $v_i(g) > v_i(X_i)$}{
$X_i \gets \{g\}$ \;
\tcp{If an agent with one good prefers one good from the pool, swap them.}}

\nl \Kw{\texttt{Step 2}}\label{step2}
\uElseIf{there is $i \in N$ with $|X_i|=2$ and a good $g \in \PP(\X)$ such that $v_i(g) > \frac{3}{2}v_i(X_i)$}{
$X_i \gets \{g\}$ \;
\tcp{Else if an agent with two goods prefers one good from the pool by a factor greater than $3/2$, swap her bundle with that good.}
}

\nl \Kw{\texttt{Step 3}} \label{step3}
\uElseIf{there is $i \in N$ with $|X_i|=1$ and  goods $g_1,g_2 \in \PP(\X)$ such that $v_i(\{g_1,g_2\}) \ge \frac{2}{3}v_i(X_i)$}{
$X_i \gets \{g_1,g_2\}$\;
\tcp{Else if an agent with one good prefers two goods from the pool by a factor {of at least  $2/3$}, swap her bundle  with these two goods.} 
}

\nl \Kw{\texttt{Step 4}}\label{step4}
\uElseIf{there is $i \in N$ with $|X_i|=2$, and goods $g \in \PP(\X)$ and $g' \in X_i$ such that $v_i(g) > v_i(g')$}{
$X_i \gets (X_i \cup \{g\}) \setminus \{g'\}$.\;
\tcp{Else if an agent with two goods strictly prefers one good from the pool to one of her own goods, swap that good with the one good from the pool.}
}

\nl \Kw{\texttt{Step 5}}\label{step5}
\uElseIf{the reduced envy graph $\Gr(\X)$ has cycles}{$\X \gets\ACR(\X,\Gr)$\;
\tcp{Else if the reduced envy graph has cycles, resolve them by swapping the bundles accordingly.}}

\nl \Kw{\texttt{Step 6}}\label{step6}
\uElseIf{in the reduced envy graph $\Gr$ there is a \emph{source} $s$ with $|X_s|=1$}{ 
$X_s \gets X_s \cup \{g^*\}$, where $g^* \in \argmax_{g \in \PP(\X)}v_s(g)$.\; 
\tcp{Else if there is a source in the reduced envy graph with a single good, add her most valuable good from the pool to her bundle.}
}

\nl \Kw{\texttt{Step 7}} \label{step7}
\uElseIf{ {$\PP(\X)=\{g\}$ and there is $ i \in N$ such that $|X_i|=1$ and $v_i(g) > \frac{2}{3}v_i(X_i)$} }
{$\X \gets \SP(\X)$\;
\tcp{Else if there is a single unallocated good and some agent with one good prefers it by a factor greater than $2/3$, run the \SP subroutine to allocate it.}}

\nl \Kw{\texttt{Step 8}}\label{step8}
\uElseIf{{there is a cycle $C$ in the enhanced envy graph $\Ge(\X)$ }}{
$\X \gets \CR(\X,\Ge, C)$\;
\tcp{Else if the enhanced envy graph has cycles, resolve one of them by swapping the bundles.}
}

\nl \Kw{\texttt{Step 9}}\label{step9}
\uElseIf{there is a path $\Pi = (s, \ldots, i)$ in\, $\Ge(\X)$  starting at a source $s$ of\, $\Ge(\X)$, along with some $g\in \PP(\X)$ and some $g' \in X_s$, so that 
$v_i(X_i)< v_i(\{g,g'\})$, if $|X_i| = 2$, \textit{\ or\ } $v_i(X_i) \le \frac32 v_i(\{g,g'\})$,  if $|X_i| = 1$
\vspace{2pt}}
{ 
$\X \gets \PRPA(\X, \Ge, \Pi)$\;
\tcp{Else if there exists a path from a source $s$ to some agent $i$ in $\Ge(\X)$, such that some good from the pool and some good from $X_s$ can be  ``favorably compared'' to $i$'s own bundle, swap the bundles along the envy path and give those two goods to $i$. We require $v_i(\{g,g'\}) > v_i(X_i)$ when $|X_i| = 2$, but only $v_i(\{g,g'\}) \ge \frac23 v_i(X_i)$  when $|X_i| = 1$.}
}

\nl \Kw{\texttt{Step 10}}\label{step10}
\Else {\Break\;}
}
\Return $\X$ 
\end{algorithm}

\restoregeometry

As a direct consequence of \cref{lem:3PA-properties}, since $|X^1_s| = 2$ for every source $s$ in $\Ge(\X^1)$ and \cref{prop:d} applies, no source $s$ of $\Ge(\X^1)$ has {any} critical goods in $\X^1$. 
To prove the lemma, we will devise and prove a series of other simpler lemmata. We start with  \cref{lem-3PA-properties-ab}, which establishes that $\X^1$ satisfies \cref{prop:a,prop:b}. This is the most challenging part of proving \cref{lem:3PA-properties}.

\begin{lemma}\label{lem-3PA-properties-ab}
Let  $\X$ be a seed allocation. Then the output of \,\PPA\!$(\left(v_i\right)_{i 
\in N}, \X)$ is a partial allocation $\X^1$ of size at most $2$ that satisfies \cref{prop:a,prop:b}.
\end{lemma}

\begin{proof}

We {begin with the observation} that $\X^1$ is of size at most $2$. Since $\X$ is a seed allocation,  initially all the bundles have cardinality $1$ or $2$. Then it is just a matter of inspection to see that no step in \PPA can increase the cardinality of a bundle from $1$ or $2$ to more than $2$.

Next we will prove by induction that \Cref{prop:a,prop:b} are satisfied at the end of the execution of the algorithm. In particular, we will prove that the properties are always satisfied right before the condition of the main while loop of the algorithm is checked. 
For the base case, it follows by assumption on the input partial allocation $\X$ that \Cref{prop:a,prop:b} are satisfied, before the first iteration of the loop. We will assume that the properties are satisfied before iteration $k$ of the loop and show that (if the loop condition is met) they are satisfied at the end of iteration $k$ as well, or, equivalently, right before iteration $k+1$. 
Notice that in each iteration of the while loop, only one of the steps of the algorithm will be executed. In particular, Step $i \in \{1,\ldots,10\}$ will only be executed if all Steps $j \in \{1,\ldots,i-1\}$ are \emph{not} executed. We will consider each possible step and argue that \Cref{prop:a,prop:b} are satisfied after the execution of that step. For ease of reference, we will use $\X$ and $\X'$ to denote the partial allocation before and after the execution of the step respectively.

To avoid unnecessary repetition, in each case we will completely disregard pairs of agents whose bundles did not change during the iteration at hand; this extends to pairs $(i, j)$, such that $i$'s bundle did not change and $j$ received a preexisting bundle via some exchange. It is implicit that for such pairs, \Cref{prop:a,prop:b} are always satisfied after the execution of the loop, given that they were satisfied before it. Also, we will omit explicit references to trivial cases like \Cref{prop:a} being true for agents with two goods or \Cref{prop:b} being true for agents with one good for whom \Cref{prop:a} holds.
\smallskip

\noindent \textit{If Step~\ref{step1} is executed:} Obviously $v_i(X_i') > v_i(X_i)$, hence \Cref{prop:a,prop:b} are satisfied for agent $i$ by the induction hypothesis. Since $|X_i'|=1$, any other agent $j \in N\setminus\{i\}$ is obviously EFX towards agent $i$, and hence \Cref{prop:a,prop:b} are satisfied for all agents. 
\smallskip

\noindent \textit{If Step~\ref{step2} is executed:} By the induction hypothesis, agent $i$ was $2/3$-EFX towards any other agent $j \in N\setminus \{i\}$ in $\X$. Since $v_i(X_i') \geq \frac{3}{2} v_i(X_i)$, this implies that in $\X'$ agent $i$ is EFX towards any other agent, and \Cref{prop:a,prop:b} are satisfied for agent $i$. 
For agent $i$, we now have $|X_i'|=1$, so any other agent $j \in N \setminus \{i\}$ is trivially EFX towards agent $i$ as well. Thus, \Cref{prop:a,prop:b} are satisfied for all agents. \smallskip

\noindent\textit{If Step~\ref{step3} is executed:} In this step, the value of agent $i$ for $X'_i$ might be reduced, to at least $2/3$ of her value for $X_i$, i.e., $v_i(X_i') \geq \frac{2}{3}v_i(X_i)$. In $\X$, since $|X_i|=1$, agent $i$ was EFX towards every other agent by the induction hypothesis. By the inequality above, agent $i$ in $\X'$ is now $2/3$-EFX towards any agent $j \in N \setminus\{i\}$. 
Now consider any other agent $j \in N \setminus\{i\}$; we have two cases:
\begin{itemize}[leftmargin=*]
        \item[-] $|X_j|=1$. In this case we know that for $\ell \in \{1,2\}$, $v_j(X_j) \geq v_j(g_\ell)$, i.e., agent $j$ (weakly) prefers her own allocation $X_j$ to any of the two goods $g_1$, $g_2$ that constitute agent $i$'s bundle $X_i'$. 
        This is because otherwise $j$ could have swapped $X_j$ with one of $g_1$ and $g_2$ in Step~\ref{step1}, which by assumption was not executed. From this, it follows that agent $j$ is EFX towards agent $i$. 
        \item[-] $|X_j|=2$. We have that for $\ell \in \{1,2\}$, $v_j(g_\ell) \leq \frac{3}{2}v_j(X_j)$; otherwise, Step~\ref{step2} could have been executed and agent $j$ would have swapped $X_j'$ with one of $g_1$ or $g_2$. It follows that agent $j$ is $2/3$-EFX towards agent $i$. 
    \end{itemize}
In both cases, \Cref{prop:a,prop:b} are satisfied for all agents.    
\smallskip

\noindent\textit{If Step~\ref{step4} is executed:} Again, we have that $v_i(X_i') > v_i(X_i)$, hence \Cref{prop:a,prop:b} are satisfied for agent $i$ by the induction hypothesis. Now consider any agent $j \in N \setminus\{i\}$; we have two cases:
\begin{itemize}[leftmargin=*]
        \item[-] $|X_j|=1$. Similarly to before, we know that $v_j(X_j) \geq v_j(g)$, i.e., agent $j$ does not prefer the good added to agent $i$'s bundle $X'_i$ from the pool, as otherwise Step~\ref{step1} of the algorithm would have been executed, and agent $j$ would have swapped $X_j$ for $g$. We also have that $v_j(X_j) \geq v_j(X_i'\setminus\{g\}) {\, = v_j(X_i\setminus\{g'\})}$, i.e., agent $j$ does not prefer the good that
        {is contained in both $X_i$ and $X'_i$}.
        This holds because of the induction hypothesis, as \cref{prop:a} ensures that agent $j$ was EFX towards agent $i$. This establishes that \cref{prop:a} holds for agent $j$ in $\X'$ as well. 
        \item[-] $|X_j|=2$. The argument here is similar to the case above. We know that $v_j(g) \leq \frac{3}{2}v_j(X_j)$, as otherwise Step~\ref{step2} of the algorithm could have been executed, and agent $j$ would have swapped her bundle $X_j$ with the good $g$ from the pool. Additionally, by \cref{prop:b} which holds for $\X$ by the induction hypothesis, we know that $v_j(X'_i \setminus\{g\}) = v_j(X_i \setminus \{g'\}) \leq \frac{3}{2} v_j(X_j)$. This establishes that in $\X'$, agent $j$ does not value any of the goods in $X'_i$ by more than $3/2$ of her value for $X_j$, and \cref{prop:b} is satisfied. 
    \end{itemize}

\noindent\textit{If Step~\ref{step5} is executed:} After this step, we obviously have $v_i(X_i') \geq v_i(X_i)$, for any $i \in N$, as the utility of any agent cannot decrease. For any agent $i$ with $|X_i| = 1$, the agent was EFX towards any other agent in $\X$, by the induction hypothesis, and therefore she is also EFX towards any other agent in $\X'$. Similarly, for any agent $i \in N$ such that $|X_i|=|X_i'|=2$, the agent is $2/3$-EFX towards any other agent in $\X'$, since this was the case in $\X$ by the induction hypothesis. It remains to show that every agent $i$ for whom $|X_i|=2$ and $|X_i'|=1$ is now EFX in $\X'$ towards every other agent, noting that the agent was $2/3$-EFX in $\X$ by the induction hypothesis. 
This is indeed the case here, since $i$ received a bundle $X'_i$ that she valued more than $3/2$ of her value for $X_i$, i.e., $v_i(X_i') > \frac{3}{2} v_i(X_i)$. 
To see this, {suppose that $j$ was the owner of $X_i'$ in $\X$, i.e., $X_j = X_i'$. The reduced envy graph $\Gr(\X)$ contained the edge $(i,j)$; by \Cref{def:Gr}, since $|X_i|=2$ and $|X_j|=1$, we have  $v_i(X_i) < \frac{2}{3} v_i(X_j) = \frac{2}{3} v_i(X_i')$.}
This establishes that in $\X'$, agent $i$ is EFX towards any other agent. In the end, \cref{prop:a,prop:b} are satisfied for all agents. \smallskip

\noindent\textit{If Step~\ref{step6} is executed:} Consider agent $s$ first; we have {$v_s(X_s') > v_s(X_s)$}. 
By the induction hypothesis, since $|X_s|=1$, agent $s$ was EFX towards every other agent by \cref{prop:a}. This means that in $\X'$ agent $s$ is EFX towards any other agent, and \cref{prop:a,prop:b} are satisfied. Now consider any other agent $j \in N \setminus\{s\}$; again we distinguish two cases. 
\begin{itemize}[leftmargin=*]
        \item[-] $|X_j|=1$. Similarly to before, we know that $v_j(X_j) \geq  v_j(g^*)$, otherwise Step~\ref{step1} would have been executed and agent $j$ would have swapped $X_j$ with $g^*$. 
        Since $s$ was a source of $\Gr(\X)$, we also know that $v_j(X_j) \geq v_j(X_s)$. {Given that $v_j(X_j') \geq v_j(X_j)$, the above two inequalities establish} that agent $j$ is EFX towards agent $s$, 
        and therefore \cref{prop:a} is satisfied. 
        \item[-] $|X_j|=2$. Again, we know that $v_j(X_j) \geq {\frac{2}{3}\,} v_j(g^*)$, otherwise Step~\ref{step2} would have been executed and agent $j$ would have swapped $X_j$ with $g^*$. Additionally, since $s$ was a source of $\Gr(\X)$ with $|X_s|=1$, we also know that $v_j(X_j) \geq \frac{2}{3}v_j(X_s)$, as otherwise there would be an edge $(j,s)$ in $\Gr(\X)$. These two inequalities, {along with the fact that $v_j(X_j') \geq v_j(X_j)$}, imply that agent $j$ is $2/3$-EFX towards agent $s$, and therefore \cref{prop:b} is satisfied. 
    \end{itemize}

\noindent\textit{If Step~\ref{step7} is executed:} The \SP subroutine performs a path resolution on the reduced graph $\Gr(\X)$. Consider the agent $i$ in the resolved path $P$ that lies at the end of the path (the ``target'' vertex of the path), who receives the good $g$ and {her preferred} good from $X_s$. This is also the agent that ``triggered'' the execution of the \SP subroutine, as $|X_i|=1$ and $v_i(g) > \frac{2}{3}v_i(X_i)$. For all agents $j \in N\setminus\{i\}$, we have that $v_j(X_j') \geq v_j(X_j)$. For those agents, very similar arguments to the ones used in the analyses of 
{Steps~\ref{step5} and
\ref{step6} (with agent $i$ taking the role that $s$ had in Step~\ref{step6})} establish that \cref{prop:a,prop:b} are satisfied in $\X'$. 
{In particular, consider any agent $j\in N\setminus\{i\}$. If $j$ is not on the path $\Pi$, her bundle is unchanged; if $j$ is on $\Pi$ and $j\neq i$, she receives the bundle of her successor on $\Pi$, which she envied in $\Gr(\X)$. In either case, $v_j(X_j')\ge v_j(X_j)$, and if such a $j$ moves from $|X_j|=2$ to $|X_j'|=1$, then, by \Cref{def:Gr}, it holds that $v_j(X_j')>\frac32 v_j(X_j)$, so $j$ becomes EFX exactly as in the analysis of Step~\ref{step5}. Hence \cref{prop:a,prop:b} hold for $j$ toward every agent holding one of the (unchanged) bundles of $\X$. It remains to check them toward $i$, whose new bundle is $X_i'=\{g_s,g\}$, with $g_s\in X_s$ and $g$ the (formerly) unallocated good. Since Steps~\ref{step1} and \ref{step2} did not run, $v_j(g)\le v_j(X_j)$, if $|X_j|=1$, and $v_j(g)\le\tfrac32 v_j(X_j)$, if $|X_j|=2$; and since $s$ is a source of $\Gr(\X)$ with $|X_s|=2$ (so no edge into $s$ is removed in forming $\Gr$), we have $v_j(X_j)\ge v_j(X_s)\ge v_j(g_s)$. Together with $v_j(X_j')\ge v_j(X_j)$ ($v_j(X_j')>\tfrac32 v_j(X_j)$ in the case $|X_j'|=1<|X_j|$) these inequalities give that $j$ is EFX towards $i$ when $|X_j'|=1$ and $2/3$-EFX towards $i$ when $|X_j'|>1$, which is exactly the role $s$ played in Step~\ref{step6}.}
\smallskip

\noindent\textit{If Step~\ref{step8} is executed:} 
Let $i \in N$ be some agent that is part of the cycle $C$ resolved by \CR; after the resolution, the agent receives agent $j$'s bundle, where $(i,j)$ was an edge in $\Ge(\X)$; i.e., $X_i' = X_j$. We consider three cases:
\begin{itemize}[leftmargin=*]
\item[-] $|X_i| = |X_j|$. In this case, we know that $(i,j) \in \Gr$, i.e., it is not red, since in $\Ge$ additional edges are from agents with bundles of cardinality $1$ to agents with bundles of cardinality $2$.
Therefore, we have that $v_i(X_i')=v_i(X_j) > v_i(X_i)$. Since agent $i$'s utility increased and her bundle did not change in size, \cref{prop:a,prop:b} are satisfied for this agent by the induction hypothesis.
\item[-] $|X_i|=2$ and $|X_j|=1$. Again, we know that $(i,j) \in \Gr$, i.e., it is not red, since in $\Ge$ we only add edges originating from agents with bundles of cardinality $1$. By definition of $\Gr(\X)$, we have that $v_i(X_i')=v_i(X_j) > \frac{3}{2}v_i(X_i)$. By the induction hypothesis, \cref{prop:b} was satisfied for agent $i$ in $\X$, and now \cref{prop:a} is satisfied for agent $i$ in $\X'$. 
\item[-] $|X_i|=1$ and $|X_j|=2$. In this case the edge $(i,j)$ may be red or not. In either case, it holds that $v_i(X_i') = v_i(X_j) \geq \frac{2}{3}v_i(X_i)$. Since \cref{prop:a} was satisfied for agent $i$ by the induction hypothesis, \cref{prop:b} is now satisfied for $i$. 
\end{itemize}
In all cases, \Cref{prop:a,prop:b} are satisfied for all agents.
\smallskip

\noindent\textit{If Step~\ref{step9} is executed:} 
Here a path resolution on the enhanced graph $\Ge(\X)$ is performed, along the path $\Pi$ from the source
agent $s$ to agent $i$. If we ignored $i$, the fact that the allocation 
$(X'_1,\ldots,X'_{i-1},X'_{i+1},\ldots,X'_n)$ satisfies \Cref{prop:a,prop:b} is shown with the exact 
same arguments we used for Step~\ref{step8}. What remains to be shown is that \Cref{prop:a,prop:b} are
satisfied for $i$ herself but also that $X'_i$ does not cause \Cref{prop:a,prop:b} to fail
for any agent $j\in N\setminus \{i\}$.

We first consider $i$'s point of view:
\begin{itemize}[leftmargin=*]
\item[-] $|X_i| = 2$. In this case, $v_i(X_i')=v_i(\{g, g'\}) > v_i(X_i)$. Since agent $i$'s utility increased and her bundle's size did not change, by the induction hypothesis, \cref{prop:b} is satisfied for $i$.
\item[-] $|X_i| = 1$. In this case, $v_i(X_i') = v_i(\{g, g'\}) \geq \frac{2}{3}v_i(X_i)$. Since \cref{prop:a} was satisfied for agent $i$ by the induction hypothesis, \cref{prop:b} is now satisfied for $i$. 
\end{itemize}

Finally, consider any $j\in N\setminus \{i\}$:
\begin{itemize}[leftmargin=*]
\item[-] $|X_j| = 1$. We have $v_j(X_j) \geq  v_j(g)$, otherwise Step~\ref{step1} would have been executed and agent $j$ would have swapped $X_j$ with $g$. Since $s$ is a source in $\Ge(\X)$, we also know that $v_j(X_j)\ge v_j(g')$. After the path resolution we may have $|X_j'|=2$ or $|X_j'|=1$. If $|X_j'|=1$, we have $v_j(X_j') \geq v_j(X_j)$ and the inequalities about $X_j$, $g$ and $g'$ imply that agent $j$ is EFX towards agent $i$, and therefore \cref{prop:a} is satisfied. If $|X_j'|=2$, we have $v_j(X_j') \geq \frac23 v_j(X_j)$ and the inequalities  imply that  $j$ is $2/3$-EFX towards agent $i$, and thus \cref{prop:b} is satisfied.

\item[-] $|X_j|=2$. We now have $v_j(X_j) \geq  \frac23 v_j(g)$, otherwise Step~\ref{step2} would have been executed and $X_j$ would have been swapped with $g$. Also, since $s$ is a source in $\Ge(\X)$, it is now the case that $v_j(X_j)\ge \frac23 v_j(g')$. After the path resolution we may have $|X_j'|=2$ or $|X_j'|=1$. If $|X_j'|=2$, we have $v_j(X_j') \geq v_j(X_j)$ and the inequalities about $X_j$, $g$ and $g'$ imply that agent $j$ is $2/3$-EFX towards agent $i$, and therefore \cref{prop:b} is satisfied. If $|X_j'|=1$, we have $v_j(X_j') \geq \frac32 v_j(X_j)$ and the inequalities  imply that  $j$ is EFX towards agent $i$, that is, \cref{prop:a} is satisfied.
\end{itemize}

\noindent\textit{If Step~\ref{step10} is executed:} With this step, the algorithm terminates. In particular, the step does not change the allocation $\X$, and therefore \Cref{prop:a,prop:b} are satisfied trivially by the induction hypothesis. 
\end{proof}

Our next lemma establishes that the allocation $\X^1$ outputted by the \PPA algorithm also satisfies \cref{prop:c,prop:d,prop:e}. 

\begin{lemma}\label{lem-3PA-properties-cde}
Let  $\X$ be a seed allocation. Then the output $\X^1$ of \,\PPA\!$(\left(v_i\right)_{i 
\in N}, \X)$ satisfies \cref{prop:c,prop:d,prop:e}.
\end{lemma}

\begin{proof}
Suppose, towards a contradiction, that $\X^1$ violates one of \cref{prop:c,prop:d,prop:e}. Assume first that it violates \cref{prop:c}. That means that there exists some agent $i\in N$ that prefers some good $g \in \PP(\X^1)$ from the pool to her bundle $X^1_i$. We consider two cases:
\begin{itemize}[leftmargin=20pt]
\item[-] $|X_i^1|=1$. In that case, the \PPA algorithm would not have terminated and 
Step~\ref{step1} would be executed.
\item[-] $|X_i^1|=2$. In that case, the \PPA algorithm would not have terminated and
{either Step~\ref{step2} (if $v_i(g) > \frac{3}{2}v_i(X^1_i)$) or Step~\ref{step4} (if $v_i(X^1_i) < v_i(g) \le  \frac{3}{2}v_i(X^1_i)$)} would be executed.
\end{itemize}

Next assume that \cref{prop:d} is not satisfied. That means that there exists some agent $i \in N$ with $|X_i^1|=2$ that has a critical good, i.e., a good $g^* \in \PP(\X^1)$ such that $v_i(g^*) > \frac{1}{2}v_i(X_i^1)$. 
This means that $v_i(g^*) > \min_{g \in X_i^1}v_i(g)$, i.e., there is a good from the pool $\PP(\X^1)$ that agent $i$ prefers to one of the goods in her bundle; this follows by the additivity of the values. 
In that case, the \PPA algorithm would not have terminated and 
one of the steps up to Step~\ref{step4} would be executed.\smallskip

Finally, assume that \cref{prop:e} is not satisfied. That has one of the following two implications:
\begin{itemize}[leftmargin=20pt]
\item[-] there exists some agent $i \in N$ with $|X_i^1|=1$ that has at least two critical goods, or
\item[-] there exists some agent $i \in N$ with $|X_i^1|=1$ that has exactly one critical good $g_i \in \PP(\X^1)$, but that good's value exceeds $2/3$ of the value of her bundle, i.e., $\frac{2}{3}v_i(X_i^1) < v_i(g_i)$. 
\end{itemize}
Consider the first case, and let $g_i$ and $\tilde{g}_i$ be two of the critical goods of agent $i$. This means that $v_i(\{g_i,\tilde{g}_i\}) > v_i(X_i^1)$ and since $g_i, \tilde{g}_i \in \PP(\X^1)$, Step~\ref{step3} of the algorithm could have been executed for the agent to swap $X_i^1$ with $\{g_i,\tilde{g}_i\}$. This contradicts the fact that the algorithm terminated. Next consider the second case, where $g_i$ is the {unique} critical good of agent $i$, for which we have $v_i(g_i) > \frac{2}{3}v_i(X_i^1)$. Assume first that there is some other good $g' \in \PP(\X^1)$; in that case we have that $v_i(\{g_i,g'\}) > \frac{2}{3}v_i(X_i^1)$, since $v_i(g')\geq 0$. This again implies that Step~\ref{step3} of the algorithm could have been executed for agent $i$ to swap $X_i^1$ with $\{g_i,g'\}$, contradicting the fact that the algorithm terminated. Finally, assume that $\PP(\X^1)=\{g_i\}$. This is precisely the condition of Step~\ref{step7} of the \PPA algorithm, and hence that step could have been executed, contradicting the fact that the algorithm terminated. Overall, this establishes that the partial allocation $\X^1$ satisfies \cref{prop:c,prop:d,prop:e}. 
\end{proof}

{The last lemma in this section} establishes that the enhanced graph $\Ge$ of $\X^1$ has a source without critical goods. 
\begin{lemma}\label{lem-3PA-source-enhanced}
Let  $\X$ be a seed allocation and $\X^1$ be the output of \,\PPA\!$(\left(v_i\right)_{i 
\in N}, \X)$. 
If\, $\X^1$ is not a complete allocation, then $\Ge(\X^1)$ has at least one source, and every source $s$ in $\Ge(\X^1)$ has $|X^1_s| = 2$.
\end{lemma}

\begin{proof}
{If $\Ge(\X^1)$ had no sources, it would contain at least one cycle. Given that $\PP(\X^1) \neq \emptyset$, the main loop of the algorithm would have been executed with the condition of Step~\ref{step8} being true in the last iteration. That would contradict the fact that $\X^1$ is the final output of \PPA. We conclude that there is at least one source $s$ in $\Ge(\X^1)$.}
By the definition of the reduced and enhanced graphs $\Gr$ and $\Ge$ respectively, it immediately follows that if $s$ is a source in $\Ge(\X^1)$, it is also a source in $\Gr(\X^1)$. 
Since $\PP(\X^1) \neq \emptyset$, it is not possible to have $|X_s^1|=1$, as then Step~\ref{step6} of the algorithm could have been executed and agent $s$ would have received an extra good from $\PP(\X^1)$. Recall, however, that {$1\le |X_s^1|\le 2$} (see the beginning of the proof of \cref{lem-3PA-properties-ab}). Thus, $|X^1_s| = 2$.
\end{proof}

Clearly, combining \cref{lem-3PA-properties-ab,lem-3PA-properties-cde,lem-3PA-source-enhanced} completes the proof of \cref{lem:3PA-properties}. 
We conclude the section with a proof that the \PPA algorithm (Algorithm \ref{alg:3PA}) terminates in polynomial time.

\begin{lemma}\label{lem:3PA-polytime}
The \PPA algorithm terminates in polynomial time.
\end{lemma}

\begin{proof}
We begin with the easy fact that there are $\binom{m}{2} + \binom{m}{1}$ 
possible bundles for an agent, given that the \PPA algorithm maintains a partial allocation of size at most $2$ {with nonempty bundles}. For any $m \ge 1$, this quantity is at most $m^2$.
Next, observe that in every iteration of the while loop of the \PPA algorithm, except in the very last one, there is at least one agent whose bundle changes. 
We next claim that although the bundle of an agent might change, it never becomes one of the bundles she saw in the past iterations.

\begin{claim}\label{claim:no_repeated_bundles}
Once the bundle of an agent changes from $S$ to $S'$ during the execution of \PPA, it will never go back to $S$ in any future iteration. 
\end{claim}

\begin{proof}[Proof of \cref{claim:no_repeated_bundles}]
Fix an agent $i \in N$ and consider the finite sequence of her bundles $X_i^{(0)}, \allowbreak X_i^{(1)}, \allowbreak X_i^{(2)}, \ldots$, where the superscript indicates how many iterations of the while loop have been executed. First, observe that the only way that $i$'s value for her bundle decreases is through {Steps~\ref{step3}, \ref{step7}, \ref{step8}, and \ref{step9}}, but this happens only by moving from a bundle of cardinality $1$ to a bundle of cardinality $2$ and, despite the decrease, she maintains at least $2/3$ of the previous value. In any other case, $i$'s value for her bundle strictly increases. In particular, when moving from a bundle of cardinality $2$ to a bundle of cardinality $1$ (through {Steps~\ref{step2}, \ref{step5}, \ref{step7}, \ref{step8}, or \ref{step9}}) this increase is by a factor strictly larger than $3/2$.

Consider all the iterations $\ell_1 < \ell_2 < \ldots$, during which the bundle of agent $i$ changes, i.e., $X_i^{(\ell_\kappa-1)} \neq X_i^{(\ell_\kappa)}$ for any $\kappa \ge 1$. Using strong induction on $\kappa$, we are going to show that $v_i(X_i^{(\ell_\kappa)})$ is strictly larger than all the preceding bundles of cardinality $|X_i^{(\ell_\kappa)}|$. This would imply \Cref{claim:no_repeated_bundles}.

For $\kappa = 1$, observe that $X_i^{(0)} = X_i^{(1)} = \ldots = X_i^{(\ell_1-1)}$. There are two simple cases. First, if $|X_i^{(\ell_1-1)}| = |X_i^{(\ell_1)}|$, then it must be $v_i(X_i^{(\ell_1-1)}) < v_i(X_i^{(\ell_1)})$, through {Steps~\ref{step1}, \ref{step5}, \ref{step7}, \ref{step8}, or \ref{step9}}, if $|X_i^{(\ell_1)}| = 1$, and through {Steps~\ref{step4}, \ref{step5}, \ref{step7}, \ref{step8}, or \ref{step9}}, if $|X_i^{(\ell_1)}| = 2$.
Secondly, if $|X_i^{(\ell_1-1)}| \neq |X_i^{(\ell_1)}|$, then $v_i(X_i^{(\ell_1)})$ is strictly larger than all the preceding bundles of cardinality $|X_i^{(\ell_1)}|$ trivially (as there are no other such bundles).

\noindent Assuming the induction hypothesis for all indices up to and including $\kappa$, consider the bundle $X_i^{(\ell_{\kappa+1})}$. 
\begin{itemize}[leftmargin=*]
\item[-] First, assume that $|X_i^{(\ell_{\kappa+1}-1)}| = |X_i^{(\ell_{\kappa+1})}|$. 
Like in the base case, it must be $v_i(X_i^{(\ell_{\kappa+1}-1)}) < v_i(X_i^{(\ell_{\kappa+1})})$, through {Steps~\ref{step1}, \ref{step5}, \ref{step7}, \ref{step8}, or \ref{step9}}, if $|X_i^{(\ell_{\kappa+1})}| = 1$, and through {Steps~\ref{step4}, \ref{step5}, \ref{step7}, \ref{step8}, or \ref{step9}}, if $|X_i^{(\ell_{\kappa+1})}| = 2$.
By the definition of the $\ell_j$s, however, we have  $X_i^{(\ell_{\kappa})} = X_i^{(\ell_{\kappa}+1)} = \ldots =  X_i^{(\ell_{\kappa+1}-1)}$ and $v_i(X_i^{(\ell_\kappa)})$ is strictly larger than all the preceding bundles of cardinality $|X_i^{(\ell_\kappa)}|$ by the induction hypothesis. Thus, $v_i(X_i^{(\ell_{\kappa+1})})$ is strictly larger than all the preceding bundles of cardinality $|X_i^{(\ell_{\kappa+1})}|$.

\item[-] Next, assume that $|X_i^{(\ell_{\kappa+1}-1)}| \neq |X_i^{(\ell_{\kappa+1})}|$. If there is no preceding bundle of cardinality $|X_i^{(\ell_{\kappa+1})}|$ then the desired property of $v_i(X_i^{(\ell_{\kappa+1})})$ trivially holds. So, assume that this is not the case, and let $r$ be the largest superscript less than $\kappa$ for which $|X_i^{(\ell_r)}| = |X_i^{(\ell_{\kappa+1})}|$. We first assume that $|X_i^{(\ell_{\kappa+1})}| = 1$. This means that the subsequence $X_i^{(\ell_r)}, X_i^{(\ell_r +1)}, \ldots, X_i^{(\ell_{\kappa+1})}$  maintains cardinality $1$ for some iterations (until $X_i^{(\ell_{r +1} -1 )}$),  then moves \textit{once} from cardinality $1$ to cardinality $2$ (namely, from $X_i^{(\ell_{r +1} -1 )}$ to $X_i^{(\ell_{r +1})}$), possibly maintains cardinality $2$ for a number of iterations (until $X_i^{(\ell_{\kappa +1} -1 )}$) and, finally, moves once from cardinality $2$ back to cardinality $1$. In terms of value, this translates to {initially being constant, followed by \textit{potentially} decreasing but maintaining at least $2/3$ of the value, then by alternatingly being constant and possibly increasing, and, finally, by  increasing} by a factor strictly larger than $3/2$. That is, $v_i(X_i^{(\ell_{\kappa+1})}) > v_i(X_i^{(\ell_r)})$. Now, given that (a) $v_i(X_i^{(\ell_r)})$ is strictly larger than all the preceding bundles of cardinality $|X_i^{(\ell_r)}|$, (b) $X_i^{(\ell_{r})} = X_i^{(\ell_{r}+1)} = \ldots =  X_i^{(\ell_{r+1}-1)}$, and (c) the bundles $X_i^{(\ell_{r +1})}, \ldots, X_i^{(\ell_{\kappa+1}-1)}$ all have cardinality $2$, we conclude that $v_i(X_i^{(\ell_{\kappa+1})})$ is strictly larger than all the preceding bundles of cardinality $|X_i^{(\ell_{\kappa+1})}|$. The remaining case is $|X_i^{(\ell_{\kappa+1})}| = 2$. The argument is completely analogous with the only difference being that the subsequence $X_i^{(\ell_r)}, X_i^{(\ell_r +1)}, \ldots, X_i^{(\ell_{\kappa+1})}$  maintains cardinality $2$ for some iterations,  then moves once from cardinality $2$ to cardinality $1$, possibly maintains cardinality $1$ for a number of iterations and, finally, moves once from cardinality $1$ back to cardinality $2$. \renewcommand\qedsymbol{{\small $\boxdot$}}\qedhere
\end{itemize}
\end{proof}

\Cref{claim:no_repeated_bundles} implies that after every iteration at least one of the possible bundles is lost forever to at least one agent. As a result, the possible choices for all agents will be exhausted after at most $n m^2$ iterations of the main loop of the \PPA algorithm, so $n m^2 + 1$ is an immediate upper bound on the maximum number of iterations before the \PPA algorithm terminates. The polynomial running time now follows by observing that only one step of the loop is executed in each iteration and that every such step runs in polynomial time. The latter is straightforward to see, given that building the graphs $\Gr$ and $\Ge$, finding a cycle or all sources in them, and resolving {one cycle in $\Ge$ or all cycles in $\Gr$ (recall \cref{obs:all-cycle-resolution-running-time})} can all be done in polynomial time. 
\end{proof}

\section{2/3-EFX on Multigraphs or For a Small Number of Agents}\label{sec:multigraph-and-few-agents}

In this section, we provide our first main results, namely that $2/3$-EFX allocations exist and can be obtained in polynomial time for (i) multigraph value instances, or (ii) additive instances with up to $\numagents$ agents. 
We begin with some auxiliary results.

\subsection{Key Lemmata}\label{sec:key-lemmas} 

Recall that in \cref{sec:3pa} we mentioned that we plan to obtain a $2/3$-EFX partial allocation $\X^2$
that does not induce any critical goods. \Cref{lem:partialefx-to-efx-extension} below establishes that obtaining such an $\X^2$ suffices; we can then complete a $2/3$-EFX allocation via running the Envy Cycle Elimination algorithm of \citet{lipton2004approximately}. 
Given an envy graph, this algorithm extends a partial allocation one good at a time by allocating the next available good to an agent who is a source in the envy graph. If such an agent does not exist, the algorithm first resolves one or more envy cycles.
For completeness, we state below a parametric version of Envy Cycle Elimination. The algorithm is stated so that it can run on a subset of the goods, starting from a partial allocation.

\SetAlgorithmName{ALGORITHM}{Algorithm}

\setcounter{algocf}{1}
\begin{algorithm}[!ht]
	\DontPrintSemicolon 
        \KwData{A partial allocation $\X$ and its envy graph ${G}(\X)$.}
        \KwResult{A complete allocation $\X'$.} 
		\For{every $g \in \PP(\X)$ in arbitrary order}{
			\While{there is no source in $G({\X})$ }{
				Find a cycle $C$ in $G({\X})$ \;
                $\X \gets \CR(\X,{G}(\X),C)$
				}
			
			Let $i\in N$ be a source in $G({\X})$ \label{line:source} \;
            $X_{i} \gets X_{i}\cup \{g\}$\; 
		}
		\Return $\X$ \; 
	\caption{$\ECE(\X, G)$}\label{alg:ece}
\end{algorithm} 

In general, tie-breaking is not important for the properties of the algorithm we need and which are summarized in the next lemma. As usual, we may assume that any tie-breaking is done  lexicographically.

\begin{lemma}[Follows by \citet{lipton2004approximately}]\label{lemma:ece}
Fix an agent $i$, and let $X^{(\ell)}_i\!$ be the bundle assigned to $i$ at the end of the $\ell$-th iteration of the \emph{for} loop, for any $\ell \in \mathbb{N}$. Then, for $k\ge \ell$, it holds that  \,$v_i\big(X^{(k)}_i\big) \geq v_i\big(X^{(\ell)}_i\big)$. Furthermore, the algorithm terminates in polynomial time.
\end{lemma}

\Cref{lem:partialefx-to-efx-extension} below follows from results by \citet{markakis2023improved} as part of a general approximation framework; see also \citep{ANM2019,farhadi2021almost}. We adapt the statement to our terminology, and provide a short proof for completeness.

\begin{lemma}[\citet{markakis2023improved}]\label{lem:partialefx-to-efx-extension}
Let $\X$ be a $2/3$-EFX partial allocation that induces no critical goods. Then a $2/3$-EFX allocation $\tilde{\X}$ can be obtained from $\X$ in polynomial time. 
\end{lemma}

\begin{proof}
The complete allocation $\tilde{\X}$ is obtained from $\X$ by running the Envy Cycle Elimination algorithm of \citet{lipton2004approximately}; in particular, here we will assume that $\tilde{\X}$ is the output of $\ECE(\X, \allowbreak G(\X))$.
The polynomial running time follows by the fact that \ECE runs in polynomial time; see \Cref{lemma:ece}. 
For the $2/3$-EFX guarantee, consider any agent $i \in N$. We know that $v_i(\tilde{X}_i) \geq v_i(X_i)$ because during a run of \ECE, the utility of any agent never decreases; see \Cref{lemma:ece}. Now consider an arbitrary agent $j \in N \setminus\{i\}$. 

If no good is ever added to $\tilde{X}_j$ by \ECE, then $\tilde{X}_j = X_{j'}$, for some $j' \in N$ (possibly different from $j$). But agent $i$ was $2/3$-EFX towards $j'$ in $\X$, and as her value did not decrease, she is now $2/3$-EFX towards $j$ in $\tilde{X}$.

The remaining case is that $\tilde{X}_j$ has been augmented by \ECE. 
Let $g$ be the last good that was added to $\tilde{X}_j$ by \ECE. 
Right before good $g$ was added to $\tilde{X}_j$, the bundle belonged to some agent $j' \in N$, who was a source of the envy graph $G(\tilde{\X}^{t})$, where $\tilde{\X}^{t}$ is the allocation at that point.  For $\tilde{X}_{i}^{t}$, the corresponding bundle of agent $i$, we have that $v_i(\tilde{X}_{i}^{t}) \geq v_i(\tilde{X}^t_{j'}) = v_i(\tilde{X}_j \setminus \{g\})$. Recall that $\PP(\X)$ does not induce any critical goods; in particular, $g$ was not critical for $i$. This directly gives us $v_i(g) \le \frac12 v_i(X_i) \le \frac12 v_i(\tilde{X}_i^t)$. This, in turn,  implies that
\[
v_i(\tilde{X}_j) = v_i(\tilde{X}_j \setminus \{g\}) + v_i(g) \le v_i(\tilde{X}_{i}^{t}) +  \frac12 v_i(\tilde{X}_i^t) = \frac{3}{2} v_i(\tilde{X}_i^t) \le \frac{3}{2} v_i(\tilde{X}_i) \,. 
\]
Hence agent $i$ is $2/3$-EFX towards agent $j$  in $\tilde{X}$. 

Since both agents $i$ and $j$ were chosen arbitrarily, we conclude that $\tilde{\X}$ is $2/3$-EFX for all agents. 
\end{proof}

Our next lemma regards the allocation of critical goods, which will transform the partial allocation $\X^1$ obtained as the output of the \PPA algorithm into the partial allocation $\X^2$ that will in turn be transformed into a complete $2/3$-EFX allocation by \cref{lem:partialefx-to-efx-extension}. In particular, it concerns the special case in which every critical good is critical for at most one agent. This task is handled by the \UCG algorithm (Algorithm \ref{alg:uncontested-critical-goods}) that allocates all the remaining critical goods by carefully resolving cycles in the envy graph of the given partial allocation. {Note that although we allow Algorithm \ref{alg:uncontested-critical-goods} to take as input the type of envy graph to be used, in our algorithms it is only called for the standard envy graph.}

\SetAlgorithmName{ALGORITHM}{Algorithm}

\setcounter{algocf}{2}
\begin{algorithm}[!ht]
\DontPrintSemicolon
\caption{$\UCG(\X,\tilde{G})$} \label{alg:uncontested-critical-goods}
\SetKwComment{Comment}{/* }{ */}
\KwData{A partial allocation $\X$ satisfying the conditions (i), (ii) and (iii) of \Cref{lem:partial23-to-full-when-critical-for-one},  and its graph $\tilde{G}(\X)$. }
\KwResult{A partial allocation $\X'$ which does not induce any critical goods.\vspace{3pt}}
    $\X \gets \ACR(\X,\tilde{G}(\X))$\;
   \While(\\ \tcp*[h]{while there exists an agent with a critical good}){there exists $i \in N$ and $g_i \in \PP(\X)$ such that $v_i(g_i) > \frac{1}{2} v_i(X_i)$}
   {
   Let $s$ be a source of $\tilde{G}(\X)$ such that there exists some path $\Pi$ from $s$ to $i$ in $\tilde{G}(\X)$\;
   \If(\\ \tcp*[h]{if agent $i$ prefers the source's bundle augmented with her critical good}){$v_i(X_s \cup \{g_i\}) > v_i(X_i)$} 
   {
    $(X_j)_{j \in N: \exists (j,\ell) \in \Pi} \gets \PR(\X,\tilde{G},\Pi)$\;
    \tcp{every agent on the path except agent $i$ receives the bundle of her successor}
    $X_i \gets X_s \cup \{g_i\}$
    \tcp*{agent $i$ receives the bundle of the source plus her critical good $g_i$}
    }
    \Else
    {
    $X_s \gets X_s \cup \{g_i\}$
    \tcp*{agent $i$'s critical good is given to the source of the path}
    }
    $\X \gets \ACR(\X,\tilde{G}(\X))$
    \tcp*{Update $\X$ by eliminating all envy cycles in $\tilde{G}(\X)$}
   }
   \Return $\X$
\end{algorithm}

This algorithm is similar in nature to the Envy Cycle Elimination algorithm of \citet{lipton2004approximately}. The key difference is that we do not simply choose any source, but a source of a path that leads to the agent with the critical good that we aim to allocate in the current round. Then, if the addition of the good creates an envy cycle involving this agent, we resolve this particular cycle first.

\begin{lemma}\label{lem:partial23-to-full-when-critical-for-one}
    Let $\X$ be a partial allocation such that (i) $\X$ is $2/3$-EFX, (ii) each agent has at most one critical good, and (iii) no  good is critical for more than one agent. Then $\UCG(\X, G(\X))$ returns, 
    in polynomial time, a partial allocation $\X'$ that is $2/3$-EFX and in which no agent has a critical good. 
\end{lemma}

\begin{proof}

First we argue that $\X'$ can be obtained in polynomial time. To see this, recall that \ACR runs in polynomial time
{on the standard envy graph} (\cref{obs:all-cycle-resolution-running-time}), as does $\PR$. With respect to the while loop of $\UCG(\X,G)$, note that in each iteration, the number of goods in the pool $\PP(\X)$ decreases, as one critical good is allocated to some agent, and no goods are ever returned to $\PP(\X)$. This implies that the algorithm will terminate after at most $m$ iterations of the while loop. Clearly, each iteration here runs in polynomial time.

The fact that in $\X'$ there are no agents with critical goods is an immediate consequence of the algorithm \UCG terminating. We next argue that $\X'$ is $2/3$-EFX. 
We will prove by induction that properties (i), (ii), and (iii) of the statement of the lemma hold at every iteration
right before the condition of
the \emph{while} loop of \UCG is checked. Then, the properties will also hold at the termination of 
the execution of $\UCG(\X,G(\X))$; this is enough to establish that $\X'$ is $2/3$-EFX, as this is property (i).

For the base case, it is easy to see that the properties are satisfied before the first iteration of the loop: they hold for the initial partial allocation $\X$ by assumption, and they continue to hold after running the $\ACR(\X,{G}(\X))$ subroutine, which only swaps bundles between agents, never reducing their values. 

Now consider the $k$-th iteration of the while loop. Let $\X$ and $\tilde{\X}$ be the partial allocations right before and right after the $k$-th iteration respectively. After the end of the iteration, we have that $v_i(\tilde{X}_i)\geq v_i(X_i)$ for all $i \in N$, and that $\PP(\tilde{\X}) \subseteq \PP(\X)$. By the induction hypothesis, it follows that properties (ii) and (iii) are trivially satisfied in $\tilde{\X}$. 
To establish property (i), notice that for every agent $j \in N\setminus\{i\}$, it holds that $v_j(X_s) \leq v_j(X_j)$, since $s$ is a source of $G(\X)$, and $v_j(g_i) \leq \frac{1}{2}v_j(X_j)$, since $g_i$ is critical only for $i$. By adding the two inequalities, we have that agent $j$ is $2/3$-EFX towards any agent who receives the bundle $X_s \cup \{g_i\}$ in $\tilde{\X}$ (either agent $i$ or agent $s$). Recall that by the induction hypothesis, $\X$ right before the $k$-th iteration was $2/3$-EFX. Since all other bundles (i.e., all bundles besides $X_s$) do not change as sets and they are possibly exchanged so that no agent's value is decreased, we conclude that $\tilde{\X}$ is also $2/3$-EFX.
\end{proof}

\subsection{2/3-EFX for Multigraph Value Instances}

Given the machinery that we have set up so far, it is fairly easy to show that $2/3$-EFX allocations can be achieved in  multigraph value instances, also in polynomial time.

\begin{algorithm}[!h]
	\DontPrintSemicolon
	\caption{$\MGA(N, (v_i)_{i\in N}, M)$} \label{alg:multigraph}
	\SetKwComment{Comment}{/* }{ */}
	\KwData{A multigraph value instance.}
	\KwResult{A $2/3$-EFX allocation $\X$. \vspace{3pt}}
	Let $\X^0$ be an arbitrary partial allocation where each bundle has cardinality $1$\;
	$\X^1 \gets \PPA((v_i)_{i\in N}, \X^0)$\;
	$\mathcal{C} \gets \{g\in \PP(\X^1) : \text{ there are distinct } i, j\in N \text{ so that } g \text{ is critical for both } i, j\}$\; 
	\tcp{the set of critical goods which are critical for more than one agent}
	Let $s$ be a source of $\Ge(\X^1)$\; 
	\tcp{at least one such $s$ exists, and has $|X^1_s|>1$ and no critical goods, by \cref{lem:3PA-properties}}
	$\tilde{X}_s \gets X^1_s \cup \mathcal{C}$ \tcp*{we just give the whole $\mathcal{C}$ to $s$}
	$\tilde{X}_i \gets X^1_i$ for any other $i\in N\setminus \{s\}$\;
	$\X^2 \gets \UCG(\tilde{\X}, G(\tilde{\X}))$
	\tcp*{we handle the remaining critical goods}
	$\X \gets \textsc{EnvyCycleElimination}(\X^2, G(\X^2))$
	\tcp*{we complete the allocation}
	\Return $\X$
\end{algorithm}

\begin{theorem}\label{thm:multigraph}
Let $(N, (v_i)_{i\in N}, M)$ be a multigraph value instance. Then $\MGA$ constructs a $2/3$-EFX allocation in polynomial time.
\end{theorem}

\begin{proof}
The partial allocation $\X^0$ where each agent receives a single good arbitrarily trivially satisfies { \cref{prop:a,prop:b}, i.e., it is a seed allocation}.  We use this allocation as input to the \PPA algorithm (Algorithm \ref{alg:3PA}) to obtain a partial allocation $\X^1$; by \cref{lem:3PA-polytime}, $\X^1$ can be obtained in polynomial time. By \cref{lem:3PA-properties}, $\X^1$ also satisfies \cref{prop:a,prop:b,prop:c,prop:d,prop:e},  and either it is a complete allocation or its enhanced graph $\Ge(\X^1)$ has at least one source $s$, and this (as any other source of $\Ge(\X^1)$) {holds a bundle of cardinality $2$} and does not have any critical goods. 
{Further, notice that, by \Cref{def:Ge}, any source $s$ in  $\Ge(\X^1)$ is also a source in $\Gr(\X^1)$ and since $|{X}^1_{s}| =  2$, \Cref{def:Gr} implies that $s$ is  a source in the standard envy graph $\G(\X^1)$ as well; hence, $v_i(X^1_s) \le v_i(X^1_i)$ for any $i\in N$.}

Next, we consider the set $\mathcal{C}$ of critical goods which are critical for more than one agent. Since this is a multigraph value instance, any good in $\mathcal{C}$ is critical for \textit{exactly} two agents, as only two agents could have positive value for it. We then assign every good $g \in \mathcal{C}$ to $s$, i.e., $\tilde{X}_s = X^{1}_s \cup \mathcal{C}$. Finding $s$, $\mathcal{C}$, and constructing $\tilde{\X}$ obviously can be done in polynomial time. 
We will show the following claim about $\tilde{\X}$: 

\begin{claim}\label{claim:multigraph_instances}
    The partial allocation $\tilde{\X}$ satisfies the conditions in the statement of \Cref{lem:partial23-to-full-when-critical-for-one}.
\end{claim}

\begin{proof}[Proof of \cref{claim:multigraph_instances}]
First notice that condition \textit{``(iii) no   good is critical for more than one agent''} is directly satisfied by the removal of $\mathcal{C}$ from $\PP(\tilde{\X})$. Further, the condition \textit{``(ii) each agent has at most one critical good''} is met because it holds for $\X^1$ (by \cref{prop:e}) and, by moving to $\tilde{\X}$, we have $\PP(\tilde{\X}) \subseteq \PP({\X^1})$ and no agent had her value reduced.
It remains to show that $\tilde{\X}$ is $2/3$-EFX (i.e., condition \textit{``(i)''} of \Cref{lem:partial23-to-full-when-critical-for-one}).
Looking at the partial allocation $\X^1$, first notice that for any good $g \in \mathcal{C}$ corresponding to an edge $(j,\ell)$, $g$ must be critical for both agents $j$ and $\ell$, but it has no value for anyone else, i.e., for any agent $i \in N\setminus\{j,\ell\}$, we have $v_i(g)=0$. 

Now consider any agent $i \in N$ with a critical good $g_i \in \mathcal{C}$. By \cref{prop:d} of $\X^1$, it must be that $|X^1_i| = 1$. Recall that the source agent $s$ had no critical goods. So, $i$ must be distinct from $s$ and hence $\tilde{X}_i = X^1_i$. Also, by \cref{prop:e} of $\X^1$, every other good $g \in \mathcal{C}\setminus\{g_i\}$ is not critical for agent $i$. This implies that $v_i(g) = 0$ for all $g \in \mathcal{C} \setminus\{g_i\}$. From this, we obtain that $v_i(\mathcal{C})=v_i(g_i) \leq \frac{2}{3}v_i({X}^1_i)$, where the last inequality follows again by \cref{prop:e} of $\X^1$. Also, since $s$ is a source in $\Ge(\X^1)$, we have $v_i(X^1_s) < \frac{2}{3} v_i(X^1_i)$, by the definition of $\Ge(\X^1)$ (\Cref{def:Ge}) and the cardinalities of $X^1_s, X^1_i$. 
Therefore, 
\[v_i(\tilde{X}_s)=v_i(X^1_s \cup \mathcal{C}) < \frac{2}{3}\, v_i(X^1_i) +\frac{2}{3}\, v_i(X^1_i) < \frac{3}{2}\, v_i(X^1_i) = \frac{3}{2}\, v_i(\tilde{X}_i)\,,\] 
where the last equality follows from the fact that $\tilde{X}_i = X^1_i$. So, agent $i$ is $2/3$-EFX towards agent $s$ in $\tilde{\X}$. Given that she was $2/3$-EFX towards anyone else in ${\X}^1$ (by \cref{prop:a,prop:b}), we conclude that agent $i$ is $2/3$-EFX towards everyone in $\tilde{\X}$.

Next, consider any agent $i \in N$ with no critical goods in $\mathcal{C}$ (where $i$ could be $s$ here). By the discussion about the value of goods in $\mathcal{C}$ above, we have $v_i(\mathcal{C})=0$. Moreover, by the choice of $s$, the definition of $\Ge(\X^1)$, and the fact that $|X^1_s| > 1$, we have $v_i(X^1_s) \leq v_i(X^1_i)$ (independently of the cardinality of $X^1_i$). Thus, we get
$v_i(\tilde{X}_s)=v_i(X^1_s \cup \mathcal{C}) \le v_i(X^1_i) + 0 \leq  \frac{3}{2}\, v_i(\tilde{X}_i)$.
So, again, agent $i$ is $2/3$-EFX towards $s$ in $\tilde{\X}$. As she was $2/3$-EFX towards anyone else in ${\X}^1$ (by \cref{prop:a,prop:b}), agent $i$ is $2/3$-EFX towards everyone in $\tilde{\X}$.
\renewcommand\qedsymbol{{\small $\boxdot$}}
\end{proof}

Given \cref{claim:multigraph_instances}, \Cref{lem:partial23-to-full-when-critical-for-one} implies that the partial allocation $\X^2$ returned by the \UCG algorithm in polynomial time is $2/3$-EFX and it does not induce any critical goods. By \cref{lem:partialefx-to-efx-extension}, we can then  obtain a complete $2/3$-EFX allocation $\X$ in polynomial time via \textsc{EnvyCycleElimination}. 
\end{proof}

\subsection{2/3-EFX for up to \numagents Agents}

Next we show the existence of $2/3$-EFX allocations for \numagents agents, and that they can be computed in polynomial time.  {This is where the significance of Step~\ref{step9} of \PPA will become apparent. The general idea of the main algorithm for a small number of agents, \FAA (Algorithm \ref{alg:few_agents}),} is very similar to \MGA, yet the process of handling the critical goods is more refined here.

\begin{algorithm}[!ht]
	\DontPrintSemicolon
	\caption{$\FAA(N, (v_i)_{i\in N}, M)$} \label{alg:few_agents}
	\SetKwComment{Comment}{/* }{ */}
	\KwData{An additive instance with at most \numagents agents.}
	\KwResult{A $2/3$-EFX allocation $\X$.\vspace{3pt}}
    Let $\X^0$ be an arbitrary partial allocation where each bundle has cardinality $1$\;
	$\X^1 \gets {\PPA((v_i)_{i\in N}, \X^0)}$\;
	$\mathcal{C} \gets \{g\in \PP(\X^1) : \text{ there are distinct } i, j\in N \text{ so that } g \text{ is critical for both } i, j\}$\; 
	\tcp{the set of critical goods which are critical for more than one agent}
    \If(\\ \tcp*[h]{recall that $|X^1_{s_1}|>1$, $|X^1_{s_2}|>1$, and $s_1, s_2$ have no critical goods, by \cref{lem:3PA-properties}\vspace{1pt}}){$|\mathcal{C}| = 2$ (say, $\mathcal{C} = \{g_1, g_2\}$) and there are at least two sources, $s_1, s_2$, in $\Ge(\X^1)$}{
	       $\tilde{X}_{s_1} \gets X^1_{s_1} \cup \{g_1\}$ \tcp*{we split $\mathcal{C}$ among $s_1, s_2$}
          $\tilde{X}_{s_2} \gets X^1_{s_2} \cup \{g_2\}$\;
	       $\tilde{X}_i \gets X^1_i$ for any other $i\in N\setminus \{s_1, s_2\}$\;        
    }
    \Else{
          Let $s$ be a source of $\Ge(\X^1)$\; 
	       \tcp{at least one such $s$ exists, and has $|X^1_s|>1$ and no critical goods, by \cref{lem:3PA-properties}}
	       $\tilde{X}_s \gets X^1_s \cup \mathcal{C}$ \tcp*{we give the whole $\mathcal{C}$ to $s$}
	       $\tilde{X}_i \gets X^1_i$ for any other $i\in N\setminus \{s\}$\;
    }
	$\X^2 \gets \UCG(\tilde{\X}, G(\tilde{\X}))$
	\tcp*{we handle the remaining critical goods}
	$\X \gets \textsc{EnvyCycleElimination}(\X^2, G(\X^2))$
	\tcp*{we complete the allocation}
	\Return $\X$
\end{algorithm}

\begin{theorem}\label{thm:7_agents}
Let $(N, (v_i)_{i\in N}, M)$ be any additive instance with $n\le \numagents$. Then \FAA constructs a $2/3$-EFX allocation in polynomial time.
\end{theorem}

\begin{proof}
{Arguing like in the first part of the proof of \cref{thm:multigraph}, we have that the partial allocation $\X^1$ is obtained in polynomial time, it satisfies \cref{prop:a,prop:b,prop:c,prop:d,prop:e}, and either it is a complete allocation or its enhanced graph $\Ge(\X^1)$ has at least one source $s$ and this (as any other source) has a bundle of size $2$. Moreover, for any $i\in N$, it holds that $v_i(X^1_s) \le v_i(X^1_i)$.}
Again we consider the set $\mathcal{C}$ of critical goods which are critical for more than one agent, but now we need to argue about its cardinality.

If there are at least two sources, $s_1 ,s_2$, in the enhanced envy graph $\Ge(\X^1)$, then 
there are at most $5$ agents that are not sources and may have a critical good. (Indeed, recall that since {$|{X}^1_{s_1}| = |{X}^1_{s_2}| = 2$, by \cref{prop:d} of $\X^1$}, $s_1 ,s_2$ cannot have critical goods.) 
This immediately implies that $|\mathcal{C}|\leq 2$. The case where $|\mathcal{C}|\leq 1$ is handled in the \emph{else} part of \FAA and is analyzed later in the proof. For now, we consider the case where $|\mathcal{C}| = 2$, say $\mathcal{C} = \{g_1, g_2\}$, which is handled in the \emph{if} part of \FAA.
To define the next partial allocation, $\tilde{\X}$, we assign one good from $\mathcal{C}$ to each of $s_1 ,s_2$, i.e., $\tilde{X}_{s_i} = X^{1}_{s_i} \cup \{g_i\}$, for $i\in\{1, 2\}$. Finding $s_1 ,s_2$, $\mathcal{C}$, and constructing $\tilde{\X}$ can be clearly done in polynomial time. 

\begin{claim}\label{claim:few_agents_if}
    If the partial allocation $\tilde{\X}$ is constructed in the \emph{if} case, it satisfies the conditions in the statement of \Cref{lem:partial23-to-full-when-critical-for-one}.
\end{claim}

\begin{proof}[Proof of \cref{claim:few_agents_if}]
The conditions \textit{``(ii) each agent has at most one critical good''} and \textit{``(iii) no  good is critical for more than one agent''} are satisfied exactly like in the proof of \cref{claim:multigraph_instances}.
We only need to show condition \textit{``(i)''}, i.e., that $\tilde{\X}$ is $2/3$-EFX.

First consider any agent $i \in N$ with a critical good in $\mathcal{C}$ under the partial allocation $\X^1$. We assume that $i$'s critical good is $g_1$, the other case being completely symmetric.  
By \cref{prop:d} of $\X^1$, we have $|X^1_i| = 1$. 
Also, by \cref{prop:e} of $\X^1$, {we have that $v_i(g_1) \leq \frac{2}{3}v_i({X}^1_i)$ and that $g_2$ is not critical for $i$; thus, $v_i(g_2) \leq v_i(g_1) \leq \frac{2}{3}v_i({X}^1_i)$.}
Recall that $s_1 ,s_2$ had no critical goods. So, $i$ must be distinct from both $s_1$ and $s_2$ and hence $\tilde{X}_i = X^1_i$. Finally, since $s_1 ,s_2$ are sources in $\Ge(\X^1)$, we have $v_i(X^1_{s_j}) < \frac{2}{3} v_i(X^1_i)$, $j\in\{1, 2\}$, by the definition of $\Ge(\X^1)$ and the cardinalities of $X^1_{s_1}, X^1_{s_2}, X^1_i$. We now have, for $j\in\{1, 2\}$,
\[v_i(\tilde{X}_{s_j})=v_i(X^1_{s_j} \cup \{g_j\}) < \frac{2}{3}\, v_i(X^1_i) +\frac{2}{3}\, v_i(X^1_i) < \frac{3}{2}\, v_i(X^1_i) = \frac{3}{2}\, v_i(\tilde{X}_i)\,.\]
So, agent $i$ is $2/3$-EFX towards both $s_1$ and $s_2$ in $\tilde{\X}$. Given that she was $2/3$-EFX towards anyone else in ${\X}^1$ (by \cref{prop:a,prop:b}), we conclude that agent $i$ is $2/3$-EFX towards everyone in $\tilde{\X}$.

Next, consider any agent $i \in N$ with no critical goods in $\mathcal{C}$ (where $i$ could be $s_1$ or $s_2$ here). This immediately implies that $\max_{j\in\{1,2\}}v_i(g_j) \le \frac{1}{2} v_i(X^1_i) \le \frac{1}{2} v_i(\tilde{X}_i)$. 
Moreover, for $j\in\{1,2\}$, by the choice of $s_j$, the definition of $\Ge(\X^1)$, and the fact that $|X^1_{s_j}| > 1$, we have $v_i(X^1_{s_j}) \leq v_i(X^1_i)$ (independently of the cardinality of $X^1_i$). Thus, we get
\[v_i(\tilde{X}_{s_j})=v_i(X^1_{s_j} \cup \{g_j\}) \le v_i(X^1_i) + \frac{1}{2}\, v_i(X^1_i) \leq  \frac{3}{2}\, v_i(\tilde{X}_i)\,.\]
So, again, agent $i$ is $2/3$-EFX towards $s_j$, for $j \in \{1,2\}$, in $\tilde{\X}$.   As she was $2/3$-EFX towards anyone else in ${\X}^1$ (by \cref{prop:a,prop:b}), agent $i$ is $2/3$-EFX towards everyone in $\tilde{\X}$.
\renewcommand\qedsymbol{{\small $\boxdot$}}
\end{proof}

Similarly to how we argued above, if there is a single source, $s$, in $\Ge(\X^1)$, then
there are at most $6$ agents that are not sources and may have a critical good. {Since each non-source agent has at most one critical good by \cref{prop:e} and a good belongs to $\mathcal{C}$ only if it is critical for at least two agents,}  this implies that $|\mathcal{C}|\leq 3$. 
Both this case and the case where we may have multiple sources  but $|\mathcal{C}|\leq 1$ are handled together in the \emph{else} part of \FAA.
There, to define the partial allocation $\tilde{\X}$, we assign the whole $\mathcal{C}$ to $s$ (which may be the unique source or any one of multiple sources). That is, $\tilde{X}_{s} = X^{1}_{s} \cup \mathcal{C}$. Finding $s$, $\mathcal{C}$, and constructing $\tilde{\X}$ can be done in polynomial time.

\begin{claim}\label{claim:few_agents_else}
    If the partial allocation $\tilde{\X}$ is constructed in the \emph{else} case, it satisfies the conditions in the statement of \Cref{lem:partial23-to-full-when-critical-for-one}.
\end{claim}

\begin{proof}[Proof of \cref{claim:few_agents_else}]
Again, the conditions \textit{``(ii) each agent has at most one critical good''} and \textit{``(iii) no good is critical for more than one agent''} are satisfied exactly like in the proof of \cref{claim:multigraph_instances}.
We only need to show condition \textit{``(i)''}, i.e., that $\tilde{\X}$ is $2/3$-EFX.

First consider any agent $i \in N$ with a critical good $g_i$ (not necessarily in $\mathcal{C}$) under the partial allocation $\X^1$. 
By \cref{prop:d} of $\X^1$, we have $|X^1_i| = 1$. 
Also, since $s$ is a source in $\Ge(\X^1)$, we have $v_i(X^1_{s}) < \frac{2}{3} v_i(X^1_i)$, by the definition of $\Ge(\X^1)$ and the cardinalities of $X^1_{s}, X^1_i$.

Since no step of the {\PPA} algorithm can be applied on $\X^1$, the condition of Step~\ref{step3} fails, i.e., for any two goods $g,g'\in {\PP(\X^1)}$, $v_i(\{g,g'\}) < \frac 23 v_i(X^1_i)$. 
However, as $g_i$ is critical for $i$, we have $v_i(g_i)>\frac 12 v_i(X^1_i)$. It then follows that for any {$g\in\PP(\X^1)\setminus\{g_i\}$}, it must be $v_i(g) < \frac 16  v_i(X^1_i)$, as $v_i(\{g,g_i\}) < \frac 23 v_i(X^1_i)$. 
Now it is not hard to bound $v_i(\mathcal{C})$. If $g_i \notin \mathcal{C}$, then $v_i(\mathcal{C}) < |\mathcal{C}|\, \frac 16 \, v_i(X^1_i) \le \frac 12 v_i(X^1_i)$. If $g_i \in \mathcal{C}$, it could be that $|\mathcal{C}| \le 2$, and thus $v_i(\mathcal{C}) <   \frac 23 v_i(X^1_i)$, by the discussion above, or it could be that $|\mathcal{C}| = 3$, say $\mathcal{C} = \{g_i, g, g'\}$, and thus $v_i(\mathcal{C}) =  v_i(\{g_i,g\}) +v_i(g') < \frac 23 v_i(X^1_i) + \frac 16 v_i(X^1_i) = \frac 56 v_i(X^1_i)$. In any case, $v_i(\mathcal{C}) < \frac 56 v_i(X^1_i)$. To show that $\tilde{\X}$ is $2/3$-EFX, we distinguish two cases. \medskip

\noindent\textbf{Case 1 ($|\mathcal{C}| = 3$):} Then any agent $i\in N\setminus\{s\}$ has a critical good in $\PP(\X^1)$. {Indeed, $|\mathcal C|=3$ forces at least $6$ agents to have a critical good; as $n\le 7$ and sources have no critical goods, $s$ is the unique source and every $i\in N\setminus\{s\}$ has a critical good.} Therefore,
\[ v_i(\tilde{X}_s)=v_i(X^1_s)+ v_i(\mathcal{C}) < \frac 23 v_i(X^1_i) + \frac 56 v_i(X^1_i) = \frac 32 v_i(X^1_i)\,. \]
So, agent $i$ is $2/3$-EFX towards $s$ in $\tilde{\X}$, and as she was $2/3$-EFX towards anyone else in ${\X}^1$ (by \cref{prop:a,prop:b}), she is $2/3$-EFX towards everyone in $\tilde{\X}$. It is also straightforward that agent $s$ is $2/3$-EFX towards everyone in $\tilde{\X}$, as her value did not decrease and every other agent maintained the bundle she had in ${\X}^1$. \medskip

\noindent\textbf{Case 2 ({$|\mathcal{C}| = 2$}):}
Showing that any agent $i\in N\setminus\{s\}$ who \textit{has} a critical good in $\PP(\X^1)$ is $2/3$-EFX towards everyone in $\tilde{\X}$ is almost the same as in Case 1, except the upper bound for $v_i(\mathcal{C})$ is now  {a non-strict} $\frac 23 v_i(X^1_i)$ rather than {a strict} $\frac 56 v_i(X^1_i)$ because the condition of Step~\ref{step3} fails for $\X^1$. 

The interesting case is that of agents who have no critical goods, so let $i\in N$ be such an agent (possibly $i$ is $s$ herself). 
We have $v_i(X^1_i)\geq v_i(X^1_s)$, since $s$ is a source of $\Ge(\X^1)$. Also, since Steps \ref{step3} and \ref{step4} of {\PPA} fail to apply on $\X^1$, we have $v_i(X^1_i)\geq v_i(\mathcal{C})$. 
{Next, notice that $s$ must be the \textit{unique} source of $\Ge(\X^1)$, as otherwise the \emph{else} case would not have been executed. This means that for any $i\in N$ there is a path $\Pi = (s, \ldots, i)$ in $\Ge(\X^1)$. Therefore, }
since {Step \ref{step9} of \PPA} fails to apply on $\X^1$, we have that for any goods $g \in X^1_s$ and $g' \in \mathcal{C}$, $v_i(X^1_i)\geq v_i(\{g,g'\})$. 
Overall, for any two goods $g,g'\in \tilde{X}_s$, $v_i(X^1_i)\geq v_i(\{g,g'\})$. This in turns means that for any three goods $g,g',g''\in \tilde{X}_s$, $v_i(X^1_i)\geq \frac 23 v_i(\{g,g', g''\})$. That is, for any good $g\in \tilde{X}_s$,  
\[v_i(\tilde{X}_i) \ge v_i(X^1_i)\geq \frac 23 v_i(\tilde{X}_s\setminus\{g\})\,,\]
and, thus, agent $i$ is $2/3$-EFX towards $s$ in $\tilde{\X}$, and as she was $2/3$-EFX towards anyone else in ${\X}^1$ (by \cref{prop:a,prop:b}), she is $2/3$-EFX towards everyone in $\tilde{\X}$. \medskip

\renewcommand\qedsymbol{{\small $\boxdot$}} 

{
\noindent\textbf{Case 3 ($|\mathcal{C}| \le 1$):} If $\mathcal{C} = \emptyset$, then $\tilde{\X} = \X^1$; thus $\tilde{\X}$ is $2/3$-EFX and there is nothing more to show. So assume that $\mathcal{C} = \{g\}$. Again, showing that any agent $i\in N\setminus\{s\}$ who \textit{has} a critical good in $\PP(\X^1)$ is $2/3$-EFX towards everyone in $\tilde{\X}$ is almost the same as in Case 1, except the  upper bound for $v_i(\mathcal{C})$ is now a non-strict $\frac 23 v_i(X^1_i)$ rather than a strict $\frac 56 v_i(X^1_i)$ because the conditions of Steps~\ref{step3} and \ref{step7} fail for $\X^1$. So, next we consider agents who have no critical goods; let $i\in N$ be such an agent (where $i$ could be $s$). Since $s$ is a source in $\Ge(\X^1)$ we have $v_i(X^1_s) \le v_i(X^1_i)$ and, thus, 
\[ v_i(\tilde{X}_s)=v_i(X^1_s)+ v_i(g) \le v_i(X^1_i) + \frac 12 v_i(X^1_i) = \frac 32 v_i(X^1_i)\,, \]
where, of course, we used that $g$ is not critical for agent $i$.
}
\end{proof}

Given \Cref{claim:few_agents_if,claim:few_agents_else}, \Cref{lem:partial23-to-full-when-critical-for-one} implies that the partial allocation $\X^2$ returned in polynomial time by the \UCG algorithm is $2/3$-EFX and it does not induce any critical goods. By \cref{lem:partialefx-to-efx-extension}, we  then  obtain a complete $2/3$-EFX allocation $\X$ in polynomial time via \textsc{EnvyCycleElimination}. 
\end{proof}

\section{2/3-EFX for $3$-Value Instances}\label{sec:three-values}

The last setting we consider is $3$-value instances with values $a\ge b \ge c\ge 0$. We assume that $a > b> c$, as, otherwise, the problem is either trivial or reduces to the $2$-value setting where it is known that EFX allocations can be computed in polynomial time \citep{amanatidis2021maximum}. Therefore, as we mentioned in \cref{sec:preliminaries}, it is without loss of generality to set $a=1$ and rescale $b$ and $c$ appropriately. The main result here is the following theorem about our \TVA algorithm (Algorithm \ref{alg:3_values}).

\begin{theorem}\label{thm:3-value-instances}
Let $(N, (v_i)_{i\in N}, M)$ be any $3$-value instance. Then \TVA constructs a $2/3$-EFX allocation in polynomial time.
\end{theorem}

\noindent For both the proof of the theorem and the description of our algorithm, we consider three cases depending on the range of the values of $b$ and $c$, namely:
\begin{threecases}
\item $b \leq {1}/{2}$;\label{case1}
\item $b > {1}/{2}$ and $b+c \geq {2}/{3}$;\label{case2}
\item $b > {1}/{2}$ and $b+c < {2}/{3}$.\label{case3}
\end{threecases}

{It turns out that \cref{case1,case2} can be completely handled by our existing machinery.}
\cref{case3}, however, turns out to be much more challenging, and requires significant refinements to our general approach. 
The variant of \PPA we design here, \PPAplusplus {(Algorithm \ref{alg:3PAplusplus} in \cref{sec:refined3PA})}, is rather intricate, as is the proof about the properties it maintains. 
So, we have all the components needed to state the complete algorithm for $3$-value instances, called \TVA (Algorithm \ref{alg:3_values}), except for \PPAplusplus, the modification of \PPA that handles \cref{case3}. We state the algorithm here and we defer the description of \PPAplusplus to \cref{subsec:teras} that is dedicated to proving \cref{thm:3-value-instances} for \cref{case3}.
It is worth noting that, given the {two variants, \PPA and \PPAplusplus}, we can directly produce a partial allocation without any critical goods in any of the three cases. That is, for $3$-value instances, we do not rely on an analog of \UCG or on other subroutines for dealing with critical goods, but only run \ECE on the output of {\PPA or \PPAplusplus}.

\begin{algorithm}[!h]
	\DontPrintSemicolon
	\caption{$\TVA(N, (v_i)_{i\in N}, M)$} \label{alg:3_values}
	\SetKwComment{Comment}{/* }{ */}

	\KwData{A $3$-value instance with values $a=1 > b > c \ge 0$.}
	\KwResult{A $2/3$-EFX allocation $\X$.\vspace{3pt}}
    Let $\X^0$ be an arbitrary partial allocation where each bundle has cardinality $1$\;

    \If(\tcp*[f]{{\cref{case1,case2}}}){{$b \le \frac{1}{2}$ \text{\textbf{\emph{or}}} $\left(b > \frac{1}{2}\right.$ and\, $\left.b+c \geq \frac{2}{3}\right)$}\vspace{2pt}}{
	    $\X^1 \gets \PPA((v_i)_{i\in N}, \X^0)$        
    }
    \Else(\tcp*[f]{\cref{case3}}){
        \If(\tcp*[f]{the $c = 0$ case is reduced to the $c >0$ case below}){$c = 0$\vspace{2pt}}{
	        Update $(v_i)_{i\in N}$ so that $c \gets \varepsilon$ for sufficiently small $\varepsilon >0$\;  
         \tcp{roughly, $\varepsilon$ will be as small as the smallest possible (weighted) difference between two bundles under the original $(v_i)_{i\in N}$ scaled down by $3m$; see \cref{lem:zero_c} } 
        }
        $\X^1 \gets \PPAplusplus((v_i)_{i\in N}, \X^0)$ \tcp*{handles \cref{case3} for $c >0$}
    }
	$\X \gets \textsc{EnvyCycleElimination}(\X^1, G(\X^1))$
	\tcp*{we just complete the allocation (no critical goods)}
	\Return $\X$

\end{algorithm}

\begin{proof}[Proof of Cases \textnormal{\ref{case1}} and \textnormal{\ref{case2}} of \cref{thm:3-value-instances}]
Let $\X^1$ be the partial allocation outputted by {\PPA} when initiated with an allocation $\X^0$ of size at least $1$ and at most $2$ which satisfies \cref{prop:a,prop:b} as shown in \TVA (Algorithm \ref{alg:3_values}). By \cref{lem:3PA-properties}, $\X^1$ satisfies \cref{prop:a,prop:b}. 
We will argue that, in both cases, there are no critical goods induced by $\X^1$. Therefore, running $\ECE$ on $\X^1$ results in a $2/3$-EFX allocation, by \cref{lem:partialefx-to-efx-extension}. The polynomial running time directly follows from the running time of {\PPA and \ECE (\cref{lem:3PA-properties,lemma:ece}).}

A simple, but very useful, observation that we will be using in both cases below is that if an agent $i$ has a critical good $g_i$, by \cref{prop:d,prop:e} of $\X^1$, it holds that $|X_i^1|=1$ and $1/2 < v_i(g_i) \le 2/3$. This implies that $v_i(X_i^1)=1$ 
or $v_i(X_i^1)=b$;
if $v_i(X_i^1)=c$, the agent cannot have a critical good. \medskip

\noindent\textbf{\cref{case1}} [$b \leq {1}/{2}$]: Suppose, towards a contradiction, that there exists some agent $i \in N$ that has a critical good $g \in \PP(\X^1)$. 
{By \cref{prop:e}, we know that $1/2 < v_i(g) \le 2/3$.} 
Consider first the case where $v_i(X^1_i)=1$. In this case, however, there {is no good $g' \in \PP(\X^1)$ such that $1/2 < v_i(g') \le 2/3$, contradicting the choice of $g$.} 
Next, consider the case $v_i(X_i^1)=b$. \cref{prop:e} ensures that $v_i(g)=c$, as otherwise we would have $v_i(g) > \frac{2}{3}v_i(X_i^1)$. Now consider a source $s$ of the enhanced envy graph $\Ge(\X^1)$; such a source exists by \cref{lem-3PA-source-enhanced} (since here we assumed $\PP(\X^1)\neq \emptyset$) and recall that $|{X_s^1}|=2$. By virtue of $s$ being a source {in $\Ge(\X^1)$ it is also a source in $G(\X^1)$ (as shown in the proof of \cref{thm:multigraph}); hence,} we have $v_i(X_i^1) \geq v_i(X_s^1)$. 
This is only possible if $v_i(X_s^1) = 2c$ or if $v_i(X_s^1)=b+c$ and $c=0$. In the latter case, $g$ is not critical for agent $i$. In the former case, the inequality $v_i(X_i^1) \geq v_i(X_s^1)$ implies that $b \geq 2c$, which means that $g$ is not critical for agent $i$. Both cases contradict the choice of $g$. \medskip

\noindent\textbf{\cref{case2}} [$b > {1}/{2}$ and $b+c \geq {2}/{3}$]:
Again, suppose, towards a contradiction, that there exists some agent $i \in N$ that has a critical good $g \in \PP(\X^1)$. 
{Since  the execution of \PPA  terminated  at Step \ref{step10} and the condition of Step \ref{step8} was not triggered,}  
we know that {$\Ge(\X^1)$ is acyclic}. 
Let $\Pi=(s, \ldots,i)$ be a maximal path in {$\Ge(\X^1)$} leading to agent $i$. Due to maximality, the initial vertex $s$ of $\Pi$ must be a source. {Note that if $s$ is $i$ itself, $i$ would be a source in $\Gr(\X^1)$ as well; since} Step \ref{step6} did not run in the last iteration of {\PPA, however}, $s$ must be distinct from $i$.
Next, let $g_s \in X^1_s$. Obviously $v_i(g_s) \geq c$. 
We will argue that $v_i(\{g_s,g\}) \geq \frac23 v_i(X^1_i)$. 
By the observation made before \cref{case1}, we have that $v_i(X_i^1)=1$ and $v_i(g)\in \{b, c\}$,
or $v_i(X_i^1)=b$ and $v_i(g)=c$.
If $v_i(g) = b$, then $v_i(\{g_s,g\}) \geq b+c\geq\frac{2}{3} \geq \frac23 v_i(X^1_i)$, since $v_i(X^1_i) \leq 1$. If $v_i(g) = c$, {the criticality of $g$ implies that} $c > \frac12 v_i(X^1_i)$ and therefore $v_i(\{g_s,g\}) \geq 2c > {v_i(X^1_i)}$. In any case, $v_i(\{g_s,g\}) \geq \frac23 v_i(X^1_i)$, directly contradicting the fact that {Step \ref{step9} did not run in the last iteration of \PPA}.
\end{proof}

\subsection{Instances with $b > {1}/{2}$ and $b+c < {2}/{3}$}\label{subsec:teras}

As we mentioned earlier, our general technique needs to be significantly refined and expanded to deal with \cref{case3}. To this end, we develop a rather complicated refinement of the {\PPA} algorithm (see Algorithm \ref{alg:3PAplusplus} below).
Before going into the technical details, it is worth discussing the additional challenges and the high-level idea behind designing the \PPAplusplus algorithm.

As we showed in Example~\ref{ex:non-existence-of-23efx-with-bundles-of-size-two}, even when considering $3$-value instances, there is no way to allocate all critical goods without resulting in allocated bundles of cardinality larger than $2$. This gives rise to two main challenges that we need to address: (i) Agents receiving a single good may not be EFX towards other agents if we just add a good to the bundle of a source of the enhanced envy graph. For example, suppose that agent $i$ has a single good for which she has value $1$ and a source $s$ has two goods for which $i$ has value $b+c< {2}/{3}$. Then $i$ will not be EFX towards $s$ if we give an extra good to $s$ for which $i$ has value $b$; this was not the case when the bundles were restricted to having cardinality at most $2$. (ii) Agents with more than $2$ goods may now have critical goods. For example, an agent with $3$ goods for which she has value $b$ may have a critical good for which she has value $1$. 

To overcome the first challenge, we modify the enhanced envy graph so that a red edge does not represent near-envy of at least $2/3$ from an agent $i$ with a critical good to a source $s$ of the reduced envy graph, but rather the fact that $i$ has value $b$ for one good in the possession of $s$. This modification resolves the situation in (i) above, but now agent $i$ may see her value drop by more than $2/3$ (yet no more than $1/2$) if she receives the bundle of $s$. This is resolved with the addition of a number of sophisticated steps in the \PPAplusplus algorithm that prevent the formation of bundles with {three or more} goods that $i$ values by more than $2/3$ even after removing a single good {from them}. 

Regarding the second challenge, we allow agents with many goods that still have a critical good to exchange \textit{all but one} of their goods with their critical good. This may reduce their value and  result in these agents not being $2/3$-EFX towards agents that hold specific types of bundles. To avoid this violation of $2/3$-EFX, we 
introduce more elaborate versions of the enhanced envy graph (see \cref{def:enhanced-envy-graph-plus,def:doubly_enhanced-envy-graph-plus}) that 
prevent the formation of such bundles through carefully defined additional edges.

For everything that follows up until \cref{lem:zero_c}, we will assume that $c>0$. When $c=0$, \cref{lem:zero_c} shows how to reduce the problem back to the $c>0$ case in polynomial time; this is also reflected to the \TVA algorithm. Hence, the $c>0$ assumption is without loss of generality.

To show {\cref{case3}} of \cref{thm:3-value-instances}, we have to go through the properties of the \PPAplusplus algorithm like we did for \cref{case1,case2} with {\PPA}. That is, we would like to show that 
the partial allocation $\X^1$ outputted by \PPAplusplus when initiated with $\X^0$  in \TVA (Algorithm \ref{alg:3_values}) satisfies \cref{prop:a,prop:b} and is computed in polynomial time. These two properties of \PPAplusplus are established in \cref{lem:3V_noCritical,lem:3V_polytime}. 

\noindent We begin with some definitions which extend those that we presented in \cref{sec:preliminaries}.

\subsubsection{Definitions and Modifications}
We first provide refined versions of the various envy-graphs that we will use. 

\begin{definition}[\REGplus $\Grplus$]\label{def:reduced-envy-graph-plus}
Given a partial allocation $\X$, the \textit{\regplus} \,$\Grplus(\X)$ of $\X$ is defined as a supergraph of $\Gr(\X)$ (recall \cref{def:Gr}) with the addition of any edge $(i,j)$ such that $|X_i|=|X_j|=1$ and $v_i(X_i)=v_i(X_j)=1$. We color those extra edges \emph{gray} and note that they represent zero envy. 
\end{definition}

\begin{definition}[\ENHplus $\Geplus$]\label{def:enhanced-envy-graph-plus}
Given a partial allocation $\X$ and its \regplus $\Grplus(\X)$, we construct the \textit{\enhplus} \,$\Geplus(\X)$ as follows. In $\Grplus(\X)$ we add any edge $(i,j)$ such that $j$ is a source of $\Grplus(\X)$ and the following three properties hold:
(a) $|X_i|=1$, {(b) $|X_j|>1$,} (c) there exists $g\in X_j$ such that $v_i(g)=b$, and (d) there exists $g'\in \PP(\X)$ such that $v_i(g')=b$. 
\end{definition}

\begin{definition}[Doubly Enhanced Graph $\Gde$]\label{def:doubly_enhanced-envy-graph-plus}
Given a partial allocation $\X$ and its  \enhplus $\Geplus(\X)$, we construct the \textit{doubly enhanced graph} $\Gde(\X)$ as follows. In $\Geplus(\X)$ we add any edge $(i,j)$ such that $j$ is a source of $\Geplus(\X)$ and the following three properties hold:
\begin{enumerate}[leftmargin=24pt,label=(\alph*)]
    \item $1+b\leq v_i(X_i) < 2$, \label{def:doubly_enhanced-envy-graph-plus-a} 
    \item from the perspective of $i$, $X_j$ consists of three goods of value $b$ and the rest of {its} goods have value $c$,
    \item there exists $g'\in \PP(\X)$ such that $v_i(g')=b$.
\end{enumerate}
In  $\Gde(\X)$ we allow \emph{self loops} $(i,j)$, where $j=i$, if the above properties hold.
\end{definition}

{Later on, we will show that for the partial allocations maintained by our algorithm, condition \ref{def:doubly_enhanced-envy-graph-plus-a} of \cref{def:doubly_enhanced-envy-graph-plus} above is equivalent to: ``from the perspective of agent $i$, $X_i$ consists of either three goods of value $b$ or one good of value $1$ and one good of value $b$, and the rest of its goods have value $c$''; see \cref{cor:Gde_part_a}.}

We remark that in the graphs described above we might have 
paths or cycles 
that consist solely of gray edges. 
Since there is no envy between any of the agents involved in these, we will refer to those as \emph{unjaundiced} paths or cycles, respectively.

\begin{definition}[Unjaundiced Path or Cycle]
Consider any {nontrivial} path $\Pi$ (resp.~cycle $C$) in $\Grplus(\X)$, $\Geplus(\X)$, or $\Gde(\X)$. We  say that $\Pi$ (resp.~$C$) is \emph{unjaundiced} if it consists entirely of gray edges; otherwise we say it is \emph{jaundiced}.
\end{definition}

Next we introduce a notion that captures the fact that there might be groups of agents who each have a {single good of value $1$} (from their perspective) but are not envied (in the sense captured by $\Grplus$) by anyone outside of their group.
\begin{definition}[\Iso \emph{Set}]
Given a partial allocation $\X$, let $S$ be a set  of agents such that {\textit{(i)}} for all $i\in S$, $v_i(X_i)=1$, and {\textit{(ii)} for any $i\in S$ and $j \notin S$,  $(j,i)$ is not an edge} in $\Grplus(\X)$. We will refer to $S$ as a \textit{\iso} set. 
\end{definition}

\begin{lemma}\label{lem:secluded}
    Given a partial allocation $\X$, there is a unique \iso set $R(\X)$ of maximum cardinality and can be found in polynomial time. 
\end{lemma}

\begin{proof}
By its definition,  
a \iso set cannot contain any agent $j$ such that $v_j(X_j) \neq 1$ nor any agent reachable in $\Grplus(\X)$ {from such a $j$}. We argue that all remaining agents form a \iso set and, thus, this is the unique \iso set of maximum cardinality. Let $R(\X) = \{i\in N : \text{ there is no path } \Pi = (j,\ldots, i) \text{ in } \Grplus(\X), \text{ such } \allowbreak \text{  that } \allowbreak  v_j(X_j) \neq 1 \}$. Clearly, $R(\X)$ can be constructed in polynomial time {by running a DFS in $\Grplus(\X)$ from each $j$ such that $v_j(X_j) \neq 1$}. Also, for any $i\in R(\X)$ we have $v_i(X_i)=1$, by the definition of $R(\X)$, and any $j\notin R(\X)$ could not be in any \iso set as argued above. What remains to be shown is that there is no edge $(j, i)$ with $i\in R(\X)$ and $j\notin R(\X)$. Suppose, towards a contradiction that $(j, i)$ was such an edge. This means that $j$ is reachable by some $j'$ such that $v_{j'}(X_{j'}) \neq 1$ {(where $j'$ could be $j$ herself)}. But then, so is $i$ and thus $i\notin R(\X)$, contradicting the choice of $(j, i)$.
\end{proof}

For the rest of the section, we will use $R(\X)$, as in the statement of \cref{lem:secluded}, to denote the unique maximal \iso set induced by $\X$.

Finally, let $k_{\max}$ be the unique integer such that 
\[b+(k_{\max}-1)c\geq \frac 23 \text{ and } b+(k_{\max}-2)c < \frac 23.\] 
This will be (an upper bound on) the \emph{maximum cardinality} of any allocated bundle as long as there are agents that have critical   goods (see \cref{lem:maxCardinality}). Since $b+c< {2}/{3}$, it must be that $k_{\max}\geq 3$.

\subsubsection{Potential Function}\label{subsec:potential}
To show the termination of our algorithm, we will need to use a potential function. For any partial allocation $\X$, we define the following sets of agents:
\[N_{[2,\infty)}(\X), \ \  N_{[1+b,2)}(\X), \ \ N_{(1,1+b)}(\X), \ \ N_{[2/3,1)}(\X), \ \ N_{1,(0,2/3)}(\X)\,,\]
where 
\begin{itemize}[leftmargin = 15pt]
    \item[-] $N_{[2,\infty)}(\X) = \{i \in N: v_i(X_i) \in [2,\infty)\}$, i.e., agents whose value in $\X$ lies in $[2,\infty)$. \smallskip
    \item[-] $N_{[1+b,2)}(\X)= \{i \in N: v_i(X_i) \in [1+b,2)\}$, i.e., agents whose value in $\X$ lies in $[1+b,2)$.\smallskip
    \item[-] $N_{(1,1+b)}(\X)= \{i \in N: v_i(X_i) \in (1,1+b)\}$, i.e., agents whose value in $\X$ lies in $(1,1+b)$. \smallskip 
    \item[-] $N_{[2/3,1)}(\X)= \{i \in N: v_i(X_i) \in [2/3,1)\}$, i.e., agents whose value in $\X$ lies in $[2/3,1)$. \smallskip
    \item[-] $N_{1,(0,2/3)}(\X)= \{i \in N: v_i(X_i)=1 \text { or } v_i(X_i) \in (0,2/3)\}$, i.e., agents whose value in $\X$ is either $1$ or lies in $(0,2/3)$. 
\end{itemize} 
    
We define as the potential function the tuple 
\[\potential(\X)=(|N_{[2,\infty)}(\X)|, |N_{[1+b,2)}(\X)|, |N_{(1,1+b)}(\X)|, |N_{[2/3,1)}(\X)|, SW_{(0,1+b)}(\X)),\]
where $SW_{(0,1+b)}(\X) = \sum_{i \in N: v_i(X_i)<1+b}v_i(X_i)$ is the social welfare of only the agents who value their bundles less than $1+b$. 
We will show that in each step of our algorithm, the potential function will only lexicographically increase, with the exception of {a polynomially bounded number of} steps during which it may remain constant.

\subsubsection{Additional Properties}
Besides \cref{prop:a,prop:b}, here we need the following two  properties, instead of \cref{prop:c,prop:d,prop:e}.

\begin{tcolorbox}
\underline{Further Desired properties of a partial allocation $\X$:}
\begin{pproperties}[topsep=5pt,itemsep=5pt]
\item For any agent $i\in N_{1,(0,2/3)}(\X)$ and any good $g\in \PP(\X)$, it holds that $v_i(g)<1$. \label{prop:F1}
\item For any agent $i\in N_{1,(0,2/3)}(\X)$ and any agent $j\neq i$, such that $|X_j|\geq 2$, it holds that $v_i(g)<1$ for any $g\in X_j$. \label{prop:F2}
\end{pproperties}
\end{tcolorbox}

\subsubsection{A Refined Version of the Property-Preserving Partial Allocation Algorithm} \label{sec:refined3PA} 
Here we present the refinement of the \PPA algorithm; see Algorithm \ref{alg:3PAplusplus}. In contrast to the results of \Cref{sec:multigraph-and-few-agents}, \PPAplusplus will produce the desired allocation $\X^2$ \emph{directly}, i.e., $\X^2$ will be the allocation $\X^1$ similarly to \cref{case1,case2}. In other words, the processes of obtaining a property-preserving partial allocation and of allocating the critical goods left in the pool will be intertwined. Still, the logic of the algorithm is similar to \PPA, and therefore we refer to this version as a refinement.

{To avoid repetition in the formal description of \PPAplusplus, we slightly abuse notation in the following way. For the extreme case where $k_{\max} = 3$, we do not want to use the doubly enhanced graph at all. Although this is explicit in Step~\ref{3Vstep_GdeCycle}, that is not the case for Steps~\ref{3Vstep_AddGood} and \ref{3Vstep_AddGoodInKMax} which are stated using $\Gde(\X)$ despite the fact that $\Geplus(\X)$ should be used instead when $k_{\max} = 3$. We have opted to only mention it as a comment in the pseudocode for the sake of presentation.}

{In Step~\ref{3Vstep_2a},} \PPAplusplus needs to allocate a few goods in a round-robin fashion. As a shortcut, we introduce Subroutine \ref{alg:single-round-robin}. \SRR takes a partial allocation and a set of agents as input, and augments the allocation by considering the specified agents one after the other, giving to each of them her favorite unallocated good.

\SetAlgorithmName{SUBROUTINE}{Subroutine}

\setcounter{algocf}{5}
\begin{algorithm}[!ht]
	\DontPrintSemicolon
	\caption{$\SRR(\X,S)$} \label{alg:single-round-robin}
	\SetKwComment{Comment}{/* }{ */}
        \SetKw{and}{and}
	\KwData{A partial allocation $\X$ and a set of agents $S$.}
	\KwResult{An updated partial allocation $\X$, in which each $i \in S$ adds her favourite available good from the pool to her bundle, in a round-robin fashion.}
 \tcp{\scriptsize Agents in $S$ pick their favorite good from the pool one by one until each one gets one good or the pool empties.}
 \While{$S\neq \emptyset$ \and $\PP(\X) \neq \emptyset$}
{
Select $i \in S$\;
Let $g^* \in {\arg\max_{g \in \PP(\X)} v_i(g)}$\;
$X_i \gets X_i \cup \{g^*\}$\; 
$S \gets S\setminus\{i\}$
}
\Return $\X$\;
\end{algorithm}

Given the number of distinct steps of \PPAplusplus, it is hard to keep track of what each step does, which one of the four different kinds of envy graphs uses, or what effect it has on the potential. To facilitate the reader's navigation through the proof, we summarize these points for each step in \cref{tab:3PAplusplus-steps}. \medskip

\begin{table}[h!t]
		\centering
		\setlength{\tabcolsep}{4pt}
		\renewcommand{\arraystretch}{1.2}
		\begin{tabular}{cp{11.5cm}cc}
			\hline
			\textbf{Step} & \textbf{What it does} & \textbf{Graph} & \textbf{$\potential$}\\ \hline
			1  & Eliminate a jaundiced cycle (swap bundles around it). & $\Grplus$ & preserve$^\dagger$\\ \hline
			2  & A singleton owner takes a strictly more valuable pool good. & --- & increase\\ \hline
			3  & An owner of $\ge2$ goods takes a pool good worth $>\tfrac32$ her bundle. & --- & increase\\ \hline
			4  & Resolve a jaundiced path whose endpoint is indifferent between her bundle and a pool good. & $\Grplus$ & increase\\ \hline
			5  & Round-robin: each agent of a secluded set takes her favourite pool good. & $\Grplus$ & increase\\ \hline
			6  & \makecell[cl]{A singleton owner takes two pool goods. \\(Or, if only one good remains, she just adds it to her bundle.)} & \makecell{---\\($\Gr$)}  & increase\\ \hline
			7  & An owner of $\ge2$ goods swaps all but her best good for a pool good. & --- & increase\\ \hline
			8  & An owner of three $b$-valued goods (with only $c$-valued goods otherwise) trades two of them for a $1$-valued pool good. & --- & preserve$^\dagger$\\ \hline
			9  & Resolve a jaundiced cycle; a near-envying agent receives two goods she values at $b$. & $\Geplus$ & increase\\ \hline
			10  & If $k_{\max}\ge4$, resolve a jaundiced cycle/self-loop; the agent gets four goods she values at $b$. & $\Gde$ & increase\\ \hline
			11  & Give a critical good to a source with room ($|X_s|<k_{\max}$). & $\Gde$ & preserve$^\dagger$\\ \hline
			12  & A critical-good owner takes a full source's bundle minus one good, plus her critical good. & $\Gde$ & increase\\ \hline
			13  & Terminate (no applicable step). & --- & ---\\ \hline
		\end{tabular}
		\caption{Summary of the steps of \PPAplusplus\ (Algorithm~\ref{alg:3PAplusplus}). ``Graph'' is the envy-graph the step operates on (``---'' meaning that only individual agents and the pool are inspected). $^\dagger$Steps~\ref{3Vstep_GrCycle}, \ref{3Vstep_swap3bToab} and \ref{3Vstep_AddGood} may leave $\potential$ unchanged, but only for polynomially many consecutive rounds (see the proof of \cref{lem:3V_polytime}); every other step increases $\potential$ lexicographically. For $k_{\max}=3$, Steps~\ref{3Vstep_AddGood} and \ref{3Vstep_AddGoodInKMax} use $\Geplus$ in place of $\Gde$.}
		\label{tab:3PAplusplus-steps}
\end{table}

\Cref{subsec:3PA++_correctness,subsec:3PA++_polytime} will be devoted to proving the following two lemmata. By \cref{lem:partialefx-to-efx-extension}, \cref{lem:3V_noCritical} implies the existence of a complete $2/3$-EFX allocation, and \cref{lem:3V_polytime} establishes that this allocation can be obtained in polynomial time. This implies \Cref{case3} of \Cref{thm:3-value-instances}, completing the theorem's proof.

\begin{lemma}
\label{lem:3V_noCritical}
Let $\X^1$ be the partial allocation outputted by \PPAplusplus. Then $\X^1$ is $2/3$-EFX, with no critical goods for any agent. 
\end{lemma}

\begin{lemma}\label{lem:3V_polytime}
	The \PPAplusplus algorithm terminates in polynomial time.
\end{lemma}

\subsubsection{The Proof of \cref{lem:3V_noCritical}}
\label{subsec:3PA++_correctness}
To prove \cref{lem:3V_noCritical}, we develop a sequence of technical lemmata and claims. We begin with some very simple observations but, as we progress, our statements (and proofs) will become increasingly complex.

\begin{lemma}
	\label{lem:valueOfc}
	It holds that $c< {1}/{6}$.
\end{lemma}

\begin{proof}
	This is easy to see, since $b+c< {2}/{3}$ and $b> {1}/{2}$, and so $c< {2}/{3} -b <{2}/{3} -{1}/{2}= {1}/{6}$.
\end{proof}

\begin{lemma}
\label{lem:maxCardinality}
At the end of any iteration of the while loop of the \PPAplusplus algorithm, the cardinality of any allocated bundle of goods is at most $k_{\max}$.
\end{lemma}

\begin{proof}
    Recall that by the fact that $b+c< {2}/{3}$ and the definition of $k_{\max}$, we have that $k_{\max}\geq 3$. The only steps where the \PPAplusplus algorithm may allocate bundles of cardinality more than $2$ are Steps~\ref{3Vstep_GdeCycle} {(one bundle of cardinality $4$)}, ~\ref{3Vstep_AddGood} {(one bundle of cardinality at most $k_{\max}$)} and ~\ref{3Vstep_AddGoodInKMax} {(one bundle of cardinality exactly $k_{\max}$)}. {Given that} Step~\ref{3Vstep_GdeCycle} only runs when $k_{\max} \geq 4$, {we have that,} in any case, after each step, the cardinality of the allocated bundles does not exceed $k_{\max}$. 
\end{proof}

\begin{lemma}
\label{lem:maxValue}
Consider any agent $i \in N$ and any set of goods  $X$, with $|X|\leq k_{\max}$. If, from the perspective of agent $i$, $X$ consists of 
\begin{enumerate}
    \item exactly one good of value $b$ and possibly other goods of value $c$, then $v_i(X)< 5/6$ and $v_i(X\setminus\{g\}) <  2/3$, for any $g\in X$.
    \item only goods of value $c$, then $v_i(X)<  1/2$.
\end{enumerate}  
\end{lemma}

\begin{proof}
For the first {part}, by the definition of $k_{\max}$, we have that $v_i(X)<{2}/{3} + c$, which implies $v_i(X\setminus\{g\}) < {2}/{3}$, for any $g\in X$. By \cref{lem:valueOfc} and the additivity of the values, we have $v_i(X)< {2}/{3} +  1/6 = 5/6$.
    For the second {part}, we have that $v_i(X)\leq k_{\max}\cdot c$. By the definition of $k_{\max}$ and \cref{lem:valueOfc}, we have that $k_{\max}\cdot c<{2}/{3} - b +2c< 2/3 -  1/2 +  2/6 =  1/2 $.
\end{proof}

\begin{corollary}\label{cor:Gde_part_a}
	Assuming that every bundle of an allocation $\X$ has cardinality at most $k_{\max}$, $i\in N_{[1+b,2)}(\X)$ is equivalent to agent $i$ seeing $X_i$ consisting of either three goods of value $b$ or one good of value $1$ and one good of value $b$, and the rest of its goods having value $c$.
\end{corollary}\smallskip

\begin{lemma}
\label{lem:Value1Cardinality1}
At the end of any iteration of the while loop of \PPAplusplus, for any agent $i \in N$ and any allocated bundle $X$, if\, $v_i(X)\in \{1, b, c\}$, then $|X|=1$. 
\end{lemma}


\SetAlgorithmName{ALGORITHM}{Algorithm}

\newgeometry{margin = 0.7in} 

\renewcommand\mycommfont[1]{\scriptsize\ttfamily\textcolor{gray}{#1}}

\SetAlgorithmName{ALGORITHM}{Algorithm}

\setcounter{algocf}{8}

\begin{algorithm}[p]
\setstretch{1.05}
\small
\DontPrintSemicolon
\SetNoFillComment
\LinesNotNumbered 
\caption{\textsc{\sc Property-Preserving Partial Allocation$^{++}$ (\PPAplusplus)} 
$\left(\left(v_i\right)_{i 
\in N}, \X \right)$} \label{alg:3PAplusplus}
\SetKwComment{Comment}{/* }{ */}
\SetKw{Continue}{continue}
\SetKw{Break}{break}
\SetKw{Step}{Step}
\SetKw{and}{and}
\SetKwData{Kw}{}
\KwData{The values $v_i(g)$ for every  $i \in N$ and every  $g \in M$, and a partial allocation $\X$ in which $|X_i|=1$ for all $i \in N$.}
\KwResult{A Property-Preserving Partial Allocation, a partial allocation $\X^1$ satisfying Properties \ref{prop:a}, \ref{prop:b}, \ref{prop:F1}, and \ref{prop:F2} \vspace{2pt}
}
\While{$\PP(\X) \neq \emptyset$}
{
\nl \Kw{\texttt{Step 1}} \label{3Vstep_GrCycle}
\uIf{there is a jaundiced cycle $C$ in $\Grplus(\X)$}
{
$\X \gets \CR(\X, \Grplus)$\;
}
\nl \Kw{\texttt{Step 2}} \label{3Vstep_swap1To1}
\uElseIf{there is $i \in N$ with $|X_i|=1$ and a good $g \in \PP(\X)$ such that $v_i(g) > v_i(X_i)$}{
$X_i \gets \{g\}$\; 
}

\nl \Kw{\texttt{Step 3}}\label{3Vstep_swap2To1}
\uElseIf{there is $i \in N$ with $|X_i| \ge 2$ and a good $g \in \PP(\X)$ such that $v_i(g) > \frac{3}{2}v_i(X_i)$ 
}{
$X_i \gets \{g\}$\; 
}

\nl \Kw{\texttt{Step 4}} \label{3Vstep_equalVinPool}
\uElseIf{there is a jaundiced path $\Pi=(i_1, \ldots, i_{\ell})$ in $\Grplus(\X)$ and a good  $g\in \PP(\X)$ so that $v_{i_{\ell}}(X_{i_{\ell}})=v_{i_{\ell}}(g)$}
{
$(X_j)_{j \in \{i_1,\ldots,i_{\ell-1}\}} \gets \PR(\X,\Grplus,\Pi)$\; 
$X_{i_\ell}=\{g\}$\; 
}

\nl \Kw{\texttt{Step 5}} \label{3Vstep_2a}
\uElseIf(\ \ \tcp*[h]{there is a nonempty \iso set }){$R(\X)\neq \emptyset$ }{
$\X \gets \SRR(\X,R(\X))$\;
}

\nl \Kw{\texttt{Step 6}} \label{3Vstep_swap1To2}{
\uElseIf{there is $i \in N$ with $|X_i|=1$ and two goods $g_1,g_2 \in \PP(\X)$ such that $v_i(\{g_1,g_2\}) > v_i(X_i)$}
{
$X_i \gets \{g_1,g_2\}$\; 
\hspace{0.77cm} \textbf{else if} $\PP(\X)=\{g\}$ and  some $i \in N$, with $|X_i|=1$ and $v_i(X_i) = v_i(g)$, is a source of\, $\Gr(\X)$ \textbf{then}\;
$X_i \gets X_i \cup \{g\}$\;
}
}

\nl \Kw{\texttt{Step 7}}\label{3Vstep_swapManyTo2}
\uElseIf{there is $i \in N$ with $|X_i|\geq 2$ and a good 
$g^* \in \PP(\X)$ such that $v_i(g^*) > v_i(X_i\setminus\{g_1\})$ for some $g_1\in \arg\max_{g \in X_i}v_i(g)$}
{ $X_i \gets \{g_1,g^*\}$\; 
}

\nl \Kw{\texttt{Step 8}}\label{3Vstep_swap3bToab}
\uElseIf{there exists $i\in N$ and goods $g_1,g_2,g_3 \in X_i$ such that $v_i(g_j)=b$, for $j \in \{1,2,3\}$, and $v_i(g)=c$, for any other good $g \in X_i \setminus \{g_1,g_2,g_3\}$ (if any), \and there exists a good $g^* \in \PP(\X)$ such that $v_i(g^*)=1$}
{
$X_i \gets \{g_1, g^*\}$\; 
}

\nl \Kw{\texttt{Step 9}}\label{3Vstep_GeCycle}
\uElseIf{in the \enhplus $\Geplus(\X)$ there is a jaundiced cycle $C$}
{
Let $i \in N$ be such that $(i,s)$ is an edge of $C$ and $s$ is a source of $\Grplus(\X)$ \tcp*{existence by \cref{lem:helpToStep9}}
$\X \gets \CR(\X, \Geplus, C)$\;
Let $g \in X_{i}$ such that $v_{i}(g) = b$
\tcp*[r]{exists by the definition of $\Grplus$, since $X_{i}$ was $s$'s bundle right before}
Let $g^* \in \PP(\X)$ such that $v_{i}(g^*)=b$ \tcp*{exists by the definition of $\Geplus$} 
$X_{i} \gets \{g, g^*\}$\; 
}

\nl \Kw{\texttt{Step 10}} \label{3Vstep_GdeCycle}
\uElseIf{$k_{\max} \geq 4$ \and in the doubly enhanced graph $\Gde(\X)$ there is a jaundiced cycle or self-loop $C$}{
Let $i \in N$ be such that $(i,s)$ is an edge of $C$ and $s$ is a source of $\Geplus(\X)$ \;
\tcp{such an $i$ always exists since Step~\ref{3Vstep_GeCycle} was not executed, and hence all cycles in $\Geplus$ ---if any--- are unjaundiced}
$\X \gets \CR(\X, \Gde, C)$\;
Let $g_1,g_2,g_3 \in X_{i}$ be such that $v_{i}(g_j)=b$, for $j \in \{1,2,3\}$\; 
\tcp{$g_1,g_2,g_3$ exist by the definition of $\Gde$, since $X_{i}$ was $s$'s bundle right before}
Let $g^* \in \PP(\X)$ such that $v_{i}(g^*)=b$ \tcp*{exists by the definition of $\Gde$} 
$X_{i} \gets \{g_1,g_2,g_3, g^*\}$\; 
}

\nl \Kw{\texttt{Step 11}}\label{3Vstep_AddGood}
\uElseIf{there is $i\in N$ with a critical good $g$ and a source $s$ in $\Gde(\X)$ with $|X_s| < k_{\max}$}{
\tcp{when $k_{\max} = 3$,  $\Geplus(\X)$ should be used instead of $\Gde(\X)$ throughout this step, including the condition above}
$X_s \gets X_s \cup \{g\}$ \;
}

\nl \Kw{\texttt{Step 12}}\label{3Vstep_AddGoodInKMax}
\uElseIf{there is an agent $i$ with a critical good $g$, and for every source $s$ in 
$\Gde(\X)$, it holds that $|X_s| = k_{\max}$}
{
	\tcp{when $k_{\max} = 3$,  $\Geplus(\X)$ should be used instead of $\Gde(\X)$ throughout this step, including the condition above}
Let $s$ be a source of $\Gde(\X)$ such that there exists a path $\Pi$ from $s$ to $i$ in $\Gde(\X)$ and let $g'\in \arg\min_{g \in X_{s}}v_i(g)$\; 
$Y \gets X_{s} \setminus \{g'\}$\; 
$(X_j)_{j \in N, \exists (j,\ell) \in \Pi} \gets \PR(\X,\Gde,\Pi)$\;
$X_i \gets Y \cup \{g\}$\;  %
}

\nl \Kw{\texttt{Step 13}}\label{3Vstep_Break}
\Else {\Break\;}

}
\Return $\X$
\end{algorithm}
\renewcommand\mycommfont[1]{\footnotesize\ttfamily\textcolor{gray}{#1}}

\restoregeometry


\begin{proof}
    We consider three cases, depending on $v_i(X)$. 
    
    \noindent \textit{Case 1: $v_i(X)=1$}. Suppose, towards a contradiction, that $|X|>1$. Note that $X$ cannot contain any good $g$ such that $v_i(g)=1$ or any two goods $g_1,g_2$ such that $v_i(g_1)=v_i(g_2)=b$. {Indeed, since $b>1/2$ and} $c>0$, in each case that would imply that $v_i(X) > 1$. This means that $X$ contains at most one good $g$ such that $v_i(g)=b$, and for any other good $g' \in X$, it holds that $v_i(g')=c$. By \cref{lem:maxValue}, it holds that $v_i(X)< 5/6<1$, a contradiction. 
    
    \noindent \textit{Case 2: $v_i(X)=b$}. Suppose, towards a contradiction, that $|X|>1$. Note that $X$ cannot contain any good $g$ such that $v_i(g)=1$ or $v_i(g)=b$, because in that case we would have $v_i(X)>b$ (again, since $c >0$). This means that $X$ contains only goods $g$ such that $v_i(g)=c$. By \cref{lem:maxValue}, $v_i(X)< 1/2$, a contradiction, since $b > 1/2$.

    \noindent \textit{Case 3: $v_i(X)=c$}.  In this case the statement of the lemma holds trivially by the fact that $c >0$.
\end{proof}

\begin{lemma}
\label{lem:helpToStep9}
    Let $\X$ be a partial allocation for which $\Grplus(\X)$ only has unjaundiced cycles, and for which $R(\X) = \emptyset$, i.e., every \iso set  is empty 
    (i.e., Steps~\ref{3Vstep_GrCycle} and ~\ref{3Vstep_2a} of the \PPAplusplus algorithm do not run). Then, there exists at least one source $s$ in the \regplus $\Grplus(\X)$.
\end{lemma}

\begin{proof}
Suppose, towards a contradiction, that there are no sources in the \regplus $\Grplus(\X)$. From {$\Grplus(\X)$} we generate a new graph $H$ as follows: as long as there is a cycle $C$ in $\Grplus(\X)$, we create a new node $x$, and for any $i\notin C$, we add an edge $(i,x)$ if there exists an edge $(i,j)$ for some $j \in C$, and an edge $(x,i)$ if there exists an edge $(j,i)$ for some $j \in C$. Then, we remove all the nodes of $C$ and all the edges adjacent to them. At the end of this procedure $H$ is acyclic. {(Essentially, we iteratively contract each cycle into a single vertex.)}

Let $s$ be a source of $H$. If $s$ is also a node of $\Grplus(\X)$ then it is a source of $\Grplus(\X)$ and we get a contradiction. Therefore, we may assume that $s$ is not a node of $\Grplus(\X)$ and let $R'(\X)$ be the set of nodes of $\Grplus(\X)$ that $s$ ``replaced'' in all rounds of the aforementioned procedure. This means that any $i\in R'(\X)$ was part of some cycle of $\Grplus(\X)$, since it was replaced. By the premise of the lemma statement, each cycle in {$\Grplus(\X)$} is unjaundiced.  This implies that for every agent $i \in R'(\X)$, it holds that $v_i(X_i)=1$. Since $s$ is a source of $H$ {and by the way $H$ is constructed}, there is no edge $(j,i)$ of {$\Grplus(\X)$} for any $i\in R'(\X)$ and any $j \notin R'(\X)$. This means that $R'(\X)$ is a nonempty \iso set, 
{which directly contradicts the statement.}
\end{proof}

Before we continue with the statement of the remaining lemmata, we define the notion of an agent hierarchy, based on our potential function defined earlier in the section. 
\begin{definition}[Agent Hierarchy for a Partial Allocation]\label{def:hierarchy}
Given a partial allocation $\X$, we consider a \emph{hierarchy} of the agents according to their value for their bundle as follows: 
\begin{levels}[left=0.4\textwidth]
    \item $N_{[2,\infty)}(\X)$\label{level1}\smallskip
    \item $N_{[1+b,2)}(\X)$ \label{level2}\smallskip
    \item $N_{(1,1+b)}(\X)$ \label{level3}\smallskip
    \item $N_{[2/3,1)}(\X)$\label{level4}\smallskip
    \item $N_{1,(0,2/3)}(\X)$\label{level5}\smallskip
\end{levels}

When comparing two partial allocations $\X$ and $\X'$, we will say that an agent \emph{falls in hierarchy from $\X$ to $\X'$}, if in $\X$ and $\X'$ she is on Levels (L.$i$) and (L.$i'$) respectively, with $i < i'$. 
\end{definition}

\begin{lemma}
\label{lem:agentHierarchy}
Let $\X$ and $\X'$ be any two partial allocations obtained in the course of the \PPAplusplus algorithm, with $\X'$ being obtained after $\X$. 
{No agent $i \in N$ falls} in hierarchy from $\X$ to $\X'$.  
\end{lemma}

\begin{proof}
An agent can only fall in hierarchy if $v_i(X_i') \neq v_i(X_i)$. We first consider the steps of the \PPAplusplus algorithm during which $v_i(X_i') < v_i(X_i)$, these are Steps~\ref{3Vstep_swap3bToab}, \ref{3Vstep_GeCycle}, \ref{3Vstep_GdeCycle} and \ref{3Vstep_AddGoodInKMax}. 
Consider Step~\ref{3Vstep_swap3bToab} and let  $i$ be the agent whose bundle changes in this step. Before Step~\ref{3Vstep_swap3bToab}, we had that $i \in N_{[1+b,2)}(\X)$ {(since $v_i(X_i) = 3b+(|X_i|-3)c \le 3b+(k_{\max}-3)c  = 2b + (b +(k_{\max}-3)c) < 2\cdot 2/3 +  2/3 = 2$, by the definition of $k_{\max}$ and \cref{lem:maxCardinality}, and $v_i(X_i)\ge 3b>1+b$)}, i.e., the agent was at \cref{level2}. After the step, she receives a bundle $X_i'$ for which $v_i(X_i')=1+b$; this means that $i \in N_{[1+b,2)}(\X')$. 
Next, consider Step~\ref{3Vstep_GeCycle}. The only agent(s) $i \in N$ for whom the value for their bundle possibly decreases are the ones with $|X_i|=1$ and $v_i(X_i)=1$. 
Those agents belong to the set of \cref{level5} in $\X$, i.e., $i \in N_{1,(0,2/3)}(\X)$, so they cannot fall in hierarchy by definition.
Step~\ref{3Vstep_AddGoodInKMax} is similar. Here the  agents $j \in N$ for whom the value for their bundle possibly decreases are either the ones with $|X_j|=1$ and $v_j(X_j)=1$ like in Step~\ref{3Vstep_GeCycle} (including the main agent $i$ of the step, by \cref{lem:3VcriticalGoods}) who cannot fall in hierarchy by definition, or the agents who drop in value but stay in $N_{[1+b,2)}(\X')$ because of edges which are unique to $\Gde$.
Finally, consider Step~\ref{3Vstep_GdeCycle} which is slightly more subtle. 
The set of agents for whom the value for their bundle possibly decreases {is for the most part exactly like in the case of Step~\ref{3Vstep_AddGoodInKMax}, only now we should also consider $i$ herself. However, it is easy to see that $v_{i}(X_{i}')=4b>2$, i.e., the agent is now at \cref{level1}.}

We lastly consider the remaining case, i.e., when $v_i(X_i') > v_i(X_i)$, and agent $i$ possibly falling in hierarchy from 
\cref{level4} to \cref{level5}, i.e., $i \in N_{[2/3,1)}(\X)$ and $i \in N_{1,(0,2/3)}(\X')$. For this to be possible, it must be the case that $2/3 \le v_i(X_i) < 1$ and $v_i(X_i')=1$. In turn, this can only happen if
$|X_i|>1$ and,
by \cref{lem:Value1Cardinality1}, $|X_i'|=1$. However, during the execution of the  \PPAplusplus algorithm, if an agent with $|X_i|\geq 2$ finds herself with a bundle {of higher value} $X_i'$ such that $|X_i'|=1$, her value must {have increased by a factor \textit{greater}} than $3/2$ (see Step~\ref{3Vstep_swap2To1} of the \PPAplusplus algorithm and 
{recall the definitions of the various envy graphs that are used in Steps~\ref{3Vstep_GrCycle}, \ref{3Vstep_equalVinPool}, \ref{3Vstep_GeCycle}, \ref{3Vstep_GdeCycle} and \ref{3Vstep_AddGoodInKMax}),}
leading to a contradiction.   
\end{proof}

\begin{lemma}
\label{lem:not=ValuesInP}
Let $\X$ be a partial allocation for which
\begin{enumerate}[label=(\alph*)]
    \item all the cycles in $\Grplus(\X)$ are unjaundiced, 
    \item all the {nontrivial} paths in $\Grplus(\X)$ ending at an agent $i$ with $v_i(X_i)=v_i(g)$ for some $g \in \PP(\X)$ are unjaundiced,
    \item all the \iso sets are empty ($R(\X)= \emptyset)$,
\end{enumerate}
(i.e., Steps~\ref{3Vstep_GrCycle}, \ref{3Vstep_equalVinPool} and \ref{3Vstep_2a} of the \PPAplusplus algorithm will not run).  Then, for any agent $i \in N$ such that either $v_i(X_i)=1$ or $i$ is not a source of the \regplus $\Grplus(\X)$, it holds that $v_i(X_i)\neq v_i(g)$ for any $g\in \PP(\X)$.
\end{lemma}

\begin{proof}
Suppose, towards a contradiction, that there exists some agent $i \in N$ that is either not a source of $\Grplus(\X)$ or for which $v_i(X_i)=1$, and some good $g \in \PP(\X)$ such that $v_i(g) = v_i(X_i)$. By the statement of the lemma, we know that any path $\Pi$ in $\Grplus(\X)$ is unjaundiced. Let $R'(\X) = \{j \in N: j \text{ is on some path $\Pi$ to $i$ in } \allowbreak \Grplus(\X)\}$, and add agent $i$ to $R'(\X)$ only if $v_i(X_i)=1$. {By the definition of gray edges,} since every agent $j \in R'(\X)$ lies on some unjaundiced path to agent $i$, we have $v_j(X_j)=1$. We remark that $R'(\X) \neq \emptyset$. Indeed, if there is at least one unjaundiced path ending at agent $i$, all the agents of that path would be included in $R'(\X)$. If there are no such paths, by the second condition of the statement of the lemma, there are no {nontrivial} paths ending at agent $i$ in $\Grplus(\X)$ in general. This implies that $i$ is a source of $\Grplus(\X)$; this would violate the assumption at the beginning of the proof, unless $v_i(X_i)=1$. In this case, however, $i \in R'(\X)$ and hence $R'(\X) \neq \emptyset$. We will next argue that $R'(\X)$ is a \iso set. 

Let $j \notin R'(\X)$; if such a $j$ does not exist then $R'(\X)$ is trivially a \iso set. Let $\ell \in R'(\X)$; such an $\ell$ exists since $R'(\X) \neq \emptyset$. Assume first that $j \neq i$. We argue that there is no edge $(j,\ell)$ in the $\Grplus(\X)$: indeed, if such an edge existed, there would be a path from $j$ to $i$ (via $\ell$, since there is a path from $j$ to $\ell$), contradicting the fact that $j \notin R'(\X)$. Since $j$ and $\ell$ were chosen arbitrarily, this implies that $R'(\X)$ is a \iso set. Now assume that $j = i$. Again we argue that there is no edge $(i,\ell)$ in the $\Grplus(\X)$: indeed, if such an edge existed, there would be a cycle $C$ including $i$ (through the edge $(i,\ell)$ and since there is a path from $\ell$ to $i$ by the definition of $R'(\X)$). Since $i = j \notin R'(\X)$ by assumption, we have that $v_i(X_i) \neq 1$. This implies that the edge $(i,\ell)$ is not gray, and $C$ is jaundiced, contradicting the first condition of the statement. The lemma follows.
\end{proof}

\begin{corollary}
\label{cor:only<ValuesInP}
Let $\X$ be a partial allocation {such that in the next} iteration of the while loop of the \PPAplusplus algorithm Steps~\ref{3Vstep_GrCycle}-\ref{3Vstep_2a} {will} not run. For any agent $i$ with $v_i(X_i)=1$, it holds that $v_i(g)<1$ for any $g\in \PP(\X)$.
\end{corollary}

\begin{proof}
The corollary follows immediately by \cref{lem:not=ValuesInP}.
\end{proof}

\begin{lemma}
\label{lem:prop(f1)}
Let $\X$ be a partial allocation {such that in the next} iteration of the while loop of the \PPAplusplus algorithm Steps~\ref{3Vstep_GrCycle}-\ref{3Vstep_2a} {will} not run. Then $\X$ satisfies \cref{prop:F1}. 
\end{lemma}

\begin{proof}
    For any agent $i \in N$ with $v_i(X_i)< 2/3$ and any good $g\in \PP(\X)$, it must be the case that $v_i(g)<1$, as otherwise one of Steps~\ref{3Vstep_swap1To1} or \ref{3Vstep_swap2To1} of the \PPAplusplus algorithm would run. If $v_i(X_i)=1$, by \cref{cor:only<ValuesInP}, it holds that $v_i(g)<1$ for any good $g\in \PP(\X)$.
\end{proof}

\begin{lemma}
\label{lem:prop(f2)}
    Let $\X$ be a partial allocation at the end of any iteration of the while loop of the \PPAplusplus algorithm such that $\PP(\X) \neq \emptyset$. Then $\X$ satisfies \cref{prop:F2}. 
\end{lemma}

\begin{proof}
We will prove the lemma by induction on the number of iterations of the while loop of the \PPAplusplus algorithm. For the base case, \cref{prop:F2} is trivially satisfied for the input allocation $\X$ before the execution of the first iteration of the while loop, since $|X_i|=1$ for all $i \in N$. 

Suppose that \cref{prop:F2} holds at the beginning of some iteration of the while loop, {when the allocation is $\X$}; we will show that after the end of the iteration, {when the allocation becomes $\X'$}, regardless of which step ran, the property still holds. Recall that, by \cref{lem:agentHierarchy}, the set $N_{1,(0,2/3)}(\X)$ can only become smaller, 
{so if $i\in N_{1,(0,2/3)}(\X') \subseteq N_{1,(0,2/3)}(\X)$, we know that \cref{prop:F2} is satisfied for $i$ in $\X$ by the induction hypothesis.}
\smallskip 

\noindent \textit{If Step~\ref{3Vstep_GrCycle} is executed:} {The set of bundles remains unchanged and no agent's value for her bundle is reduced. 
Therefore, for any agent $i\in N_{1,(0,2/3)}(\X')$, \cref{prop:F2} \emph{almost}  holds by the induction hypothesis.
The only additional bundle we may need to take into account for such an agent $i$ is $X_i$ itself, if the latter happened to change hands and $|X_i|\ge 2$. However, since $i\in N_{1,(0,2/3)}(\X)$, for any $g\in X_i$ we know that $v_i(g)<1$. We conclude that \cref{prop:F2} is satisfied in this case.} \smallskip 

\noindent {\textit{If Step~\ref{3Vstep_swap1To1} or Step~\ref{3Vstep_swap2To1} is executed:} The only bundle that is modified, say from $X_i$ to  $X'_i$, 
consists of a single good. Thus, \cref{prop:F2} holds by the induction hypothesis. } \smallskip

\noindent {\textit{If Step~\ref{3Vstep_equalVinPool} is executed:} This is very similar to Step~\ref{3Vstep_GrCycle}, the only difference being that a single bundle changed, namely, $X_{i_1}$ was removed and $X'_{i_\ell} = \{g\}$ was added to $\X'$. Since 
$X'_{i_\ell}$ is irrelevant for other agents with respect to \cref{prop:F2} (as $|X'_{i_\ell}| = 1$), \cref{prop:F2} is satisfied as in the case of Step~\ref{3Vstep_GrCycle}. }  \smallskip

\noindent \textit{If Step~\ref{3Vstep_2a} is executed:} The only agents whose bundles are modified are those in the \iso set $R(\X)$. In particular, each agent $i \in R(\X)$ 
{has a singleton in $\X$ (by \cref{lem:Value1Cardinality1}) and} receives one additional good from the pool in $\X'$. Hence, $i \notin N_{1,(0,2/3)}(\X')$. Consider now such an agent $i\in R(\X)$ and any agent $j \in N_{1,(0,2/3)}(\X')$. 
{Let $g$ be the unique good in $X'_i \setminus X_i$ (which belongs to $\PP(\X)$). 
First consider the case where $v_j(X_j)<1$.}
Since Step~\ref{3Vstep_swap1To1} and Step~\ref{3Vstep_swap2To1} did not run, it follows that $v_j(g)<1$. Moreover, by the definition of $R(\X)$, we have that 
{$(j,i)\notin E(\Grplus(\X))$} and thus $v_j(g')<1$ for {the unique} $g' \in X_i$. If $v_j(X_j)=1$, by \cref{cor:only<ValuesInP}, we have that $v_j(g)<1$. Moreover, by the definition of $R(\X)$, we have that $v_j(X_i) < v_j(X_j)=1$ (otherwise there would be a gray edge $(j,i)$), and again we have $v_j(g')<1$ for {the unique} $g' \in X_i$. This establishes that \cref{prop:F2} holds for $\X'$. \smallskip

\noindent {\textit{If any of Steps~\ref{3Vstep_swap1To2} -- \ref{3Vstep_swap3bToab}  and \ref{3Vstep_AddGood} is executed:} The only bundle that is modified from $\X$ to $\X'$ now consists only of goods from $\PP(\X)$ and---possibly---one or more goods from a bundle $X$ of $\X$, such that $|X|\geq 2$ or $X$ belonged to a source in $\Grplus(\X)$. By the definition of $\Grplus(\X)$,} the induction hypothesis, and  \cref{lem:prop(f1)}, we conclude that \cref{prop:F2} holds for $\X'$.  \smallskip

\noindent {\textit{If any of Steps~\ref{3Vstep_GeCycle}, \ref{3Vstep_GdeCycle}, or  \ref{3Vstep_AddGoodInKMax} is executed:} Now some of the bundles remain the same, although they may change hands via \CR or \PR, as in Steps~\ref{3Vstep_GrCycle} and \ref{3Vstep_equalVinPool}. With respect to these bundles, \cref{prop:F2} is satisfied by the induction hypothesis, like in the case of Steps~\ref{3Vstep_GrCycle} and \ref{3Vstep_equalVinPool}. 
However, other bundles may be modified---one for each of Steps~\ref{3Vstep_GeCycle}, \ref{3Vstep_GdeCycle}, and  \ref{3Vstep_AddGoodInKMax}. Any bundle that is modified from $\X$ to $\X'$ now consists only of goods from $\PP(\X)$ and from a subset of a bundle $X$ of $\X$, such that $|X|\geq 2$ or $X$ belonged to a source in $\Grplus(\X)$, like in Steps~\ref{3Vstep_swap1To2} -- \ref{3Vstep_swap3bToab}  and \ref{3Vstep_AddGood}. 
With respect to all these modified bundles, \cref{prop:F2} is satisfied by the definition of $\Grplus(\X)$, the induction hypothesis, and  \cref{lem:prop(f1)}, like in the case of Steps~\ref{3Vstep_swap1To2} -- \ref{3Vstep_swap3bToab}  and \ref{3Vstep_AddGood}.}
\end{proof}

\begin{lemma}
\label{lem:One_b}
Let $\X$ be a partial allocation at the end of any iteration of the while loop of the \PPAplusplus algorithm, and let $X$ be any allocated bundle with $|X|\geq 3$. For any agent $i\in N_{1,(0,2/3)}(\X)$, {if $g \in \argmax_{x\in X} v_i(x)$, then}
\begin{itemize}
    \item[-] $v_i(g)\leq b$;
    \item[-] for any $g' \in X\setminus\{g\}$, it holds that $v_i(g')=c$.
\end{itemize}
\end{lemma}

\begin{proof}
By \cref{lem:prop(f2)}, for any $g'\in X$, we have that $v_i(g')<1$, and thus $v_i(g') \le b$. Suppose, towards a contradiction, that there exist distinct $g_1,g_2\in X$ such that $v_i(g_1)=v_i(g_2)=b$. 

    Consider the first iteration, {say the $r$-th iteration overall,} of the while loop after which $\{g_1, g_2\} \in \tilde{X}$ for some bundle $\tilde{X}$ and none of them is removed from $\tilde{X}$ afterwards until the allocation $\X$ is reached. Let $\X'$ and $\X''$ be the partial allocation right after and right before, respectively, the execution of this $r$-th iteration, and similarly let $X'$ and $X''$ be the aforementioned set $\tilde{X}$, after and prior, respectively, the execution of the $r$-th iteration. So, {$\{g_1,g_2\} \subseteq X'$ but $\{g_1,g_2\} \not\subseteq X''$}.  

    If $|X'|=2$, agent $i$ would prefer $X'$ over her own bundle in all iterations until $\X$ is reached. This is because $i\in N_{1,(0,2/3)}(\X)$ and, by \cref{lem:agentHierarchy}, agent $i$'s value for her bundle is at most $1$ \textit{in all those iterations}; also,  recall that $v_i(X')=2b>1$. The cardinality of a bundle 
    {of size $2$ can only increase in Step~\ref{3Vstep_AddGood}. Further, notice that in this case a bundle of size $2$ that increases its cardinality was allocated to an agent who was a \textit{source} of the doubly enhanced or the \enhplus  at the beginning of the corresponding iteration. However, the presence of $i$ makes it impossible for any future variant of bundle $X'$ to ever belong to a source of either of these graphs in any iteration} until $\X$ is reached. This directly contradicts the fact that $|X|\geq 3$.

    Suppose now that $|X'|\geq 3$. In general, a bundle $\tilde{X}$ {such that $|\tilde{X}| \ge 3$} is updated to another bundle with cardinality at least $3$ only in Steps \ref{3Vstep_GdeCycle} -- \ref{3Vstep_AddGoodInKMax}. In these steps, only a single good is added to $\tilde{X}$ from the pool. Assume, without loss of generality, that $g_1\in X''$ and $g_2 \in \PP(X'')$, i.e., $g_1$ was already {in $X''$ right before the $r$-th iteration}, and $g_2$ was added to $X''$ from the pool during {this iteration}. This also means that $|X''|\geq 2$; {otherwise, the resulting bundle $X'$ would have at most $2$ items}. Regardless of which of Steps~\ref{3Vstep_GdeCycle} -- \ref{3Vstep_AddGoodInKMax} ran the $r$-th iteration, the bundle $X''$ had to be assigned to some source $s$ of the \enhplus $\Geplus(\X'')$. We will show that this cannot happen by considering all different cases for agent $i$, who, by \cref{lem:agentHierarchy}, is in  $N_{1,(0,2/3)}(\X'')$:
    \begin{itemize}[leftmargin=*]
        \item[-] $v_i(X_i'')\leq b$. In this case $i$ would envy $s$ in $\X''$, since $X''$ contains at least two goods, one of which has value $b$ for $i$. So, $s$ cannot be a source in $\Geplus(\X'')$.\smallskip
        \item[-] $b<v_i(X_i'')<1$. {In this case, agent $i$ and the bundle $X_i''$ satisfy the conditions of \cref{lem:maxValue}.} Thus, there exists a good $g'\in X_i''$, such that $v_i(g')=b$, and $v_i(X_i''\setminus\{g'\}) <  1/2$. Since $v_i(g_2) = b$, the condition of Step~\ref{3Vstep_swapManyTo2} would have been triggered, where agent $i$ would swap $X_i''\setminus\{g'\}$ with $g_2$. So, the \PPAplusplus algorithm would not execute any of Steps~\ref{3Vstep_GdeCycle} -- \ref{3Vstep_AddGoodInKMax} {during this iteration}. \smallskip 
        \item[-] $v_i(X_i'')=1$. In this case $g_2$ is a critical good for agent $i$. Since $v_i(g_1)=b$ and $i$ has a critical good, $(i,s)$ should be an edge in \enhplus $\Geplus(\X'')$ and therefore $s$ cannot be a source of $\Geplus(\X'')$.
    \end{itemize}
    Since in every case our assumption that there exist $g_1,g_2\in X$ such that $v_i(g_1)=v_i(g_2)=b$ leads to a contradiction, the lemma follows.
\end{proof}

\begin{lemma}
\label{lem:Three_b}
Let $\X$ be a partial allocation at the end of any iteration of the while loop of the \PPAplusplus algorithm, and let $X$ be any allocated bundle. For any agent $i\in N_{[1+b,2)}(\X)$, it holds that $v_i(X\setminus\{g\})<b+ 5/3$, for any $g \in X$.
\end{lemma}

\begin{proof}
 In the following two claims we show that
\begin{itemize}[leftmargin=20pt]
    \item[-] there are no two goods $g_1,g_2 \in X$ such that $v_i(g_1)=v_i(g_2)=1$, unless $|X|=2$,
    \item[-] if $|X|\geq 5$ there are no four goods $g_1,g_2,g_3,g_4 \in X$ such that $v_i(g_j) \geq b$, for $j \in \{1,2,3,4\}$.
\end{itemize}
This means that, for any $g \in X$, 
\begin{itemize}[leftmargin=20pt]
    \item[-] if $|X|\leq 4$, it holds that $v_i(X\setminus\{g\})\leq 1 + 2b<b+ 5/3$, and
    \item[-] if $|X|\geq 5$,  $v_i(X\setminus\{g\})\leq 1 + 2b + (k_{\max}-4)c < b+ 5/3$, 
\end{itemize}
{where the last inequality follows by the definition of $k_{\max}$.}

\begin{claim}\label{claim:no_two_unit_goods}
If $|X|\geq 3$ there are no two goods $g_1,g_2\in X$ such that $v_i(g_1)=v_i(g_2)=1$. 
\end{claim}

\begin{proof}[Proof of \cref{claim:no_two_unit_goods}]
{We proceed as in the proof of \cref{lem:One_b}.}
Suppose, towards a contradiction, that there exist $g_1,g_2\in X$ such that $v_i(g_1)=v_i(g_2)=1$. 
Consider the first iteration, {say the $r$-th iteration overall,} of the while loop after which
$g_1$ and $g_2$ were assigned together in the same bundle $\tilde{X}$ and none of them is removed from $\tilde{X}$ until the allocation $\X$ is reached. 
Let $\X'$ and $\X''$ be the partial allocation right after and right before, respectively, the execution of this $r$-th iteration, and similarly let $X'$ and $X''$ be the aforementioned set $\tilde{X}$, after and prior, respectively, the execution of the $r$-th iteration. So, {$\{g_1,g_2\} \subseteq X'$ but $\{g_1,g_2\} \not\subseteq X''$}.   

If $|X'|=2$, then $v_i(X')=2$ and agent $i$ would envy the set $X'$ in all iterations until $\X$ is reached. This is because $i\in N_{[1+b,2)}(\X)$ and, by \cref{lem:agentHierarchy}, agent $i$ values her bundle strictly less than $2$ \textit{in all those iterations}. 
{As we noted in the proof of \cref{lem:One_b}, however, the cardinality of a bundle of size $2$ can only increase in Step~\ref{3Vstep_AddGood} and such a bundle was allocated to an agent who was a source of the doubly enhanced or the \enhplus  at the beginning of the corresponding iteration. The fact that the presence of $i$ prevents any future variant of bundle $X'$ from ever belonging to a source of these graphs in any iteration until $\X$ is reached contradicts the fact that $|X|\geq 3$.}

Suppose now that $|X'|\geq 3$. {Recall that} a bundle $\tilde{X}$ {such that $|\tilde{X}| \ge 3$} is updated to another bundle with cardinality at least $3$ only in Steps \ref{3Vstep_GdeCycle} -- \ref{3Vstep_AddGoodInKMax}.
In these steps, only a single good is added to $\tilde{X}$ from the pool. Assume, without loss of generality, that $g_1\in X''$ and $g_2 \in \PP(X'')$, i.e., $g_1$ was already {in $X''$ right before the $r$-th iteration}, and $g_2$ was added to $X''$ from the pool during {this iteration}. We claim that another step should have been executed {during the $r$-th iteration} instead of Steps~\ref{3Vstep_GdeCycle} -- \ref{3Vstep_AddGoodInKMax}, since $g_2$ was in the pool and $v_i(g_2)=1$.

If $i \in N_{1,(0,2/3)}{(\X'')}$, {then \cref{prop:F1} is not satisfied for $\X''$ and,} by \cref{lem:prop(f1)}, one of the conditions of Steps \ref{3Vstep_GrCycle} -- \ref{3Vstep_2a} of the \PPAplusplus algorithm would have been triggered instead. 
Consider now the case where $i\in N_{[2/3,1)}(\X'')\cup N_{(1,1+b)}(\X'')$. In that case, $|X_i''|\geq 2$ and $2/3\leq v_i(X_i'')<1+b$. 
The second case of \cref{lem:maxValue} implies that there exists $g'\in X_i''$ such that $v_i(g')\geq b$ {(or else we would have $v_i(X_i'') < 1/2$).  
This, however, means that $v_i(X_i''\setminus \{g'\}) < 1 = v_i(g_2)$ and the condition of Step~\ref{3Vstep_swapManyTo2} should have been triggered instead,} with agent $i$ swapping $X_i''\setminus \{g'\}$ for $g_2$.  
Finally, suppose that $i \in  N_{[1+b,2)}(\X'')$. If there exists $g'\in X_i''$ such that $v_i(g')= 1$, then the condition of Step \ref{3Vstep_swapManyTo2} should have been triggered instead, {using the exact same argument as in the previous case}. If no such $g'$ exists, it must be the case that there are three goods in $X_i''$ that $i$ values at $b$, in which case the condition of Step~\ref{3Vstep_swap3bToab} should have been triggered instead.
Note that, by \cref{lem:agentHierarchy}, $i\notin N_{[2,\infty)}(\X'')$.

In every case, the assumption that there exist $g_1,g_2\in X$ such that $v_i(g_1)=v_i(g_2)=1$ leads to a contradiction, so the statement follows. \renewcommand\qedsymbol{{\small $\boxdot$}}
\end{proof}

 \begin{claim}\label{claim:no_four_b_goods}
If $|X|\geq 5$, there are no four goods $g_1,g_2,g_3,g_4 \in X$ such that $v_i(g_j)\geq b$, for $j \in \{1,2,3,4\}$.
\end{claim}

\begin{proof}[Proof of \cref{claim:no_four_b_goods}]
{Again we follow the template of the proof of \cref{lem:One_b}.}
Suppose, towards a contradiction, that there exist $g_1,g_2, g_3, g_4\in X$ each of which is valued at least $b$ by agent $i$. 
Consider the first iteration, {say the $r$-th iteration overall,} of the while loop after which
$g_1,g_2, g_3, g_4$ were part of the same bundle $\tilde{X}$ and none of them is removed from $\tilde{X}$ until the allocation $\X$ is reached. 
Let $\X'$ and $\X''$ be the partial allocation right after and right before, respectively, the execution of this $r$-th iteration, and similarly let $X'$ and $X''$ be the aforementioned set $\tilde{X}$, after and prior, respectively, the execution of the $r$-th iteration. So, {$\{g_1,g_2,g_3,g_4\} \subseteq X'$ but $\{g_1,g_2,g_3,g_4\} \not\subseteq X''$}. 
	
A bundle $\tilde{X}$ {such that $|\tilde{X}| \ge 4$} is updated to another bundle with cardinality at least $4$ only in Steps \ref{3Vstep_GdeCycle} -- \ref{3Vstep_AddGoodInKMax}. In these steps, only a single good is added to $\tilde{X}$ from the pool. Assume, without loss of generality, that $g_1, g_2, g_3\in X''$ and {$g_4 \in \PP(X'')$, i.e., $g_1, g_2, g_3$ were already in $X''$ right before the $r$-th iteration, and $g_4$ was added to $X''$ from the pool during this iteration}. 

If Step \ref{3Vstep_GdeCycle} is executed during the $r$-th iteration, then $X'=\{g_1, g_2, g_3, g_4\}$. The cardinality can only increase to at least $5$ in Step \ref{3Vstep_AddGood}. For this to happen, this bundle should be assigned to a source of the doubly enhanced envy graph {after the $r$-th iteration but before the allocation $\X$ is reached. 
This would contradict the fact that agent $i$, by \cref{lem:agentHierarchy}, values her bundle strictly less than any future variant of bundle $X'$ \textit{in all those iterations}.}

If one of Steps \ref{3Vstep_AddGood} and \ref{3Vstep_AddGoodInKMax} is executed during the $r$-th iteration, then $X''$ is assigned to some source $s$ of the doubly enhanced graph $\Gde(\X'')$. If $v_i(X_i'')<1+b$, then $i$ would envy $s$, since $v_i(X'')\geq 3b>1+b$, and $s$ would not be a source in $\Gde(\X'')$. 
So, by \cref{lem:agentHierarchy}, it can only be the case that $i \in  N_{[1+b,2)}(\X'')$.
Recall that $g_4\in \PP(\X'')$ and $v_i(g_4)\geq b$. If $v_i(g_4)=1$, either the condition of Step~\ref{3Vstep_swapManyTo2} {should have been triggered because of $i$}, if $i$ values one good in her bundle at $1$, or the condition of Step~\ref{3Vstep_swap3bToab} {should have been triggered because of $i$}, because the only other case for $i\in  N_{[1+b,2)}(\X'')$ is that she values {at least} three goods in her bundle at $b$, {by \cref{cor:Gde_part_a}}. Suppose now that $v_i(g_4)=b$. If agent $i$ values any of $g_1,g_2,g_3\in X''$ at $1$ {or if $X''$ contains more than three goods worth $b$ to $i$}, then $i$ strictly envies $s$, so $s$ cannot be a source of $\Gde(\X'')$. If $v_i(g_1)=v_i(g_2)=v_i(g_3)=b$ for $g_1,g_2,g_3\in X''$, there should be an edge $(i,s)$ (it may as well be a self loop) in  $\Gde(\X'')$, by definition, and---again---$s$ cannot be a source in $\Gde(\X'')$.

In every case, the assumption that there exist $g_1,g_2, g_3, g_4\in X$ such that $\min_{j\in[4]} v_i(g_j) \ge b$ leads to a contradiction, so the statement follows. \renewcommand\qedsymbol{{\small $\boxdot$}}
 \end{proof}
With those two claims we complete the proof of the lemma. 
\end{proof}

\begin{lemma}
\label{lem:3VcriticalGoods}
Let $\X$ be a partial allocation {such that in the next} iteration of the while loop of the \PPAplusplus algorithm Steps~\ref{3Vstep_GrCycle}-\ref{3Vstep_swap3bToab} {will} not run. For any agent $i \in N$ who has a critical good $g$ in $\X$, it holds that $|X_i|=1$, $v_i(X_i)=1$, and $v_i(g)=b$. 
\end{lemma}

\begin{proof}
Consider any agent $i \in N$ and any good $g \in \PP(\X)$. Regardless of whether $|X_i|=1$ or $|X_i|>1$, it must be the case that $v_i(X_i) \geq v_i(g)$, as otherwise one of the conditions of Step~\ref{3Vstep_swap1To1}, Step~\ref{3Vstep_swap2To1}, {or Step~\ref{3Vstep_swapManyTo2} would be triggered}. Given this, {by \cref{lem:Value1Cardinality1}, \cref{lem:not=ValuesInP}, and the condition of Step~\ref{3Vstep_swap1To2} not being triggered}, it follows that $v_i(g) < v_i(X_i)$.

Consider first the case where $|X_i|=1$. In this case, {$v_i(X_i)> v_i(g) \ge c$} implies either that $v_i(X_i)=b$ or that $v_i(X_i)=1$. In the former case, we obviously have that $v_i(g)=c$, and by \cref{lem:valueOfc}, we have that $v_i(X_i)>2 \cdot v_i(g)$, since $b > 1/2$. Thus, $g$ cannot be a critical good for agent $i$. In the latter case, $g$ obviously cannot be critical for agent $i$ if $v_i(g)=c$ and, therefore, we conclude that $v_i(g)=b$.

Now consider the case where $|X_i|>1$. We will show that it is not possible for agent $i$ to have a critical good. First, if $v_i(g) =c$, then we trivially have that $v_i(X_i)\geq 2 \cdot v_i(g)$ since $|X_i|>1$, and $g$ cannot be critical for $i$. Therefore we may assume that $v_i(g) \geq b$. Let $g^{*} \in \arg\max_{g \in X_i}v_i(g)$ be a good of maximum value for agent $i$ in her bundle $X_i$. We argue that there exists $g \in X_i \setminus\{g^*\}$ such that $v_i(g)\geq b$. If this was not the case, by applying \cref{lem:maxValue} on the set $X_i \setminus\{g^*\}$, we would have $v_i(X_i \setminus\{g^*\}) < \frac12 < b$, and therefore the agent should swap $X_i \setminus\{g^*\}$ with $g$ from the pool (since $v_i(g) \geq b$) in Step~\ref{3Vstep_swapManyTo2}, which by assumption will not run. From this, we infer that there are \textit{at least two} goods $g^*$, $\tilde{g}$ in $X_i$ such that {$v_i(g^*) \ge v_i(\tilde{g}) \geq b$  (recall that $g^{*} \in \arg\max_{g \in X_i}v_i(g)$)}. If $v_i(g)=b$, this implies that $g$ cannot be critical for agent $i$, so it must be the case that $v_i(g)=1$. In turn, since Step~\ref{3Vstep_swapManyTo2} will not run, we have that {$v_i(X_i\setminus\{g^*\})\geq 1$}. We consider three cases:
\begin{itemize}[leftmargin=18pt]
    \item[-] $v_i({g^*})=1$. In that case we have $v_i(X_i) \geq 2 \geq 2\cdot v_i(g)$, and hence $g$ cannot be critical for agent $i$.\smallskip
    \item[-] {$v_i(g^*)=b$ and} $v_i(X_i \setminus \{{g^*}\}) \geq 3b$. Again, in that case we have $v_i(X_i) \geq 2 \geq 2\cdot v_i(g)$, and hence $g$ cannot be critical for agent $i$.\smallskip
    \item[-] {$v_i(g^*)=b$,} there exist exactly three goods $g_1$, $g_2$, $g_3$ in $X_i$ {such that $v_i(g_1) = v_i(g_2) = v_i(g_3) = b$}, and $v_i(g')=c$ for {any good $g' \in X_i\setminus \{g_1, g_2, g_3\}$}. This is precisely the condition of Step~\ref{3Vstep_swap3bToab}, contradicting the assumption of the lemma that this step will not run. 
\end{itemize}
{Overall, the only case that does not lead to a contradiction is when $|X_i|=1$, $v_i(X_i)=1$, and $v_i(g)=b$.}
\end{proof}

The following lemma establishes that \cref{prop:a,prop:b} are satisfied at the end of each iteration of the while loop of the algorithm, unless the algorithm has terminated with a complete $2/3$-EFX allocation, in which case we are done.

\begin{lemma}\label{lem:3Vpropertiesab}
Let $\X$ be a partial allocation obtained at the end of any iteration of the while loop of the \PPAplusplus algorithm. 
\begin{itemize}[leftmargin=20pt]
    \item[-] If $\PP(\X) \neq \emptyset$, then $\X$ satisfies \cref{prop:a,prop:b}.
    \item[-] If $\PP(\X) = \emptyset$, then $\X$ is a complete $2/3$-EFX allocation.
\end{itemize}
\end{lemma}

\begin{proof}
We will prove the statement by induction on the number of iterations of the while loop of the \PPAplusplus algorithm. Before the first iteration \cref{prop:a,prop:b} are satisfied, and if $\PP(\X) = \emptyset$, then $\X$ is a complete $2/3$-EFX allocation. Now consider any arbitrary iteration of the while loop, and let $\X$ be the partial allocation of the algorithm right before the iteration, and $\X'$ be the partial allocation right after the iteration. We first state and prove the following two useful claims which address situations that may occur in several steps of the algorithm.

\begin{claim}
	\label{cl:(a)(b)2Goods}
	Let $X$ be one of the allocated bundles in $\X$, {such that $|X|\geq 2$ \emph{or} $X$ is allocated to a source of $\Gr(\X)$. Assume that for agents $i,j \in N$ we have the following in $\X'$:}
	\begin{itemize}[leftmargin=20pt]
		\item[-] $X_i' = \{g,g'\}$ for {some} good $g\in \PP(\X)$ and {some} good $g'\in X$,  \emph{and}\smallskip
		\item[-] $v_j(X_j')\geq \frac 32 v_j(X_j)$, if $|X_j|>|X_j'|=1$, and $v_j(X_j')\geq v_j(X_j)$, otherwise.
	\end{itemize}
	Then, if $|X_j'|=1$, $j$ is EFX towards $i$, and if $|X_j'|>1$, $j$ is $2/3$-EFX towards $i$.
\end{claim}

\begin{proof}[Proof of \cref{cl:(a)(b)2Goods}]
{Given how $X_i'$ is formed, we note that during the iteration of the while loop of \PPAplusplus that starts with $\X$ and outputs $\X'$ Steps~\ref{3Vstep_swap1To1} and \ref{3Vstep_swap2To1} do not run.}
Consider first the case in which $|X_j|=1$. Since Step~\ref{3Vstep_swap1To1} cannot run, for every good $\tilde{g} \in \PP(\X)$ we have $v_j(\tilde{g}) \leq v_j({X_j})$ and, hence, also $v_j(g) \leq v_j(X_j)$. 
{We also claim that $v_j(X_j) \geq {v_j(g')}$. Indeed, this follows directly if $|X| = 1$ (i.e., $X = \{g'\}$) and $X$ is allocated to a source of $\Gr(\X)$; if $|X|\geq 2$, it follows by \cref{prop:a} (which holds by the induction hypothesis) and the fact that $X$ is one of the bundles in $\X$.
These,} together with the fact that $v_j(X_j') \geq v_j(X_j)$, imply that in $\X'$ agent $j$ is EFX towards agent $i$. 

Next, consider the case in which $|X_j|>1$. Since {Step~\ref{3Vstep_swap2To1}} cannot run, for every good $\tilde{g} \in \PP(\X)$ we have that $v_j(\tilde{g}) \leq \frac{3}{2}v_j(X_j)$, which implies $v_j(g) \leq \frac{3}{2}v_j(X_j)$. 
{Further, we claim that $v_j(X_j) \geq \frac{2}{3}{v_j(g')}$. Indeed, if $|X| = 1$, this follows directly by the definition of $\Gr(\X)$, whereas, if $|X|\geq 2$, it follows by \cref{prop:b} (which holds by the induction hypothesis) and the fact that $X$ is one of the bundles in $\X$.
That is, agent $j$ holding her old bundle $X_j$ would be $2/3$-EFX towards agent $i$ holding her new bundle $X_i'$.}
If $|X_j'|=1$, by the statement of the claim we have that $v_j(X_j') \geq \frac32 v_j(X_j)$, which implies that in $\X'$ agent $j$ is EFX towards agent $i$. If $|X_j'|>1$, on the other hand, we have that $v_j(X_j') \geq v_j(X_j)$, which implies that agent $j$ is $2/3$-EFX towards agent $i$. \renewcommand\qedsymbol{{\small $\boxdot$}}
\end{proof}

\begin{claim}
\label{cl:(a)(b)ManyGoods}
{Assume that during the iteration of the while loop of \PPAplusplus that starts with $\X$ and outputs $\X'$ one of Steps~\ref{3Vstep_GeCycle} -- \ref{3Vstep_AddGoodInKMax} is executed.}
Let $s$ be any source in the \regplus {$\Grplus(\X)$}, and {assume that for agents $i,j \in N$ we have the following in $\X'$:}
\begin{itemize}[leftmargin=20pt]
    \item[-] {$X_i' = Y \cup \{g\}$, for some $Y \subseteq X_s$ and some}  good $g\in \PP(\X)$, and $|X_i'|\leq k_{\max}$, \emph{and}\smallskip
    \item[-] $v_j(X_j')\geq \frac 32 v_j(X_j)$, if $|X_j|>|X_j'|=1$, and $v_j(X_j')\geq v_j(X_j)$, otherwise.
\end{itemize}
Then, if $|X_j'|=1$, $j$ is EFX towards $i$, and if $|X_j'|>1$, $j$ is $\frac 23$-EFX towards $i$.
\end{claim}

\begin{proof}[Proof of \cref{cl:(a)(b)ManyGoods}]
If $|X_j|=1$, {we claim that $v_j(X_j)\geq b$. Indeed, since $\PP(\X)\neq\emptyset$ and Steps~\ref{3Vstep_GrCycle} -- \ref{3Vstep_swap3bToab} do not run, \cref{lem:3VcriticalGoods} applies. If $v_j(X_j)=c$, then every $g'\in\PP(\X)$ would be critical for $j$, so \cref{lem:3VcriticalGoods} would give $v_j(X_j)=1$; a contradiction. Hence, $v_j(X_j)\in\{b,1\}$.}

If $v_j(X_j) = b$, {we similarly claim that $v_j(g)=c$; indeed, by \cref{lem:3VcriticalGoods}, $j$ has no critical good since $v_j(X_j)\neq1$, so $v_j(g)\le\frac12 v_j(X_j)=\tfrac b2<b$ for every good in $\PP(\X)$. Hence, $v_j(g)=c$.} Since $s$ is a source of {$\Grplus(\X)$ and $|X_j|=1$}, it holds that  $v_j(X_j)\geq v_j(X_s) = v_j(X_s\cup \{g\}) - c \geq v_j(X_i' \setminus \{g'\})$, for any $g'\in X_i'$, and so agent $j$ is EFX  towards agent $i$. 

Suppose now that $v_j(X_j)=1$. {Arguing like above,} it holds that $v_j(g)\leq b$, {or $g$ would be critical for agent $j$ with $v_j(g) = 1$, contradicting \cref{lem:3VcriticalGoods}.} Since $s$ is a source of {$\Grplus(\X)$, if $|X_s| = 1$ then $v_j(X_s) < 1$, by the definition of $\Grplus(\X)$, whereas if $|X_s| \ge 2$, then, by \cref{prop:F2}, we still have  $v_j(g') < 1$ for any $g' \in X_s$. In any case,} we have that $v_j(g')\leq b$, for any $g'\in X_s$. 
If $|X_i'|\leq 2$,  it holds that $v_j(X_i'\setminus\{g'\})\leq b<v_j(X_j)$, for any $g'\in X_i'$. If $|X_i'|\geq 3$, by \cref{lem:maxValue,lem:One_b}, we have that $v_j(X_j) = 1 > 5/6 > v_j(X_i')$.  Therefore, in this case too, agent $j$ is EFX towards agent $i$.\smallskip

Suppose now that $|X_j|>1$. By \cref{lem:3VcriticalGoods}, it holds that $v_j(g) \leq \frac 12 v_j(X_j)$, and since $s$ is a source of $\Grplus(\X)$, it {either holds that $v_j(X_j)\geq v_j(X_s)$ and $|X_s|\ge 2$, or $v_j(X_j)\geq \frac{2}{3} v_j(X_s)$ and $|X_s| = 1$. In either case,} we obtain that 
{agent $j$ holding her old bundle $X_j$ would be $2/3$-EFX towards agent $i$ holding her new bundle $X_i'$.}
We conclude the proof by observing that
\begin{itemize}[leftmargin=20pt]
\item[-] If $|X_j'|=1$, by the statement of the claim, $v_j(X_j')\geq \frac 32 v_j(X_j)$, and so $j$ is EFX towards $i$. 
\item[-] If $|X_j'|>1$, by the statement of the claim, $v_j(X_j')\geq v_j(X_j)$, and so $j$ is $\frac 23$-EFX towards $i$.  \renewcommand\qedsymbol{{\small $\boxdot$}}\qedhere
\end{itemize}
\end{proof}

Now we can continue with the proof of \cref{lem:3Vpropertiesab}. We consider cases depending on which step of the algorithm ran in this iteration of the while loop.
\smallskip

\noindent \textit{If Step~\ref{3Vstep_GrCycle} or Step~\ref{3Vstep_equalVinPool} is executed:} For each agent $i \in N$ we have $v_i(X'_i)\geq v_i(X_i)$. Therefore, by the induction hypothesis, it suffices to consider the case of an agent $j$ that goes from $|X_j|=2$ (and for which \cref{prop:b} was satisfied) to $|X'_j|=1$ (where now \cref{prop:a} must be satisfied).  
Since in these steps we are considering the \regplus $\Grplus(\X)$, this only happens when $v_j(X'_j) \geq \frac32 v_j(X_j)$. This establishes \cref{prop:a}, and thus both properties are satisfied. \smallskip

\noindent \textit{If Step~\ref{3Vstep_swap1To1} is executed:} Consider the agent $i \in N$ whose bundle gets updated in this step, and notice that $v_i(X_i')\geq v_i(X_i)$. {Thus, \cref{prop:a} (and, trivially, \cref{prop:b}) for agent $i$, which is satisfied in $\X$ by the induction hypothesis, is maintained for $i$ in $\X'$ as well.}
Since $|X'_i|=1$, every other agent $j \in N$ is EFX towards agent $i$, and so \cref{prop:a,prop:b} in $\X'$ are satisfied for those agents as well.\smallskip

\noindent \textit{If Step~\ref{3Vstep_swap2To1} is executed:} Consider the agent $i \in N$ whose bundle gets updated in this step. By \cref{prop:b} which holds by the induction hypothesis, agent $i$ was $2/3$-EFX towards any other agent $j \in N$ in $\X$. Since $v_i(X'_i) \geq \frac{3}{2}v_i(X_i)$, agent $i$ is EFX towards any other agent in $\X'$, which implies \cref{prop:a,prop:b} for agent $i$. Since $|X'_i|=1$, every other agent $j \in N$ is EFX towards agent $i$, and so \cref{prop:a,prop:b}  in $\X'$ are satisfied for those agents as well. \smallskip

\noindent \textit{If Step~\ref{3Vstep_2a} is executed:}  Consider the agents in $R(\X)$ whose bundles change from $\X$ to $\X'$, and notice that for each $i \in R(\X)$, we have $|X_i|=1$ and $|X'_i|=2$, as well as $v_i(X'_i) > v_i(X_i)$. Consider any other agent $j \in N$. Assume first that $j \notin R(\X)$, {thus $X_j' = X_j$}. By the definition of $R(\X)$, 
{either $v_j(X_j)\geq \frac{2}{3} v_j(X_i)$ and $|X_j|\ge 2$, or $v_j(X_j)\geq v_j(X_i)$ and $|X_j| = 1$.}
(otherwise there would be an edge $(j,i)$ in $\Grplus(\X)$). Furthermore, 
\begin{itemize}[leftmargin=20pt]
\item[-] since Step~\ref{3Vstep_swap1To1} did not run, if $|X_j|=1$, it holds that $v_j(X_j) \geq v_j(g)$ for all $g \in \PP(\X)$,
\item[-] since Step~\ref{3Vstep_swap2To1} did not run, if $|X_j| \ge 2$, it holds that $v_j(X_j) \geq \frac{2}{3} v_j(g)$ for all $g \in \PP(\X)$. 
\end{itemize}
This means that
\begin{itemize}[leftmargin=20pt]
    \item[-] if {$|X_j'|=|X_j|=1$, then $v_j(X'_j) = v_j(X_j) \geq v_j(X'_i \setminus \{g\})$ for any of the two goods $g\in X'_i$}, and hence agent $j$ is EFX towards agent $i$ in $\X'$, 
    \item[-] if {$|X_j'|=|X_j|\ge 2$, then $v_j(X'_j) = v_j(X_j)  \geq \frac23 v_j(X'_i \setminus \{g\})$ for any of the two goods $g\in X'_i$}, and hence agent $j$ is $2/3$-EFX towards agent $i$ in $\X'$. 
\end{itemize}
Since $v_i(X'_i) > v_i(X_i)$, agent $i$ is EFX towards agent $j$ in $\X'$ by the induction hypothesis, which ensured that she was EFX towards $j$ in $\X$. {Overall, \cref{prop:a,prop:b} are satisfied for such agents $i$ and $j$.}

Now assume that $j \in R(X)$. 
Since both $i$ and $j$ are in the set $R(\X)$ it follows that $v_j(X_j) \geq v_j(X_i) = v_j(g_i)$ where $g_i$ is {the unique good in} $X_i$. Letting {$g_i'$ be the good that is added to $X_i$} by the \SRR routine in Step~\ref{3Vstep_2a} {(i.e., $\{g_i'\} = X'_i\setminus\{g_i\}$)}, we also have that $v_j(X_j) \geq v_j(g_i')$, as otherwise the condition of Step~\ref{3Vstep_swap1To1} would have been triggered. Since $v_j(X'_j) > v_j(X_j)$, agent $j$ is EFX towards agent $i$ in $\X'$, and \cref{prop:a,prop:b} are satisfied for agent $j$. 

Finally, for any pair of agents $i,j \notin R(\X)$, their bundles have not changed from $\X$ to $\X'$, so when considered against each other, \cref{prop:a,prop:b} hold by the induction hypothesis. \smallskip

\noindent \textit{If Step~\ref{3Vstep_swap1To2} is executed:} Consider the agent $i$ whose bundle gets updated in this step and notice that $|X_i|=1$. Since $v_i(X_i') > v_i(X_i)$ and \cref{prop:a} is satisfied for the agent in $\X$ by the induction hypothesis,  \cref{prop:a,prop:b} are satisfied for agent $i$ in $\X'$ as well. Consider any other agent $j \in N$. Looking at agent $i$'s bundle, $X'_i$, we observe that it contains $2$ goods which were either both received from $\PP(\X)$ {(first part of Step~\ref{3Vstep_swap1To2})}, or one came from $\PP(\X)$ and one from $X_i$, $i$ being a source of {$\Gr(\X)$} {(second part of Step~\ref{3Vstep_swap1To2}).}
As we argued in Step~\ref{3Vstep_2a} above, 
\begin{itemize}[leftmargin=20pt]
	\item[-] since Step~\ref{3Vstep_swap1To1} did not run, if $|X_j|=1$, it holds that $v_j(X_j) \geq v_j(g)$ for all $g \in \PP(\X)$,
	\item[-] since Step~\ref{3Vstep_swap2To1} did not run, if $|X_j| \ge 2$, it holds that $v_j(X_j) \geq \frac{2}{3} v_j(g)$ for all $g \in \PP(\X)$. 
\end{itemize}
{Assuming that $X'_i$ is built in the first part of Step~\ref{3Vstep_swap1To2}, we directly have that} agent $j$ is EFX towards agent $i$ if $|X_j|=1$ and she is $2/3$-EFX towards $i$ if $|X_j|\ge 2$. 

{When $X'_i$ is built in the second part of Step~\ref{3Vstep_swap1To2}, \cref{cl:(a)(b)2Goods} can be applied, implying that again agent $j$ is EFX towards agent $i$ if $|X_j|=1$ and she is $2/3$-EFX towards $i$ if $|X_j|\ge 2$.}

{Invoking the induction hypothesis for all other pairs of agents,} we conclude that \cref{prop:a,prop:b} are satisfied for everyone.
\smallskip

\noindent \textit{If Step~\ref{3Vstep_swapManyTo2} is executed:} Consider the agent $i \in N$ whose bundle gets updated in this step and notice that $v_i(X_i') > v_i(X_i)$. Since $|X_i|\geq 2$ and $|X'_i|\geq 2$, the induction hypothesis implies that \cref{prop:b} is satisfied for agent $i$ (and so is \cref{prop:a} vacuously). Looking at agent $i$'s new bundle $X'_i$, we observe that it contains $2$ goods, one that she received from $\PP(\X)$ and one from $X_i$, {which has cardinality at least $2$}. By \cref{cl:(a)(b)2Goods}, \cref{prop:a,prop:b} are  satisfied for any other agent $j \in N \setminus\{i\}$, so overall both properties are satisfied.
\smallskip

\noindent \textit{If Step~\ref{3Vstep_swap3bToab} is executed:} Consider the agent $i \in N$ whose bundle gets updated in this step. By definition of the sets in the hierarchy (see \cref{def:hierarchy}), $i\in {N_{[1+b,2)}(\X) \cap N_{[1+b,2)}(\X')}$. In particular, we have that $v_i(X'_i)\geq 1+b$. Consider any other agent $j \in N$. By \cref{lem:Three_b}, for any $g\in X_j$, it holds that 
\[\frac 23\, v_i(X_j\setminus\{g\})<\frac 23 \Big(b+\frac 53 \Big)=b-\frac b3 +\frac {10}9< b - \frac 16 + \frac {10}9  = b +\frac{17}{18}<1+b,\] 
and therefore \cref{prop:b} is still satisfied for agent $i$. {Since $X'_i$ contains $2$ goods, one from $\PP(\X)$ and one from $X_i$, which has cardinality at least $2$}, by Claim~\ref{cl:(a)(b)2Goods}, \cref{prop:a,prop:b} are satisfied for any agent $j \in N \setminus \{i\}$, so overall both properties are satisfied. \smallskip

\noindent \textit{If Step~\ref{3Vstep_GeCycle} or Step~\ref{3Vstep_GdeCycle} is executed:}  Consider {the only} bundle $X$ that changes {its content (as opposed to just its owner)} during this step. {In $\X$, $X$ is the bundle $X_s$ of some source $s$ of $\Grplus(\X)$.} Its updated version in $\X'$ is $X'_i$, {where agent $i$ is as described in the formal description of Steps~\ref{3Vstep_GeCycle} and \ref{3Vstep_GdeCycle}. It is easy to see that in Step~\ref{3Vstep_GeCycle} agent $i$ goes from having value at most $1$ to having value $2b$ and in Step \ref{3Vstep_GdeCycle} agent $i$ goes from having value less than $2$ to having value $4b$. Therefore, in both cases, \cref{prop:b} is satisfied from $i$'s perspective.} 
Looking at $X'_i$, we also observe that it consists of one good $g^* \in \PP(\X)$ and {some goods} of $X_s$. By \cref{cl:(a)(b)ManyGoods}, \cref{prop:a,prop:b} are  satisfied in $\X'$ for any agent $j$ such that $v_j(X_j')\geq \frac 32 v_j(X_j)$, if $|X_j|>|X_j'|=1$, and $v_j(X_j')\geq v_j(X_j)$, otherwise. Taking the definitions of the corresponding envy graphs into consideration, that would include any agent $j\in N\setminus\{i\}$ whose value was not reduced during this iteration.

{Now consider any agent $j \in N$ such that $v_j(X'_j) < v_j(X_j)$. Note that this implies that either $j\in N_{1,(0,2/3)}(\X')$, or $j\in N_{[1+b,2)}(\X')$ (the latter happens only if during Step~\ref{3Vstep_GdeCycle}). Suppose first that $j\in N_{1,(0,2/3)}(\X')$. This can only happen because an edge specific to $\Geplus(\X)$ and, therefore, $|X'_j|\ge 2$ and $v_j(X'_j) > b$. For any other agent $\ell \in N \setminus \{j\}$, by \cref{lem:prop(f2)} (if $|X_{\ell}|\le 2$) or by \cref{lem:One_b,lem:maxValue} (if $|X_{\ell}|\ge 3$), we have that $v_j(X_{\ell}'\setminus\{g\})<\frac 23$, for any $g\in X_{\ell}'$. Therefore, we have $v_j(X_j')>b> \frac 49 >\frac 23 v_j(X'_{\ell}\setminus\{g\})$ for any $g\in X_{\ell}'$. That is, \cref{prop:b} is satisfied from $j$'s perspective. }

Finally, consider the case where $j\in N_{[1+b,2)}(\X')$. {This can only happen because of an edge specific to $\Gde(\X)$. Then, however, to show that \cref{prop:b} is satisfied from $j$'s perspective we may repeat---verbatim---the argument for the case of Step~\ref{3Vstep_swap3bToab} above. }

{Invoking the induction hypothesis for all other pairs of agents, we conclude that \cref{prop:a,prop:b} are satisfied for everyone.}

\noindent \textit{If Step~\ref{3Vstep_AddGood} is executed:} In this step, only some bundle $X_i$ is modified {from $\X$ to} $\X'$, by setting $X_i' = X_i \cup \{g\}$ for some good $g \in \PP(\X)$. Of course, $v_i(X_i')\geq v_i(X_i)$, which, combined with the fact that $|X_i| \leq |X'_i|$ implies that \cref{prop:a,prop:b} are satisfied {from $i$'s perspective} by the induction hypothesis. Notice also that agent $i$ is a source of $\Gde(\X)$ and therefore also a source of $\Grplus(\X)$. By \cref{cl:(a)(b)ManyGoods}, \cref{prop:a,prop:b} are  satisfied in $\X'$ for any agent $j \in N\setminus\{i\}$. \smallskip

\noindent \textit{If Step~\ref{3Vstep_AddGoodInKMax} is executed:} 
Consider the agent $i \in N$ whose bundle changes its content (as opposed to just its owner) during this step. Here the algorithm assumes the existence 
of a source $s$ of $\Gde(\X)$. Then this $s$ dictates who agent $i$ should be, as $i$ is the endpoint of a path $\Pi$ that starts at $s$ in $\Gde(\X)$. Such an $s$ is guaranteed to exist by the time the condition of Step~\ref{3Vstep_AddGoodInKMax} is considered, as it is explained in the final arguments of the proof of \cref{lem:3V_noCritical} below. 
We have that $X_i' = (X_s \setminus \{g'\}) \cup \{g\}$, where $g'$ is $i$'s less preferred good from $s$'s bundle in $\X$ and $g \in \PP(\X)$ is $i$'s critical good in $\X$.
Therefore, for any $j \in N \setminus \{i\}$, \cref{cl:(a)(b)ManyGoods} can be applied, implying that agent $j$ is EFX towards agent $i$, if $|X_j|=1$, and she is $2/3$-EFX towards $i$, if $|X_j|\ge 2$. The fact that \cref{prop:a,prop:b} are fully satisfied for everyone in $N \setminus \{i\}$ follows then by repeating the very simple argument used for Steps~\ref{3Vstep_GrCycle} and \ref{3Vstep_equalVinPool} with respect to the agents who received another bundle along the path $\Pi$.

What is still left to show is that \cref{prop:b} is satisfied from $i$'s perspective in $\X'$ (recall that $|X'_i| = k_{\max}\ge 3$). By the condition of Step~\ref{3Vstep_AddGoodInKMax}, agent $i$ has a critical good $g_i \in \PP(\X)$, thus, \cref{lem:3VcriticalGoods} implies that $v_i(X_i)=1$ and, by \cref{lem:Value1Cardinality1}, we also have $|X_i|=1$.

Consider some agent $j \in N\setminus\{i\}$ with bundle $X_j'$ (where, possibly, $X'_j = X_i$). First of all, if $X'_j = X_i$, then by the above discussion, agent $i$ is trivially $2/3$-EFX towards $j$ in $\X'$. So, we may assume that $X_j' = X_{\ell}$ for some $\ell \in N\setminus\{i\}$. By \cref{lem:prop(f2)} (if $|X_{\ell}|\le 2$) or by \cref{lem:One_b,lem:maxValue} (if $|X_{\ell}|\ge 3$), we have that $v_i(X_{\ell}\setminus\{\tilde{g}\})<\frac 23$ for any $\tilde{g}\in X_{\ell}$. Therefore, we have $v_i(X_i') > v_i(g)  = b > \frac 49 >\frac 23 v_i(X_{\ell}\setminus\{\tilde{g}\}) = \frac 23 v_i(X'_{j}\setminus\{\tilde{g}\})$, for any $\tilde{g}\in X'_{j}$. This establishes that \cref{prop:b} is satisfied from $i$'s perspective in $\X'$. 

Overall, invoking the induction hypothesis for all other pairs of agents, we conclude that \cref{prop:a,prop:b} are satisfied for everyone in $N$.\smallskip

\noindent \textit{If Step~\ref{3Vstep_Break} is executed:} In this case $\X'=\X$ and the algorithm terminates; \cref{prop:a,prop:b} are satisfied trivially by the induction hypothesis. 
\end{proof}

\paragraph{Concluding the proof of \cref{lem:3V_noCritical}} We argue that as long as there exists some agent $i \in N$ with a critical good, some step of the \PPAplusplus algorithm other than Step~\ref{3Vstep_Break} will be triggered. Suppose, towards a contradiction, that in the last iteration of the while loop of the algorithm Step~\ref{3Vstep_Break} ran. Since Steps~\ref{3Vstep_GrCycle} and ~\ref{3Vstep_2a} did not run, by \cref{lem:helpToStep9} the \regplus $\Grplus(\X)$ has a source vertex $s$. Consider the \enhplus $\Geplus(\X)$, and assume first that it does not have a source; this means that $s$ must be contained in some cycle of $\Geplus(\X)$. Consider any such cycle $C$, and let $(i,s)$ be an edge in $C$. Since $s$ is a source of $\Grplus(\X)$, $(i,s)$ is not an edge of $\Grplus(\X)$. From this it follows that $(i,s)$ is not a gray edge, which in turn implies that $C$ is jaundiced. That means that the condition of Step~\ref{3Vstep_GeCycle} of the \PPAplusplus algorithm should have been triggered instead, a contradiction. 

Now assume that there is a source in $\Geplus(\X)$, and consider the doubly enhanced graph $\Gde(\X)$. Similarly to before, if there is no source in $\Gde(\X)$, there would be a cycle that is jaundiced and the condition of Step~\ref{3Vstep_GdeCycle} should have been triggered instead; a contradiction. 
So, we may assume that $\Gde(\X)$ has at least one source. If some source $s$ of $\Gde(\X)$ has $|X_s|<k_{\max}$, then, since $i$ has a critical good, the condition of Step~\ref{3Vstep_AddGood} is satisfied; a contradiction. {Otherwise, every source of $\Gde(\X)$ has $|X_s|=k_{\max}$, and to reach a contradiction through Step~\ref{3Vstep_AddGoodInKMax} we need a source $s$ with $|X_s|=k_{\max}$ \emph{and} a path from $s$ to $i$. Towards that, we show that $i$ is reachable from a source of $\Gde(\X)$. First, let $A$ be the set of vertices from which $i$ is reachable in $\Gde(\X)$ (including $i$); then every edge entering $A$ originates in $A$. Next we contract the strongly connected components of $\Gde(\X)[A]$ into a DAG. Let $S$ be the vertex set of a connected component that is a source of that DAG; by the property of $A$ we mentioned above, $S$ has no incoming edges from $N\setminus S$ in $\Gde(\X)$ at all. 
If $|S|\ge2$, then $S$ is a nontrivial strongly connected component, so every vertex of $S$ lies on a cycle of $\Gde(\X)$. Each such cycle is unjaundiced since Steps~\ref{3Vstep_GrCycle}, \ref{3Vstep_GeCycle}, and \ref{3Vstep_GdeCycle} did not run, hence consists of gray edges and every agent of $S$ has value $1$. As no edge enters $S$ from outside, $S$ is then a nonempty \iso set, contradicting $R(\X)=\emptyset$ (as Step~\ref{3Vstep_2a} did not run). Therefore $|S|=1$, i.e., $S=\{s\}$ for a single vertex $s$, which is a source of $\Gde(\X)$  from which $i$ is reachable; as every source has $|X_s|=k_{\max}$, the condition of Step~\ref{3Vstep_AddGoodInKMax} is satisfied; a contradiction.}
 \qed
 \smallskip

\subsubsection{The Proof of \cref{lem:3V_polytime}} 
\label{subsec:3PA++_polytime}
We will show that the potential function we introduced in \cref{subsec:potential} lexicographically increases at each round of \PPAplusplus that performs any step, except possibly Steps \ref{3Vstep_GrCycle}, \ref{3Vstep_swap3bToab}, and \ref{3Vstep_AddGood}. Applying Steps \ref{3Vstep_GrCycle}, \ref{3Vstep_swap3bToab}, or \ref{3Vstep_AddGood} may leave the potential function unaffected, yet we will argue that those three steps may be applied sequentially for a polynomially bounded number of steps. Recall the potential function: 
\[\potential(\X)= (|N_{[2,\infty)}(\X)|, |N_{[1+b,2)}(\X)|, |N_{(1,1+b)}(\X)|, |N_{[2/3,1)}(\X)|, SW_{(0,1+b)}(\X))\]
By \cref{lem:agentHierarchy}, while running \PPAplusplus none of the  first four coordinates of the potential function ever decreases. Hence, $\potential(\X)$ increases {whenever} one of its terms increases {(even if the last coordinate occasionally decreases)}.

\begin{itemize}[leftmargin=*, itemsep=3pt]
    \item[-] If Step \ref{3Vstep_swap1To1} or \ref{3Vstep_swap2To1} is applied, one agent improves her value by getting a single good, so $SW_{(0,1+b)}(\X)$ strictly improves.
    
    \item[-] If Step \ref{3Vstep_equalVinPool} is applied, no agent decreases their value. Let $(i',i)$ be the {\em last} non-gray edge in the path. We first argue that $|X_{i}|=1$. This is true, because either agent $i$ is not the last vertex on the path and the next edge is gray, or $i$ is the last vertex on the path and we know that $v_{i}(X_{i})=v_{i}(g)$ for some $g\in \PP(\X)$. By Lemma~\ref{lem:Value1Cardinality1}, the latter implies $|X_{i}|=1$. 
    In any case, agent $i'$ strictly improves her value by getting a single good, so $SW_{(0,1+b)}(\X)$ strictly improves.
    
    \item[-] If Step \ref{3Vstep_2a} is applied, no agent in $R(\X)$  belonged to any of the sets $N_{[2,\infty)}(\X)$, $N_{[1+b,2)}(\X)$, $N_{(1,1+b)}(\X)$, since they had value $1$ for their bundles. However, after applying Step \ref{3Vstep_2a}, {at least one agent from $R(\X)$} belongs to one of these sets ({typically, this happens with \textit{all} the agents} unless the pool empties and the algorithm terminates), so the cardinality of at least one of them increases.
    
    \item[-] If Step \ref{3Vstep_swap1To2} is applied no value is decreased and {the agent $i$ whose bundle is modified strictly increases hers. If $i$ received two goods from the pool (first part of Step~\ref{3Vstep_swap1To2}), then, by Lemma~\ref{lem:prop(f1)}, agent $i$ values her new bundle} by at most $2b$, so $SW_{(0,1+b)}(\X)$ strictly improves. 
    {If $i$ received one good from the pool (second part of Step~\ref{3Vstep_swap1To2}), then either $SW_{(0,1+b)}(\X)$ strictly improves or $|N_{[1+b,2)}(\X)|$ increases; in any case the second part of Step~\ref{3Vstep_swap1To2} completes the allocation and the algorithm terminates after the next iteration. }

    \item[-] If Step \ref{3Vstep_swapManyTo2} is applied, let $i$ be the agent that modifies her bundle. If her value for her updated bundle is less than $1+b$, $SW_{(0,1+b)}(\X)$ strictly improves.  The {only other possible cases about $i$'s value for her updated bundle} are that it is either $1+b$, in which case $|N_{[1+b,2)}(\X)|$ increases, or it is $2$, in which case $|N_{[2,\infty)}(\X)|$ increases.
    
    \item[-] If Step \ref{3Vstep_GeCycle} or \ref{3Vstep_GdeCycle} is applied, some agents may reduce their value, but in Step \ref{3Vstep_GeCycle} $|N_{(1,1+b)}(\X)|$ increases (an agent with value {at most} $1$ improves her value to $2b$) and in Step \ref{3Vstep_GdeCycle} $|N_{[2,\infty)}(\X)|$ increases (an agent with value less than $2$ improves her value to $4b$). 

    \item[-] If Step \ref{3Vstep_AddGoodInKMax} is applied, $|N_{[2/3,1)}(\X)|$ increases, since an agent with value $1$ {(by \cref{lem:3VcriticalGoods})} receives a set for which she has value $b+(k_{\max}-1)c\geq \frac 23$.

    \item[-] If Step \ref{3Vstep_swap3bToab} is applied, the potential function remains unchanged because only a single agent from $N_{[1+b,2)}(\X)$ changes her bundle and she remains in $N_{[1+b,2)}(\X)$.

    \item[-] If Step \ref{3Vstep_GrCycle} or \ref{3Vstep_AddGood} is applied, no agent reduces their value, whereas the potential may or may not improve. 
\end{itemize}

Next we argue that it is not possible to consecutively  apply  Steps \ref{3Vstep_GrCycle}, \ref{3Vstep_swap3bToab}, and \ref{3Vstep_AddGood}  $2mn^2$ times without increasing $\potential(\X)$. First note that Step \ref{3Vstep_GrCycle} can be executed consecutively less than $n^2$ times, since each time the number of edges in $\Gr(\X)$ decreases. If only Steps~\ref{3Vstep_GrCycle} and \ref{3Vstep_swap3bToab} are executed, Step \ref{3Vstep_swap3bToab} can be executed at most $m$ times, since the cardinality of the pool always increases, so both steps can be consecutively executed less than $mn^2$ times. 
In the following claim we prove that if {at any point} Steps \ref{3Vstep_GrCycle} and \ref{3Vstep_swap3bToab} cannot be applied, then Step \ref{3Vstep_swap3bToab} may only be applied after the potential function has increased.

 \begin{claim}\label{claim:step8}
     If for some allocation $\X$ the conditions of Steps \ref{3Vstep_GrCycle} and \ref{3Vstep_swap3bToab} of \PPAplusplus {are false}, Step \ref{3Vstep_swap3bToab} may be only applied after the increment of the potential function.
 \end{claim}
 
 \begin{proof}[Proof of \cref{claim:step8}]
If any of the Steps \ref{3Vstep_swap1To1}-\ref{3Vstep_swapManyTo2},  
\ref{3Vstep_GeCycle}, \ref{3Vstep_GdeCycle}, or \ref{3Vstep_AddGoodInKMax}  
could be applied on $\X$, 
the potential function would increase and the claim follows. So, assume that {none of Steps \ref{3Vstep_GrCycle}-\ref{3Vstep_swap3bToab} can} be applied on $\X$. This means that for any agent $i \in N_{[1+b,2)}(\X)$ and any good $g \in \PP(\X)$, $v_i(g)<1$, otherwise either Step \ref{3Vstep_swapManyTo2} or Step \ref{3Vstep_swap3bToab} would apply. Suppose, towards a contradiction, that it is possible that only Steps \ref{3Vstep_AddGood} and \ref{3Vstep_GrCycle} are executed until we get an allocation $\X'$ where Step \ref{3Vstep_swap3bToab} can be applied again; otherwise the potential function would increase and the claim would follow. The set $N_{[1+b,2)}(\X)$ remains the same during all those rounds, otherwise the potential function would increase. During Steps \ref{3Vstep_GrCycle} and \ref{3Vstep_AddGood}, $\PP(\X)$ may only stay the same or shrink, so for any agent $i \in N_{[1+b,2)}(\X')$ and any good $g \in \PP(\X') \subseteq \PP(\X)$, $v_i(g)<1$, and therefore Step \ref{3Vstep_swap3bToab} cannot be applied on $\X'$. This contradicts our assumption {about $\X'$}, and the claim follows. 
\end{proof}

As an immediate consequence of the claim, we get that after less than $mn^2$ steps (during which we have assumed that the potential remains the same), {neither Step~\ref{3Vstep_GrCycle} nor Step~\ref{3Vstep_swap3bToab} can} be applied anymore. Then, however, only Steps~\ref{3Vstep_AddGood} and \ref{3Vstep_GrCycle} may be executed without the improving $\potential(\X)$. 
Similarly to {our argument about Steps~\ref{3Vstep_GrCycle} and \ref{3Vstep_swap3bToab}} above, if only Steps \ref{3Vstep_GrCycle} and \ref{3Vstep_AddGood} are executed, Step \ref{3Vstep_AddGood} can be executed at most $m$ {consecutive} times, since it always shrinks the pool, so both steps can be executed sequentially less than $mn^2$ times. Overall, all three steps can be consecutively executed, without the execution of any other step, for less than $2mn^2$ times.

Equivalently, within $2mn^2$ steps, $\potential(\X)$ will improve. Already this suffices to show that \PPAplusplus terminates. This follows from the simple fact that eventually either the allocation will be complete, or the potential function will become equal to $(n, 0, 0, 0, 0)$, i.e., all agents will enter the set $N_{[2,\infty)}(\X)$. Once an agent gets into $N_{[2,\infty)}(\X)$ (i.e., \cref{level1}), 
by \cref{lem:agentHierarchy}, she never leaves this set (as she cannot fall in hierarchy). It is easy to check that if  $N_{[2,\infty)}(\X) = N$, none of the  Steps \ref{3Vstep_GrCycle}-\ref{3Vstep_AddGoodInKMax} can be applied and the algorithm terminates. Nevertheless, we are going to show something stronger, namely that \PPAplusplus terminates in polynomially many iterations. As each iteration only involves checking the conditions of the different steps in polynomial time and one of the steps running, again in polynomial time, this would prove the statement. 

So, every $2mn^2$ steps at most, $\potential(\X)$  improves. We are going to distinguish two types of improvements: \textit{welfare improvements}, where only the restricted social welfare $SW_{(0,1+b)}(\X)$ increases, and \textit{ascend improvements}, where {at least one} of the other coordinates of $\potential(\X)$ increases, i.e., at least one agent ascends in the hierarchy. 
Each agent can trigger at most $4$ ascend improvements, as she never falls in hierarchy. So there are at most $4n$ ascend improvements in total. In fact, by the above discussion, after at most $4n$ ascend improvements, the \PPAplusplus algorithm must terminate. 

Finally, we are going to bound the number of welfare improvements between two ascend improvements. By the definition of an ascend improvement, between any two of them, the set of agents contributing to $SW_{(0,1+b)}(\X)$, i.e., agents at \cref{level3,level4,level5}, remains unchanged. 
As long as none of the Steps \ref{3Vstep_2a}, \ref{3Vstep_GeCycle}, \ref{3Vstep_GdeCycle}, or \ref{3Vstep_AddGoodInKMax} is executed (all of which force an ascend improvement) every agent at \cref{level3,level4,level5} never decreases her value. Moreover, every $2mn^2$ steps at most, at least one of those agents strictly increases her value through Steps 
\ref{3Vstep_swap1To1}-\ref{3Vstep_equalVinPool} or \ref{3Vstep_swapManyTo2}, \textit{always} by getting a bundle of cardinality $1$ or $2$, as one can see by simple inspection of the different steps. Every time this happens for a specific agent $i\in N_{(1,1+b)}(\X) \cup N_{[2/3,1)}(\X) \cup N_{1,(0,2/3)}(\X)$, she must have a distinct value for her bundle. Since there are only {7 possible values for} bundles of cardinality at most $2$ that allow $i$ to be in that set {(namely, $c$, $2c$, $b$, $b+c$, $2b$, $1$, and $1+c$)}, $i$ can trigger at most $6$ welfare improvements between two ascend improvements. So, there are at most $6n$ welfare improvements, in total, between two ascend improvements.

Therefore, every $14mn^3$ steps or less, at least one ascend improvement occurs. If the \PPAplusplus algorithm does not terminate earlier, within $56mn^4$ iterations at most, all possible ascend improvements will happen, guaranteeing termination.  \qed

\subsubsection{Dealing with Zero Values}\label{app:zero_c}
In the pseudocode of \TVA the case where $c = 0$ is handled by replacing these zeros by a sufficiently small number. Here we will show that this is indeed a valid approach that reduces the problem to the $c > 0$ case we already dealt with. Moreover, as it is evident by the definition of $c'$ in the statement of the next lemma, a sufficient value for $\varepsilon$ in \TVA can be found in $O(m^2)$ time.

\begin{lemma}\label{lem:zero_c}
    Consider a $3$-value instance $(N, (v_i)_{i\in N}, M)$ such that $a=1 > b> c=0$ and construct the instance $(N, (\tilde{v}_i)_{i\in N}, M)$ by replacing any value that is $0$ by 
    \[c' = \frac{1}{3m}\,\min\{|k+\ell b| : k,\ell \in \{-3m,\ldots,  2m\} \textnormal{\ \ and\ \ } k+\ell b \neq 0\}\,.\]
    Then any $2/3$-EFX allocation for $(N, (\tilde{v}_i)_{i\in N}, M)$, is also a $2/3$-EFX allocation for $(N, ({v}_i)_{i\in N}, M)$.
\end{lemma}

\begin{proof}
First, note that for any two $S, T\subseteq M$ and any $i\in N$, we have that $|2v_i(S) - 3v_i(T)|$ is either $0$ or takes a value in the set $\{|k+\ell b| : k,\ell \in \{-3m,\ldots,  2m\} \textnormal{\ \ and\ \ } k+\ell b \neq 0\}$. In particular, if $|2v_i(S) - 3v_i(T)|\neq 0$, then $|2v_i(S) - 3v_i(T)| \ge \min\{|k+\ell b| : k,\ell \in \{-3m,\ldots,  2m\} \textnormal{\ \ and\ \ } k+\ell b \neq 0\} =  3mc'$.

Let $\X$ be a $2/3$-EFX allocation for $(N, (\tilde{v}_i)_{i\in N}, M)$. {First, observe that if $m\le n$ then it is necessary that $|X_i|\le 1$, for all $i\in N$, and $\X$ is trivially a $2/3$-EFX allocation for $(N, ({v}_i)_{i\in N}, M)$ as well. So, for the rest of the proof, we may assume that $m>n$. Given that all values are positive in $(N, (\tilde{v}_i)_{i\in N}, M)$, we must have $|X_i|\ge 1$, for all $i\in N$. This implies that $|X_i|< m$, for all $i\in N$.} 

Consider any two agents $i, j \in N$. Of course, we have  $\tilde{v}_i(X_i) \ge \frac23 \max_{g\in X_j}\tilde{v}_i(X_j \setminus \{g\})$. By the definition of $\tilde{v}_i$, we have $\tilde{v}_i(X_i)\le {v}_i(X_i) + |X_i|c'$ and $\max_{g\in X_j}\tilde{v}_i(X_j \setminus \{g\}) \ge \max_{g\in X_j}{v}_i(X_j \setminus \{g\})$.
That is 
\[{v}_i(X_i) + |X_i|c' \ge \frac23\max_{g\in X_j}{v}_i(X_j \setminus \{g\}),  \text{\ \ or equivalently,\ \ } 2 {v}_i(Y_j) - 3{v}_i(X_i) \le 3|X_i|c'\,,\]
where $Y_j\subseteq X_j$ such that $v_i(Y_j) = \max_{g\in X_j}{v}_i(X_j \setminus \{g\})$.

Suppose that
${v}_i(X_i)  < \frac23\max_{g\in X_j}{v}_i(X_j \setminus \{g\})$. Then, $2 {v}_i(Y_j) - 3{v}_i(X_i) > 0$ and, by the above discussion, it must be $2 {v}_i(Y_j) - 3{v}_i(X_i) \ge 3mc' > 3|X_i|c' \ge 2 {v}_i(Y_j) - 3{v}_i(X_i)$, which is a contradiction. Hence, $\X$ is a $2/3$-EFX allocation for $(N, ({v}_i)_{i\in N}, M)$ as well. 
\end{proof}

\section{Conclusion and Discussion}\label{sec:conclusion}

In this work, we have pushed the frontier of existence and computation results for approximate EFX allocations
in three natural settings for which we show that $2/3$-EFX allocations can be computed in polynomial time. Interestingly, these settings strictly generalize the settings where the state of the art achieves exact EFX allocations.
The algorithmic framework that we develop, via the introduction of the \PPA algorithm and its variants, has the potential to be a powerful tool in obtaining approximate EFX allocations, possibly in polynomial time. While the existence of EFX allocations remains the holy grail of this line of work, one of the most sensible next steps would be to aim for the existence of $2/3$-EFX allocations for (general) additive valuation functions. The case of $3$-value instances already reveals many of the intricacies one has to deal with in this endeavor. 

A related direction is to explore the applicability of our approach in different variants of the problem, as it has already been done by \citet{PrakashMN25} and \citet{kakatelis2026almost}, or even in other fairness notions, e.g., as in the recent work of \citet{filos2026approximate}.

Finally, a  natural question here is whether our framework can be adjusted for the case where the resources are indivisible chores, i.e., each item has non-positive value for everyone.  
The existence of $4$-EFX allocations (under the natural modification of \cref{def:ef-efx}) for agents with additive valuation functions was shown in the breakthrough work of \citet{GargMQ25}, but polynomial computation is only known for a constant number of agents. 
Currently, the best known approximation factors---achievable in polynomial time---are $2$ for a constant number of  agents \citep{GargM26} and $O(n^2)$ in general \citep{ZhouW24}. Although the mechanics of criticality are quite different for chores, we believe that the high level approach of this work, i.e., the recursive improvement of an initial partial allocation and the use of multiple envy graphs, could be a good starting point for obtaining efficient constant-factor approximation algorithms, at least in restricted cases.

\section*{Acknowledgments}

Georgios Amanatidis was supported by the NWO Veni project No.~VI.Veni.192.153. Aris Filos-Ratsikas was supported by the UK Engineering and Physical Sciences Research Council (EPSRC) grant EP/Y003624/1. Alkmini Sgouritsa was supported by
the H.F.R.I.~Project No.~15635 which is implemented in the framework of H.F.R.I call ``Basic research Financing (Horizontal support of all Sciences)'' under the National Recovery and Resilience Plan ``Greece 2.0'' funded by the European Union--NextGenerationEU. Georgios Amanatidis and Alkmini Sgouritsa were also supported by
the project MIS 5154714 of the National Recovery and Resilience Plan ``Greece 2.0'' funded by the European Union--NextGenerationEU.

\bibliographystyle{plainnat}
\bibliography{efx}

\end{document}